\documentclass[a4paper,reqno]{amsart} 
\usepackage{amsfonts}
\usepackage{amssymb}
\usepackage{amsthm}
\usepackage{amsmath}
\usepackage{mathrsfs}
\usepackage{dsfont}
\usepackage{stmaryrd}
\usepackage[english]{babel}
\usepackage[left=3.5cm,right=3.5cm,top=3.5cm,bottom=3cm,headsep=0.7cm]{geometry}

\usepackage{pgf,tikz}
\usetikzlibrary{arrows}
\usetikzlibrary{intersections}
\usepackage[scriptsize,bf]{caption}
\usepackage{floatrow}
\usepackage{enumerate}
\usepackage{enumitem}
\usepackage{hyperref}
\usepackage[normalem]{ulem}

\theoremstyle{plain}
\begingroup
\theoremstyle{plain}
\newtheorem{theorem}{Theorem}[section]

\newtheorem{proposition}[theorem]{Proposition}
\newtheorem{lemma}[theorem]{Lemma}
\theoremstyle{definition}
\newtheorem{definition}[theorem]{Definition}
\theoremstyle{remark}
\newtheorem{remark}[theorem]{Remark}

\endgroup

\theoremstyle{definition}
\theoremstyle{remark}

\numberwithin{equation}{section}

\newcommand{\RR}{\mathbb{R}}
\newcommand{\NN}{\mathbb{N}}
\newcommand{\ZZ}{\mathbb{Z}}
\newcommand{\CC}{\mathbb{C}}

\renewcommand{\SS}{\mathbb{S}}
\newcommand{\WW}{\mathbb{W}}
\renewcommand{\S}{\mathcal{S}}

\renewcommand{\H}{\mathcal{H}}
\newcommand{\D}{\mathcal{D}}
\newcommand{\DD}{\mathrm{D}}

\renewcommand{\L}{\mathcal{L}}
\newcommand{\M}{\mathcal{M}}

\mathchardef\emptyset="001F
\renewcommand{\d}[1]{\, \mathrm{d} #1}
\newcommand{\de}{\partial}
\newcommand{\e}{\varepsilon}
\renewcommand{\tilde}{\widetilde}
\newcommand{\x}{{\times}}
\newcommand{\ol}{\overline}
\newcommand{\ul}{\underline}
\newcommand{\sm}{\setminus}
\newcommand{\dist}{{\rm dist}}
\newcommand{\argmin}{\mathrm{argmin}}

\newcommand{\cart}{\mathrm{cart}}

\newcommand{\weak}{\rightharpoonup}
\newcommand{\wstar}{\stackrel{*}\rightharpoonup}
\renewcommand{\flat}{\stackrel{\mathrm{f}}\to}
\newcommand{\mres}{\mathbin{\vrule height 1.6ex depth 0pt width 0.13ex\vrule height 0.13ex depth 0pt width 1.3ex}}
\newcommand{\integral}[3]{\int_{#1} \! #2 #3}
\newcommand{\nn}{{\langle i , j \rangle}}
\newcommand{\PC}{\mathcal{PC}}
\newcommand{\geo}{\mathrm{d}_{\SS^1}}
\newcommand{\compact}{\subset\subset}
\renewcommand{\Supset}{\supset\supset}
\newcommand{\w}{{\! \wedge \!}}
\newcommand{\supp}{\mathrm{supp}}

\author{Marco Cicalese}
\address[Marco Cicalese]{Technische Universit\"at M\"unchen, Munich, Germany}
\email{cicalese@ma.tum.de}

\author{Gianluca Orlando}
\address[Gianluca Orlando]{Technische Universit\"at M\"unchen, Munich, Germany}
\email{orlando@ma.tum.de}

\author{Matthias Ruf}
\address[Matthias Ruf]{Ecole polytechnique f\'ed\'erale de Lausanne, Lausanne, Switzerland}
\email{matthias.ruf@epfl.ch}

\title[The $N$-clock model: Variational analysis for fast and slow $N$]{The $N$-clock model: Variational analysis \\ for fast and slow divergence rates of $N$}

\begin{document}

	\begin{abstract}
		We study a nearest neighbors ferromagnetic spin system on the square lattice in which the spin field is constrained to take values in a discretization of the unit circle consisting of $N$ equi-spaced vectors, also known as $N$-clock model. We find a fast rate of divergence of $N$ with respect to the lattice spacing for which the $N$-clock model has the same discrete-to-continuum variational limit of the $XY$ model, in particular concentrating energy on topological defects of dimension 0. We prove the existence of a slow rate of divergence of $N$ at which the coarse-grain limit does not detect topological defects, but it is instead a $BV$-total variation. Finally, the two different types of limit behaviors are coupled in a critical regime for $N$, whose analysis requires the aid of Cartesian currents. 
	\end{abstract}

	\maketitle
	
	
	\vspace{1em}
	

	\setcounter{tocdepth}{1}
	\tableofcontents

	\section{Introduction}
	
	The emergence of phase transitions mediated by the formation of topological singularities has been proposed in the pioneering works on the ferromagnetic $XY$ model~\cite{Ber, Kos, Kos-Tho}. The latter describes a system of $\SS^1$ vectors sitting on a square lattice, in which only the nearest neighbors interact in a ferromagnetic way. If the spin field is allowed to attain only finitely many, say $N$, equi-spaced values on $\SS^1$ (as in the $N$-clock model) this topological concentration is ruled out. Instead, the phase transitions are characterized by a typical domain structure as in Ising systems. These two different behaviors lead to the natural question whether the $N$-clock model approximates the $XY$ model as $N \to +\infty$. Fr\"ohlich and Spencer give a positive answer to this question, showing in~\cite{Fro-Spe} that the $N$-clock model (for $N$ large enough) presents phase transitions mediated by the formation and interaction of topological singularities.
	
	The results in this paper and in~\cite{CicOrlRuf,CicOrlRuf3} concern a related problem regarding the behavior of low-energy states of the two systems in the discrete-to-continuum variational analysis as the lattice spacing vanishes and $N$ diverges simultaneously. With the help of fine concepts in geometric measure theory and in the theory of cartesian currents, these results show to which extent the coarse-grain limit of the $N$-clock model resembles the one of the $XY$ model obtained in~\cite{Ali-Cic, Ali-DL-Gar-Pon}. To state precisely the results, we set the mathematical framework for the problem.

	We consider a bounded, open set with Lipschitz boundary $\Omega \subset \RR^2$. 
	Given a small parameter $\e > 0$, we consider the square lattice $\e \ZZ^2$ and we define $\Omega_\e := \Omega \cap \e \ZZ^2$. The $XY$ energy (referred to its minimum) is defined on spin fields $u \colon \Omega_\e \to \SS^1$ by 
	\begin{equation*}
		XY_\e(u) = \frac{1}{2} \sum_\nn \e^2 | u(\e i) - u(\e j) |^2 ,
	\end{equation*}
	where the sum is taken over ordered pairs of nearest neighbors $\nn$, i.e., $(i, j) \in \ZZ^2 \x \ZZ^2$ such that $|i - j| = 1$ and $\e i, \e j \in \Omega_\e$. We consider the additional parameter $N_\e \in \NN$ or, equivalently, $\theta_\e := \tfrac{2\pi}{N_\e}$, and we set
	\[
	\S_\e := \{\exp(\iota k \theta_\e) \colon  k = 0, \dots, N_\e-1\} \, , 
	\] 
	where $\iota$ is the imaginary unit. The admissible spin fields we consider here are only those taking values in the discrete set $\S_\e$. We study the energy defined for every $u \colon \Omega_\e \to \SS^1$ by
	\begin{equation} \label{eq:def of E}
		E_\e(u) := \begin{cases}
			XY_\e(u) & \text{if } u \colon \Omega_\e \to \S_\e \, , \\
			+ \infty & \text{otherwise}.
		\end{cases}
	\end{equation}
	We are interested in the behavior of low-energy states $u_\e$ such that $E_\e(u_\e) \leq C \kappa_\e$ as $\e \to 0$. To this end, given $\theta_\e \to 0$, we find the relevant scaling $\kappa_\e$ and we study the $\Gamma$-limit of $\frac{1}{\kappa_\e} E_\e$. The limit strongly depends on the rate of convergence $\theta_\e \to 0$ and can be characterized by interfacial-type singularities~\cite{Caf-DLL, Ali-Bra-Cic, Bra-Pia, Cic-Sol, Ali-Gel, Bra-Cic, Bra-Kre, Cic-For-Orl, Bac-Cic-Kre-Orl-surf} (see also~\cite{Bra-Cic-Ruf,Ali-Cic-Ruf}) or vortex-like singularities~\cite{Pon, Ali-Cic, Ali-Cic-Pon, Ali-DL-Gar-Pon, Bra-Cic-Sol, Bad-Cic-DL-Pon, Ali-Bra-Cic-DL-Pia}, possibly coexisting. In this paper we are interested in the following regimes: $\e |\log \e | \ll \theta_\e$, $\theta_\e \sim \e |\log \e|$, and $\theta_\e \ll \e$. The case $\e \ll \theta_\e \ll \e |\log \e|$ has been covered in~\cite{CicOrlRuf}.
	
	To understand how the limit is affected by the choice of $\theta_\e \to 0$, we start by considering the following example. Let $\Omega = B_{1/2}(0)$ be the ball of radius $\tfrac{1}{2}$ centered at 0, let $v_1 = \exp(\iota \varphi_1)$, $v_2 = \exp(\iota \varphi_2) \in \SS^1$, and let us define for $x = (x_1,x_2)$
	\begin{equation} \label{eq:pure jump function}
		u(x) := \begin{cases}
			v_1 & \text{if } x_1 \leq 0 \, ,\\
			v_2 & \text{if } x_1 > 0 \, .
		\end{cases}
	\end{equation}
	For $\e i = (\e i_1,\e i_2) \in \Omega_\e$ we define
	\begin{equation} \label{eq:transition for jump}
		u_\e(\e i) := \begin{cases}
			v_1 & \text{if } \e i_1  \leq 0 \, ,\\
			\exp\Big(\iota\big( (\varphi_1 - \varphi_2) \big(1 - \tfrac{\e i_1}{\eta_\e} \big)  + \varphi_2 \big)\Big) & \text{if } 0 < \e i_1 \leq \eta_\e \, , \\
			v_2 & \text{if } \e i_1 > \eta_\e \, .
		\end{cases}
	\end{equation}
	If $u_\e$ satisfies the constraint $u_\e \colon \Omega_\e \to \S_\e$, then $|\varphi_1-\varphi_2|\frac{\e}{\eta_\e} \sim \theta_\e$, i.e., $\eta_\e \sim |\varphi_1-\varphi_2|\frac{\e}{\theta_\e}$, see Figure~\ref{fig:transition}. As a result 
	\begin{equation*}
		\frac{1}{\kappa_\e} E_\e(u_\e) \sim \Big( 1 - \cos\Big(\frac{\e}{\eta_\e} (\varphi_1 - \varphi_2)\Big) \Big) \frac{\eta_\e}{\kappa_\e}  \sim \big( 1 - \cos(\theta_\e ) \big) \frac{\e}{\theta_\e \kappa_\e} |\varphi_1 - \varphi_2| \sim \frac{\e \theta_\e}{\kappa_\e}  |\varphi_1 - \varphi_2|\, .
	\end{equation*}
	This suggests that the nontrivial scaling $\kappa_\e = \e \theta_\e$ leads to a finite energy proportional to~$|\varphi_1-\varphi_2|$. The construction can be optimized by choosing the angles $\varphi_1$ and $\varphi_2$ in such a way that $|\varphi_1 - \varphi_2|$ equals the geodesic distance on $\SS^1$ between $v_1$ and $v_2$, namely $\geo(v_1,v_2)$. 
    \begin{figure}[H]
        \includegraphics{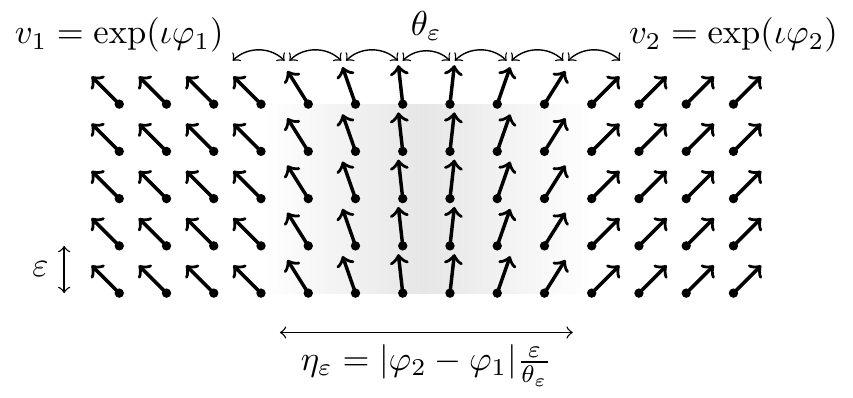}

		\caption{Construction which shows that $\frac{1}{\e \theta_\e}E_\e$ approximates the geodesic distance between the two values $v_1$ and $v_2$ of a pure-jump function. During the transition between $v_1$ and $v_2$ in the strip of size $\eta_\e = |\varphi_1 - \varphi_2| \frac{\e}{\theta_\e}$ the minimal angle between two adjacent vectors is $\theta_\e$.}
		
		\label{fig:transition}
	\end{figure}
	The fact that $\geo(v_1,v_2)$ is the total variation (in the sense of~\cite[Formula~(2.11)]{Amb}) of the $\SS^1$-valued pure-jump function $u$ defined in~\eqref{eq:pure jump function} suggests that at the scaling $\kappa_\e = \e \theta_\e$ the $\Gamma$-limit of $\tfrac{1}{\e \theta_\e} E_\e$ might be finite on the class $BV(\Omega;\SS^1)$ of $\SS^1$-valued functions of bounded variation. This is confirmed by the following theorem, for which we introduce some notation, cf.\ \cite{Amb-Fus-Pal}. Given a function $u \in BV(\Omega;\SS^1)$, its distributional derivative $\DD u$ can be decomposed as $\DD u = \nabla u \L^2 + \DD^{(c)} u + (u^+-u^-) \otimes \nu_u \H^1 \mres J_u$, where $\nabla u$ denotes the approximate gradient, $\L^2$ is the Lebesgue measure in $\RR^2$, $\DD^{(c)} u$ is the Cantor part of $\DD u$, $\H^1$ is the 1-dimensional Hausdorff measure, $J_u$ is the $\H^1$-countably rectifiable jump set of $u$ oriented by the normal $\nu_u$, and $u^+$ and $u^-$ are the traces of $u$ on $J_u$. By $|\, \cdot \,|_1$ we denote the 1-norm on vectors and by $| \, \cdot \, |_{2,1}$ the anisotropic norm on matrices given by the sum of the Euclidean norms of the columns.

	\begin{theorem}[Regime $\e |\log \e| \ll \theta_\e \ll 1$] \label{thm:e log smaller theta}
		Assume that $\e |\log \e| \ll \theta_\e \ll 1$. Then the following results hold:
		\begin{itemize}
			\item[i)] (Compactness) Let $u_\e \colon \Omega_\e \to \S_\e$ be such that $\frac{1}{\e \theta_\e} E_\e(u_\e) \leq C$.
			Then there exists a subsequence (not relabeled) and a function $u \in BV(\Omega;\SS^1)$ such that $u_\e \to u$ in $L^1(\Omega;\RR^2)$. 
			\item[ii)] ($\Gamma$-liminf inequality) Assume that $u_\e \colon \Omega_\e \to \S_\e$ and $u \in BV(\Omega;\SS^1)$ satisfy $u_\e \to u$ in $L^1(\Omega;\RR^2)$. Then 
			\begin{equation*}
				\integral{\Omega}{|\nabla u|_{2,1} }{\d x} + |\DD^{(c)} u|_{2,1}(\Omega) + \integral{J_u}{\geo(u^-,u^+)|\nu_{u}|_1}{\d \H^1} \leq \liminf_{\e \to 0} \frac{1}{\e \theta_\e} E_\e(u_\e) \, .
			\end{equation*} 
			\item[iii)] ($\Gamma$-limsup inequality) Let $u \in BV(\Omega;\SS^1)$. Then there exists a sequence $u_\e \colon \Omega_\e \to \S_\e$ such that $u_\e \to u$ in $L^1(\Omega;\RR^2)$ and
			\begin{equation*}
				\limsup_{\e \to 0} \frac{1}{\e \theta_\e} E_\e(u_\e) \leq \integral{\Omega}{|\nabla u|_{2,1} }{\d x} + |\DD^{(c)} u|_{2,1}(\Omega) + \integral{J_u}{\geo(u^-,u^+)|\nu_{u}|_1}{\d \H^1}.
			\end{equation*}  
		\end{itemize}
	\end{theorem}
	
    \begin{figure}[H]
        \includegraphics{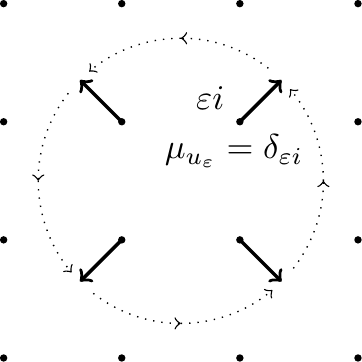}
		% \begin{tikzpicture}[scale=0.6]
			
		% 	\def\d{{sqrt(2)/2}};
		% 	\def\b{{sqrt(2)+1}};
		% 	\def\c{{-sqrt(2)-1}};

		% 	\foreach \i in {-3,-1,1,3} 
		% 	\foreach \j in {-3,-1,1,3}
		% 	\draw[fill=black] (\i,\j) circle(0.05cm);
			
		% 	\draw[->, line width=1pt] (1,1) --++(\d,\d);
		% 	\draw[->, line width=1pt] (-1,1) --++(-\d,\d);
		% 	\draw[->, line width=1pt] (-1,-1) --++(-\d,-\d);
		% 	\draw[->, line width=1pt] (1,-1) --++(\d,-\d);
			
		% 	\draw (1,1) node[anchor=south east] {$\varepsilon i$};
		% 	\draw (1-0.1,0.5) node {$\mu_{u_\e} = \delta_{\varepsilon i}$};

		% 	\draw[<-,dotted] (\b,0) arc(0:-40:\b);
		% 	\draw[->,dotted] (\b,0) arc(0:40:\b);
			
		% 	\draw[->,dotted] (0,\b) arc(90:130:\b);
		% 	\draw[<-,dotted] (0,\b) arc(90:50:\b);

		% 	\draw[<-,dotted] (\c,0) arc(180:140:\b);
		% 	\draw[->,dotted] (\c,0) arc(180:220:\b);
			
		% 	\draw[<-,dotted] (0,\c) arc(270:230:\b);
		% 	\draw[->,dotted] (0,\c) arc(270:310:\b);
		% \end{tikzpicture}
		\caption{Example of discrete vorticity measure equal to a Dirac delta on the point $\e i \in \e \ZZ^2$. By following a closed path on the square of the lattice with the top-right corner in $\e i$,  the spin field covers the whole $\SS^1$. The discrete vorticity measure can only have weights in $\{-1,0,1\}$.}
		
		\label{fig:discrete vorticity}
	\end{figure}
	
	The previous theorem does not hold true if $\theta_\e \lesssim \e |\log \e|$, i.e., $\frac{\theta_\e}{\e |\log \e|} \to C \in [0,+\infty)$. In this regime, an additional object plays a role, namely the discrete vorticity measures $\mu_{u_\e}$ associated to the spin field $u_\e$ (see Figure~\ref{fig:discrete vorticity} and cf.~\eqref{eq:discrete vorticity measure} for the precise definition). By~\eqref{eq:def of E}, we have 
	\begin{equation} \label{eq:no vortices allowed}
		\frac{1}{\e^2 |\log \e|} XY_\e(u_\e) =  \frac{\e \theta_\e}{\e^2 |\log \e|} \frac{1}{\e \theta_\e}  E_\e(u_\e) \sim \frac{\theta_\e}{\e |\log \e|} \sim C \, .
	\end{equation}
	The bound $\frac{1}{\e^2 |\log \e|} XY_\e(u_\e) \leq C$ yields compactness for the discrete vorticity measure $\mu_{u_\e}$. More precisely, in~\cite{Ali-Cic} it is proven that $\mu_{u_\e} \flat \mu$ up to a subsequence in the flat convergence (i.e., in the norm of the dual of Lipschitz functions with compact support, see~\eqref{eq:defflatconv}), where $\mu = \sum_{h=1}^N d_h \delta_{x_h}$, $x_h \in \Omega$, $d_h \in \ZZ$, is a measure that represents the vortex-like singularities of the spin field $u_\e$ as $\e$ goes to zero. The limit of $\frac{1}{\e \theta_\e} E_\e(u_\e)$ is, in general, strictly greater than the anisotropic total variation in $BV(\Omega;\SS^1)$ obtained in Theorem~\ref{thm:e log smaller theta}, since $u_\e$ must satisfy the topological constraint $\mu_{u_\e} \flat \mu$. To describe the limit, we associate to $u_\e$ with $\frac{1}{\e \theta_\e} E_\e(u_\e) \leq C$ the current $G_{u_\e}$ given by the extended graph in $\Omega \x \SS^1$ of its piecewise constant interpolation, see Subsection~\ref{s.discretecurrents}. In Section~\ref{sec:compactness} we prove a compactness result for $G_{u_\e}$ to deduce that $G_{u_\e} \weak T$ in the sense of currents. In Proposition~\ref{prop:bd of Gu is mu} we show that $\de G_{u_\e} = - \mu_{u_\e} \x \llbracket \SS^1 \rrbracket$, where $\llbracket \SS^1 \rrbracket$ is the current given by the integration over $\SS^1$ oriented counterclockwise. Since $\mu_{u_\e} \flat \mu$, the limit $T \in \mathrm{Adm}(\mu,u;\Omega)$ of the currents~$G_{u_\e}$ satisfies, among other properties that characterize the class $\mathrm{Adm}(\mu,u;\Omega)$ given in~\eqref{eq:def of Adm}, $\de T = -\mu \x \llbracket \SS^1 \rrbracket$. For this reason, in general, the current $T$ is different from the graph $G_u$ of the limit map $u$, which may have a boundary different from $-\mu \x \llbracket \SS^1 \rrbracket$. Nevertheless, $T$ can be represented as 
	\begin{equation} \label{eqintro:structure}
		T = G_{u} + L \x \llbracket \SS^1 \rrbracket  \, ,
	\end{equation}
	where $L$ is an integer multiplicity $1$-rectifiable current, which keeps track of the possible concentration of  $|\DD u_\e|$ on $1$-dimensional sets. In~\cite{CicOrlRuf} we have proved that the energy concentration on vortex-like singularities and the $BV$-type concentration on 1-dimensional sets occur at two separate energy scalings if $\e \ll \theta_\e \ll \e |\log \e|$. In this paper, we investigate the critical regime $\theta_\e \sim \e |\log \e|$ and prove that the two concentration effects appear simultaneously and are coupled in the limit energy by  
	\begin{equation*}
		\begin{split}
			\mathcal{J}(\mu,u;\Omega) := \inf \bigg\{ \integral{J_T}{\ell_T(x) |\nu_T(x)|_1}{\d \H^1(x)} \ : & \ T \in \mathrm{Adm}(\mu,u;\Omega) \bigg\} \, .
		\end{split}
	\end{equation*}
	Here $J_T$ is the $1$-dimensional jump-concentration set of $T$ oriented by the normal $\nu_T$, accounting for both the jump set of $u$ and the support of the concentration part $L$ in the decomposition~\eqref{eqintro:structure}. At each point $x \in J_T$, the current $T$ has a vertical part, given by a curve in $\SS^1$ which connects the traces of $u$ on the two sides of $J_T$; $\ell_T(x)$ is its length.

	\begin{theorem}[Regime $\theta_\e \sim \e |\log \e|$] \label{thm:theta equal e log}
		Assume that $\theta_\e = \e |\log \e|$. Then the following results hold:
		\begin{itemize}
			\item[i)] (Compactness) Let $u_\e \colon \Omega_\e \to \S_\e$ be such that $\frac{1}{\e^2 |\log \e|} E_\e(u_\e)\leq C$. Then there exists a measure $\mu = \sum_{h=1}^{M} d_h \delta_{x_h}$, $x_h \in \Omega$, $d_h \in \ZZ$, such that (up to a subsequence) $\mu_{u_\e} \flat \mu$ and there exists a function $u \in BV(\Omega;\SS^1)$ such that (up to a subsequence) $u_\e \to u$ in $L^1(\Omega;\RR^2)$. 
			\item[ii)] ($\Gamma$-liminf inequality) Let $u_\e \colon \Omega_\e \to \S_\e$, let  $\mu = \sum_{h=1}^{M} d_h \delta_{x_h}$, $x_h \in \Omega$, $d_h \in \ZZ$, and let  $u \in BV(\Omega;\SS^1)$. Assume that $\mu_{u_\e} \flat \mu$ and $u_\e \to u$ in $L^1(\Omega;\RR^2)$. Then 
			\begin{equation*}
				\integral{\Omega}{|\nabla u|_{2,1} }{\d x} + |\DD^{(c)} u|_{2,1}(\Omega) + \mathcal{J}(\mu,u;\Omega) + 2\pi|\mu|(\Omega) \leq \liminf_{\e \to 0}  \frac{1}{\e^2 |\log \e|} E_\e(u_\e) \, .
			\end{equation*} 
			\item[iii)] ($\Gamma$-limsup inequality) Let  $\mu = \sum_{h=1}^{M} d_h \delta_{x_h}$, $x_h \in \Omega$, $d_h \in \ZZ$ and let $u \in BV(\Omega;\SS^1)$. Then there exists a sequence $u_\e \colon \Omega_\e \to \S_\e$ such that $\mu_{u_\e} \flat \mu$, $u_\e \to u$ in $L^1(\Omega;\RR^2)$, and
			\begin{equation*}
				\limsup_{\e \to 0} \frac{1}{\e^2 |\log \e|} E_\e(u_\e)  \leq \integral{\Omega}{|\nabla u|_{2,1} }{\d x} + |\DD^{(c)} u|_{2,1}(\Omega) + \mathcal{J}(\mu,u;\Omega) + 2\pi|\mu|(\Omega) \, .
			\end{equation*}  
		\end{itemize}
	\end{theorem}
	
	Theorem~\ref{thm:e log smaller theta}, Theorem~\ref{thm:theta equal e log}, and the results proven in~\cite{CicOrlRuf} lead, in particular, to the conclusion that for $\e \ll \theta_\e$ the $N_\e$-clock model does not share the same asymptotic behavior of the $XY$ model. For the latter model, the following asymptotic expansion is known to hold true in the sense of $\Gamma$-convergence (see~\cite[Section~4]{Ali-DL-Gar-Pon} and also \cite{Bet-Bre-Hel, San-Ser-book, Alb-Bal-Orl, Ali-Pon} for the Ginzburg-Landau model) 
	\begin{equation} \label{eqintro:development of XY}
		\frac{1}{\e^2} XY_\e(u_\e) \sim 2 \pi M |\log \e| + \WW(\mu) + M \gamma   \, ,
	\end{equation}
	where $\WW$ is a Coulomb-type interaction potential referred to as renormalized energy, while~$\gamma$ is the core energy carried by each vortex, cf.~\eqref{eq:renormalized energy as limit} and \eqref{eq:def of core} for the precise definitions. The next theorem shows that the $E_\e$ energy has the same asymptotic expansion if $\theta_\e \ll \e$, finally providing a precise range for $\theta_\e$ for which the $N_\e$-clock model approximates the $XY$ model.

	\begin{theorem}[Regime $\theta_\e \ll \e$] \label{thm:theta smaller eps}
		Assume that $\theta_\e \ll \e$. Then the following results hold:
		\begin{enumerate}
			\item[i)] (Compactness) Let $M \in \NN$ and assume that  $u_\e \colon \Omega_\e \to \S_\e$  satisfies $\frac{1}{\e^2} E_\e(u_\e) - 2 \pi M |\log \e| \leq C$. Then there exists a measure $\mu = \sum_{h=1}^N d_h \delta_{x_h}$, $x_h \in \Omega$, $d_h \in \ZZ$, with $|\mu|(\Omega) \leq M$ such that (up to a subsequence) $\mu_{u_\e} \flat \mu$. Moreover, if $|\mu|(\Omega) = M$, then $|d_h|=1$.
			\item[ii)] ($\Gamma$-liminf inequality) Let $\mu = \sum_{h=1}^M d_h \delta_{x_h}$, $x_h \in \Omega$, $|d_h| = 1$,  and let $u_\e \colon \Omega_\e \to \S_\e$ be such that $\mu_{u_\e} \flat \mu$. Then
			\begin{equation*}
				\WW(\mu) + M \gamma \leq \liminf_{\e \to 0} \Big( \frac{1}{\e^2} E_\e(u_\e) - 2 \pi M |\log \e| \Big)  \, .
			\end{equation*}
			\item[iii)] ($\Gamma$-limsup inequality) Let $\mu = \sum_{h=1}^M d_h \delta_{x_h}$, $x_h \in \Omega$, $|d_h| = 1$. Then there exists $u_\e \colon \Omega_\e \to \S_\e$ such that $\mu_{u_\e} \flat \mu$ and
		\end{enumerate}
		\begin{equation*}
			\limsup_{\e \to 0} \Big( \frac{1}{\e^2} E_\e(u_\e) - 2 \pi M |\log \e| \Big)\leq \WW(\mu) + M \gamma   \, .
		\end{equation*}
	\end{theorem}
	
	We summarize all the results obtained in this paper, in~\cite{CicOrlRuf}, and in~\cite{CicOrlRuf3} in Table~\ref{table:energy behavior}.
	
	\begin{table}[H]
		\renewcommand{\arraystretch}{2}
		\scalebox{0.9}{
			\begin{tabular}{ | c | c | c | c |  }
				\hline  
				Regime & Energy bound & Limit of $\mu_{u_\e}$  & Energy behavior \\
				\hline 
				\hline
				$\theta_\e$ finite & $\frac{1}{\e} E_\e \leq C$ & not relevant & interfaces \\
				\hline
				$\e |\log \e| \ll \theta_\e$ & $\frac{1}{\e \theta_\e} E_\e \leq C$  & not relevant & $BV$ \\
				\hline
				&  & &  vortices    \\[-1.3em] 
				$\theta_\e \sim \e |\log \e|$ & $\frac{1}{\e \theta_\e} E_\e \leq C$ & $\mu_{u_\e} \flat \mu$ & + \\[-1.3em]
				& & & $BV$ + concentration \\
				\hline 
				&  & &  vortices    \\[-1.3em] 
				$\e \ll \theta_\e \ll \e |\log \e|$ & $\frac{1}{\e \theta_\e} E_\e  - 2\pi M |\log \e| \frac{\e}{\theta_\e}\leq C$ & $\mu_{u_\e} \flat \mu$ & + \\[-1.3em]
				& & & $BV$ + concentration \\
				\hline
				$\theta_\e \ll \e$ & $\frac{1}{\e^2} E_\e  - 2\pi M |\log \e| \leq C$  & $\mu_{u_\e} \flat \mu$  & $XY$ \\
				\hline
			\end{tabular}
		}
		\renewcommand{\arraystretch}{1}
		\caption{In this table we summarize our results. By ``interfaces'' we mean that the energy concentrates on $1$-dimensional domain walls that separate the different phases~\cite{CicOrlRuf3}, while ``$BV$'' denotes a $BV$-type total variation, Theorem~\ref{thm:e log smaller theta}. The expression ``$BV$+concentration'' indicates the presence in a $BV$-type energy of a surface term of the form $\mathcal{J}(\mu,u;\Omega)$ which accounts for concentration effects on 1-dimensional surfaces, as in~\cite{CicOrlRuf} and Theorem~\ref{thm:theta equal e log}. By ``vortices'' we mean that a logarithmic energy is carried by the system for the creation of vortex-like singularities in the limit. Finally, ``$XY$'' expresses the fact that the energy is a good approximation (at first order) of the classical $XY$ model, Theorem~\ref{thm:theta smaller eps}.}
		\label{table:energy behavior}
	\end{table}
	
	It is worth mentioning that the paper also contains an extension result for Cartesian currents, Lemma~\ref{lemma:extension of currents}, that we reckon to be of independent interest. 

\section{Notation and preliminaries}
In order to avoid confusion with lattice points we denote the imaginary unit by $\iota$. We shall often identify $\RR^2$ with the complex plane~$\CC$. Given a vector $a = (a_1,a_2) \in \RR^2$, its $1$-norm is $|a|_1 = |a_1|+|a_2|$. We define the $(2,1)$-norm of a matrix $A = (a_{ij}) \in \RR^{2 \times 2}$ by
\begin{equation*}
|A|_{2,1} := \big(a^2_{11} + a^2_{21}\big)^{1/2} + \big(a^2_{12} + a^2_{22}\big)^{1/2}.
\end{equation*}
If $u, v \in \SS^1$, their geodesic distance on $\SS^1$ is denoted by $\geo(u,v)$. It is given by the angle in $[0,\pi]$ between the vectors $u$ and $v$, i.e., $\geo(u,v) = \arccos(u\cdot v)$. Observe that
\begin{equation} \label{eq:geo and eucl}
\tfrac{1}{2}|u - v| = \sin\big( \tfrac{1}{2}\geo(u,v) \big)  \quad\text{and}\quad|u - v| \leq \geo(u,v) \leq \frac{\pi}{2}|u - v|.
\end{equation}
Given two sequences $\alpha_{\e}$ and $\beta_{\e}$, we write $\alpha_{\e}\ll \beta_{\e}$ if $\lim_{\e \to 0}\tfrac{\alpha_{\e}}{\beta_{\e}}=0$ and $\alpha_\e \sim \beta_\e$ if $\lim_{\e \to 0}\tfrac{\alpha_{\e}}{\beta_{\e}} \in (0,+\infty)$. We will use the notation $\deg(u)(x_0)$ to denote the topological degree of a continuous map $u \in C(B_\rho(x_0) \sm \{x_0\}; \SS^1)$, i.e., the topological degree of its restriction $u|_{\de B_r(x_0)}$, independent of $r < \rho$.

We denote by $I_\lambda(x)$ the half-open squares given by 
\begin{equation} \label{eq:half-open cubes}
	I_{\lambda}(x) = x + [0,\lambda)^2.
\end{equation}

\subsection{BV-functions}\label{subsec:BV}
In this section we recall basic facts about functions of bounded variation. For more details we refer to the monograph \cite{Amb-Fus-Pal}.

Let $O\subset\RR^d$ be an open set. A function $u\in L^1(O;\RR^n)$ is a function of bounded variation if its distributional derivative $\DD u$ is given by a finite matrix-valued Radon measure on $O$. We write $u\in BV(O;\RR^n)$.	
\\
The space $BV_{{\rm loc}}(O;\mathbb{R}^n)$ is defined as usual. The space $BV(O;\RR^n)$ becomes a Banach space when endowed with the norm $\|u\|_{BV(O)}=\|u\|_{L^1(O)}+|\DD u|(O)$, where $|\DD u|$ denotes the total variation measure of $\DD u$. The total variation with respect to the anisotropic norm $| \cdot |_{2,1}$ is denoted by $|\DD u|_{2,1}$. When $O$ is a bounded Lipschitz domain, then $BV(O;\RR^n)$ is compactly embedded in $L^1(O;\RR^n)$. We say that a sequence $u_n$ converges weakly$^*$ in $BV(O;\RR^n)$ to $u$ if $u_n\to u$ in $L^1(O;\RR^n)$ and $\DD u_n\overset{*}{\rightharpoonup}\DD u$ in the sense of measures.

Given $x\in O$ and $\nu\in \SS^{d-1}$ we set
\begin{equation*}
B^{\pm}_{\rho}(x,\nu)=\{y\in B_{\rho}(x):\;\pm (y-x) \cdot \nu >0\}\,.
\end{equation*}
We say that $x\in O$ is an approximate jump point of $u$ if there exist $a\neq b\in\mathbb{R}^n$ and $\nu\in \SS^{d-1}$ such that
\begin{equation*}
\lim_{\rho\to 0}\frac{1}{\rho^d}\int_{B_{\rho}^+(x,\nu)}|u(y)-a|\,\mathrm{d}y=\lim_{\rho\to 0}\frac{1}{\rho^d}\int_{B^-_{\rho}(x,\nu)}|u(y)-b|\,\mathrm{d}y=0 \, .
\end{equation*}
The triplet $(a,b,\nu)$ is determined uniquely up to the change to $(b,a,-\nu)$. We denote it by $(u^+(x),u^-(x),\nu_u(x))$ and let $J_u$ be the set of approximate jump points of $u$. The triplet $(u^+,u^-,\nu_u)$ can be chosen as a Borel function on the Borel set $J_u$. Denoting by $\nabla u$ the approximate gradient of $u$, we can decompose the measure $\DD u$ as
\begin{equation*}
\DD u(B)=\int_B\nabla u\,\mathrm{d}x+\int_{J_u\cap B}(u^+(x)-u^-(x))\otimes\nu_u(x)\,\mathrm{d}\mathcal{H}^{d-1}+\DD^{(c)}u(B) \, ,
\end{equation*}
where $\DD^{(c)}u$ is the so-called Cantor part and $\DD^{(j)}u = (u^+ - u^- )\otimes\nu_u \mathcal{H}^{d-1} \mres J_u$ is the so-called jump part.
 
\subsection{Results for the classical $XY$ model}
We recall here some results about the classical~$XY$ model, namely when the spin field $u_\e \colon \Omega_\e \to \SS^1$ is not constrained to take values in a discrete set.

Following~\cite{Ali-Cic-Pon}, in order to define the discrete vorticity of the spin variable, it is convenient to introduce the projection $Q \colon \RR \to 2 \pi \ZZ$ defined by
\begin{equation} \label{eq:def of projection}
Q(t) := \argmin \{|t - s| \ : \ s \in 2 \pi \ZZ \} \, ,
\end{equation}
with the convention that, if the argmin is not unique, then we choose the one with minimal modulus. Then for every $t \in \RR$ we define (see Figure~\ref{fig:graph of Psi})
\begin{equation} \label{eq:def of Psi}
\Psi(t) := t - Q(t)  \in [-\pi, \pi] \, .
\end{equation}

\begin{figure}[H]
    \includegraphics{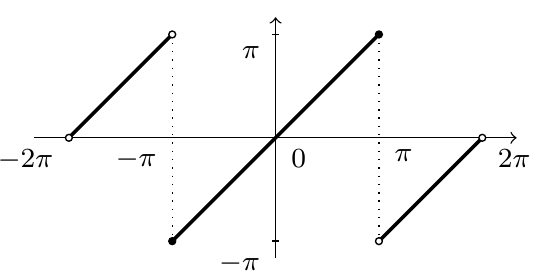}
	% \begin{tikzpicture}[scale=0.35]
	% \clip (-8,-4) rectangle (8,4);
	% \draw[->] (-7,0) -- (7,0);
	% \draw[->] (0,-3.5) -- (0,3.5);
	% \draw[line width=1pt] (-3,-3)--(3,3);
	% \draw[style=dotted] (3,3) -- (3,-3);
	% \draw[line width=1pt] (3,-3) -- (6,0);
	% \draw[style=dotted] (-3,3) -- (-3,-3);
	% \draw[line width=1pt] (-6,0) -- (-3,3);
	% \draw (0,0) node[anchor=north west] {\footnotesize $0$};
	% \draw (3,0) node[anchor=north west] {\footnotesize$\pi$};
	% \draw (6,0) node[anchor=north west] {\footnotesize$2\pi$};
	% \draw (-3,0) node[anchor=north east] {\footnotesize$-\pi$};
	% \draw (-6,0) node[anchor=north east] {\footnotesize$-2\pi$};
	
	% \draw (-0.1,3) -- (0.1,3);
	% \draw (0,3) node[anchor=north east] {\footnotesize$\pi$};
	% \draw (-0.1,-3) -- (0.1,-3);
	% \draw (0,-3) node[anchor=north east] {\footnotesize$-\pi$};
	
	% \draw[fill=black] (3,3) circle (0.1);
	% \draw[fill=white] (-3,3) circle (0.1);
	% \draw[fill=white] (3,-3) circle (0.1);
	% \draw[fill=black] (-3,-3) circle (0.1);
	% \draw[fill=white] (6,0) circle (0.1);
	% \draw[fill=white] (-6,0) circle (0.1);
	% \end{tikzpicture}
	\caption{Graph of the function $\Psi$ for $t \in (-2 \pi, 2 \pi)$. Observe that $\Psi$ is an odd function.}
	\label{fig:graph of Psi}	
\end{figure}
Let $u  \colon \e \ZZ^2 \to \SS^1$  and let $\varphi \colon \e\ZZ^2 \to [0, 2\pi)$ be the phase of $u$ defined by the relation $u = \exp(\iota \varphi)$.  The discrete vorticity of $u$ is defined for every $\e i \in \e \ZZ^2$ by 
\begin{equation} \label{eq:discrete vorticity}
\begin{split}
d_u(\e i) := & \frac{1}{2\pi} \Big[ \Psi\big(\varphi(\e i + \e e_1) - \varphi(\e i) \big) + \Psi\big(\varphi(\e i + \e e_1 + \e e_2) - \varphi(\e i + \e e_1) \big)  \\
& \quad + \Psi\big(\varphi(\e i + \e e_2) - \varphi(\e i + \e e_1 + \e e_2) \big) + \Psi\big(\varphi(\e i) - \varphi(\e i + \e e_2) \big)  \Big] \, . 
\end{split}
\end{equation}
As already noted in~\cite{Ali-Cic-Pon}, the discrete vorticity $d_u$ only takes values in $\{-1,0,1\}$, i.e., only singular vortices can be present in the discrete setting. We introduce the discrete measure representing all vortices of the discrete spin field defined by
\begin{equation} \label{eq:discrete vorticity measure}
\mu_u := \sum_{\e i \in \e \ZZ^2} d_u(\e i) \delta_{\e i + (\e,\e)} \, . 
\end{equation} 

\begin{remark}\label{eq:shiftedvorticity}
In \cite{Ali-Cic-Pon, Ali-DL-Gar-Pon} the vorticity measure $\mathring\mu_u$ is supported in the centers of the squares completely contained in $\Omega$, i.e., 
\begin{equation*}
\mathring\mu_{u}=\sum_{\substack{\e i\in \e\ZZ^2\\ \e i+[0,\e]^2\subset\Omega}}d_u(\e i)\delta_{\e i + 1/2(\e,\e)}\,.
\end{equation*}
In this paper we prefer definition \eqref{eq:discrete vorticity measure} since it fits well with our definition of discrete currents in Section \ref{s.discretecurrents} on the whole set $\Omega$. However, as we will borrow some results from \cite{Ali-Cic-Pon,Ali-DL-Gar-Pon}, we have to ensure that these definitions are asymptotically equivalent with respect to the flat convergence defined below. 
\end{remark}
\begin{definition}[Flat convergence] Let $O \subset \RR^2$ be an open set. A sequence of measures $\mu_j \in \M_b(O)$ converges {\em flat} to $\mu \in \M_b(O)$, denoted by $\mu_j \flat \mu$, if 
\begin{equation}\label{eq:defflatconv}
\sup_{\substack{\psi\in C^{0,1}_c(O)\\ \|\psi\|_{C^{0,1}}\leq 1}} \left|\integral{O}{ \psi}{ \d \mu_j}-\integral{O}{ \psi}{ \d \mu}\right| \to 0 \, .
\end{equation}
\end{definition}
Observe that the flat convergence is weaker than the weak* convergence. The two notions are equivalent when the measures~$\mu_j$ have equibounded total variations.

The two vorticity measures $\mu_{u}$ and $\mathring\mu_{u}$ are then close as explained in Lemma~\ref{lemma:vorticityequiv} below. For $A \subset \RR^2$ we shall use the localized energy given by 
\begin{equation*}
	E_\e(u;A) :=  \frac{1}{2} \sum_{\substack{\langle i, j \rangle \\ \e i, \e j \in A }} \e^2 |u(\e i) - u(\e j)|^2 .
\end{equation*} 
We shall adopt the same notation for the $XY_\e$ energy. We work with spin fields $u_{\e} \colon \e\ZZ^2\to\SS^1$ defined on the whole lattice $\e \ZZ^2$. We can always assume that $XY_\e(u_\e;\overline{\Omega}^\e) \leq C XY_\e(u_\e;\Omega)$, where $\overline{\Omega}^\e$ is the union of the squares $\e i + [0,\e]^2$ that intersect $\Omega$. (If not, thanks to the Lipschitz regularity of $\Omega$, we modify $u_\e$ outside $\Omega$ in such a way that the energy estimate is satisfied, see~\cite[Remark~2]{Ali-Cic}.)

\begin{lemma}\label{lemma:vorticityequiv}
Assume that $u_{\e} \colon \e\ZZ^2\to\SS^1$ is a sequence such that $\frac{1}{\e^2}XY_{\e}(u_{\e})\leq C|\log \e |$. Then $\mu_{u_{\e}}\mres\Omega-\mathring{\mu}_{u_{\e}}\flat 0$. 
\end{lemma}
\begin{proof}
Note that, for any $\psi\in C_c^{0,1}(\Omega)$ with $\|\psi\|_{C^{0,1}}\leq 1$, we have
\begin{equation*}
|\langle \mu_{u_{\e}}\mres\Omega-\mathring{\mu}_{u_\e},\psi\rangle|\leq \sum_{\e i\in\e\ZZ^2 \cap \overline \Omega^\e}|d_{u_\e}(\e i)|\frac{\e}{\sqrt{2}}
\leq \frac{\e}{\sqrt{2}}|\mathring{\mu}_{u_\e}|(\overline \Omega^\e)\leq C\e \frac{1}{\e^2}XY_{\e}(u;\overline \Omega^\e)\leq C\e|\log \e|\,,
\end{equation*}
where in the last but one inequality we used \cite[Remark 3.4]{Ali-Cic-Pon}. This proves the claim. 
\end{proof}

We recall the following compactness and lower bound for the $XY$ model. 

\begin{proposition} \label{prop:XY classical}
Let $u_\e \colon \e \ZZ^2 \to \SS^1$ and assume that there is a constant $C > 0$ such that $\frac{1}{\e^2|\log \e|} XY_\e(u_\e) \leq C$. Then there exists a measure $\mu \in \M_b(\Omega)$ of the form $\mu=\sum_{h=1}^Md_h\delta_{x_h}$ with $d_h\in\ZZ$ and $x_h\in\Omega$, and a subsequence (not relabeled) such that $\mu_{u_\e} \mres \Omega \flat \mu$. Moreover 
\begin{equation*}
2 \pi |\mu|(\Omega)  \leq \liminf_{\e \to 0} \frac{1}{\e^2|\log \e|} XY_\e(u_\e) \, .
\end{equation*}
\end{proposition}
\begin{proof}
% Note that \cite[Theorem 4.2 \& Remark 4.3]{Ali-Cic} guarantee that (up to a subsequence) the Jacobians of the piecewise affine interpolations of the spin field converge flat to a measure of the claimed form satisfying also the lower bound. 
In~\cite[Theorem~3.1-(i)]{Ali-DL-Gar-Pon} it is proven that (up to a subsequence) the discrete vorticity measures~$\mathring{\mu}_{u_{\e}}$ converge flat to a measure of the claimed form satisfying also the lower bound. The claim thus follows from Lemma \ref{lemma:vorticityequiv}. 
\end{proof}

\section{Currents}

For the theory of currents and cartesian currents we refer to the books~\cite{Gia-Mod-Sou-I, Gia-Mod-Sou-II}. We recall here some notation, definitions, and basic facts about currents. We additionally prove some technical lemmata that we need in this paper and that were also used in \cite{CicOrlRuf}.

\subsection{Definitions and basic facts} \label{sec:currentsbasics}
Given an open set $O \subset \RR^d$, we denote by $\D^k(O)$ the space of $k$-forms $\omega \colon O \mapsto \Lambda^k \RR^d$ that are $C^\infty$ with compact support in $O$. A {\em $k$-current} $T \in \D_k(O)$ is an element of the dual of~$\D^k(O)$. The duality between a $k$-current and a $k$-form $\omega$ will be denoted by $T(\omega)$. The {\em boundary} of a $k$-current $T$ is the $(k{-}1)$-current $\de T \in \D_{k-1}(O)$ defined by $\de T(\omega) := T(\! \d \omega)$ for every $\omega \in \D^{k-1}(O)$  (or $\partial T:=0$ if $k=0$). As for distributions, the {\em support} of a current $T$ is the smallest relatively closed set $K$ in $O$ such that $T(\omega) = 0$ if $\omega$ is supported outside $K$. 
Given a  smooth map $f \colon O \to O' \subset \RR^{N'}$ such that $f$ is proper\footnote{that means, $f^{-1}(K)$ is compact in $O$ for all compact sets $K\subset O'$.}, $f^\# \omega  \in \D^k(O)$ denotes the pull-back of a $k$-form $\omega \in \D^k(O')$ through $f$. The {\em push-forward} of a $k$-current $T \in \D_k(O)$ is the $k$-current $f_\# T \in \D_k(O')$ defined by $f_\# T(\omega) := T(f^\# \omega)$. Given a $k$-form $\omega \in \D^k(O)$, we can write it via its components
\begin{equation*}
\omega = \sum_{|\alpha| = k} \omega_\alpha \d x^\alpha  \, , \quad \omega_\alpha \in C^\infty_c(O) \, ,
\end{equation*}
 where the expression $|\alpha|=k$ denotes all multi-indices $\alpha = (\alpha_1, \dots, \alpha_k)$ with $1 \leq \alpha_i \leq d$, and $\d x^\alpha = \d x^{\alpha_1} \w \dots \w \d x^{\alpha_k}$. The norm of $\omega(x)$ is denoted by $|\omega(x)|$ and it is the Euclidean norm of the vector with components $(\omega_\alpha(x))_{|\alpha|=k}$. The {\em total variation} of a $k$-current $T \in \D_k(O)$ is defined by
\begin{equation*}
|T|(O) := \sup \{T(\omega) \ : \ \omega \in \D^k(O), \ |\omega(x)| \leq 1 \} \, .
\end{equation*}
If $T \in \D_k(O)$ with $|T|(O) < \infty$, then we can define the measure $|T| \in \M_b(O)$ 
\begin{equation*}
|T|(\psi) := \sup \{T(\omega) \ : \ \omega \in \D^k(O), \ |\omega(x)| \leq \psi(x) \}
\end{equation*}
for every $\psi \in C_0(O)$, $\psi \geq 0$. As a consequence of Riesz's Representation Theorem (see \cite[2.2.3, Theorem 1]{Gia-Mod-Sou-I}) there exists a $|T|$-measurable function $\vec{T} \colon O \mapsto \Lambda_k \RR^d$ with $|\vec{T}(x)| = 1$ for $|T|$-a.e.\ $x \in O$ such that 
\begin{equation} \label{eq:representation}
T(\omega) = \integral{O}{\langle \omega(x), \vec{T}(x) \rangle}{ \d |T|(x) }
\end{equation}
for every $\omega \in \D^k(O)$. We note that if $T$ has finite total variation, then it can be extended to a linear functional acting on all forms with bounded, Borel-measurable coefficients  via the dominated convergence theorem. 
In particular, in this case the push-forward $f_\# T$ can be defined also for $f\in C^1(O,O')$ with bounded derivatives, cf.\ the discussion in \cite[p. 132]{Gia-Mod-Sou-I}.

A set $\M \subset O$ is a countably $\H^k$-rectifiable set if it can be covered, up to an $\H^k$-negligible subset, by countably many $k$-manifolds of class $C^1$. As such, it admits at $\H^k$-a.e.\ $x \in \M$ a tangent space $\mathrm{Tan}(\M,x)$ in a measure theoretic sense. A current $T \in \D_k(O)$ is an {\em integer multiplicity (i.m.) rectifiable current} if it is representable as 
\begin{equation} \label{eq:im rectifiable}
	T(\omega) = \integral{\M}{\langle \omega(x), \xi(x) \rangle \theta(x)}{\d \H^k(x)} \, ,\quad \text{for } \omega \in \D^k(O) \, ,
\end{equation}
where $\M \subset O$ is a $\H^k$-measurable and countably $\H^k$-rectifiable set, $\theta \colon \M \to \ZZ$ is  locally $\H^k \mres \M$-summable, and $\xi \colon \M \to \Lambda_k \RR^d$ is a $\H^k$-measurable map such that $\xi(x)$ spans $\mathrm{Tan}(\M,x)$ and $|\xi(x)| = 1$ for $\H^k$-a.e.\ $x \in \M$. 
We use the short-hand notation $T=\tau(\mathcal{M},\theta,\xi)$.  
One can always replace $\mathcal{M}$ by the set $\mathcal{M}\cap\theta^{-1}(\{0\})$, so that we may always assume that $\theta\neq 0$. 
Then the triple $(\mathcal{M},\theta,\xi)$ is uniquely determined up to $\mathcal{H}^k$-negligible modifications. Moreover, one can show, according to the Riesz's representation in~\eqref{eq:representation}, that $\vec{T}=\xi$ and the total variation\footnote{For i.m.\ rectifiable currents, the total variation coincides with the so-called mass. Hence, we will not distinguish between these two concepts.} is given by $|T|=|\theta|\mathcal{H}^k\mres\mathcal{M}$.

If $T_j$ are i.m.\ rectifiable currents and $T_j \weak T$ in $\D_k(O)$ with $\sup_j (|T_j|(V) + |\de T_j|(V) ) < \infty$ for every $V \compact O$, then by the Closure Theorem~\cite[2.2.4, Theorem 1]{Gia-Mod-Sou-I} $T$ is an i.m.\ rectifiable current, too. By $\llbracket \M \rrbracket$ we denote the current defined by integration over $\M$.

\subsection{Currents in product spaces} Let us introduce some notation for currents defined on the product space~$\RR^{d_1} \x \RR^{d_2}$. We will denote by $(x,y)$ the points in this space. The standard basis of the first space $\RR^{d_1}$ is $\{ e_1,\ldots, e_{d_1} \}$, while $\{ \bar e_1,\ldots, \bar e_{d_2} \}$ is the standard basis of the second space $\RR^{d_2}$. Given $O_1\subset \RR^{d_1},O_2 \subset \RR^{d_2}$ open sets, $T_1 \in \D_{k_1}(O_1)$, $T_2 \in \D_{k_2}(O_2)$ and a $(k_1+k_2)$-form $\omega \in \D^{k_1+k_2}(O_1 \x O_2)$ of the type
\begin{equation*}
\begin{split}
\omega(x,y) = \sum_{\substack{|\alpha|=k_1\\ |\beta|=k_2}} \omega_{\alpha \beta}(x,y) \d x^\alpha \w \d y^\beta \, ,
\end{split}
\end{equation*} 
the product current $T_1 \times T_2 \in \D_{k_1+k_2}(O_1 \x O_2)$ is defined by
\begin{equation*}
T_1 \x T_2( \omega) := T_1 \Big( \sum_{|\alpha|=k_1} T_2\Big(\sum_{|\beta| = k_2} \omega_{\alpha \beta}(x,y) \d y^\beta \Big) \d x^\alpha \Big),
\end{equation*}
while $T_1 \x T_2(\phi \d x^\alpha \w \d y^\beta) = 0$ if $|\alpha|+|\beta| =k_1+k_2$ but $|\alpha|\neq k_1$, $|\beta| \neq k_2$.

\subsection{Graphs} Let $O \subset \RR^d$ be an open set and $u \colon \Omega \to \RR^2$ a Lipschitz map. Then we can consider the $d$-current associated to the graph of $u$ given by $G_u := (\mathrm{id} \x u)_\# \llbracket O \rrbracket \in \D_2(O\x \RR^2)$, where $\mathrm{id} \x u \colon O \to O \x \RR^2$ is the map $(\mathrm{id} \x u)(x) = (x,u(x))$.  Note that by definition we have
\begin{equation*}
G_u(\omega)=\integral{O}{\langle \omega(x,u(x)),M(\nabla u(x))\rangle}{\d x}
\end{equation*}
for all $\omega\in\mathcal{D}^{d}(O\times \RR^2)$, with the $d$-vector \begin{equation}\label{eq:minors}
M(\nabla u)=(e_1+\partial_{x^1} u^1\bar e_1+\partial_{x^1} u^2\bar e_2)\wedge\ldots\wedge(e_d+\partial_{x^d}u^1\bar e_1+\partial_{x^d} u^2\bar e_2)\,.
\end{equation}
The above formula can be extended to the class $\mathcal{A}^1(O;\RR^2)$ defined by
\begin{equation*}
\mathcal{A}^1(O;\RR^2):=\{u\in L^1(O;\RR^2): u\text{ approx.\ diff.\ a.e.\ and all minors of $\nabla u$ are in }L^1(O)\}\,.
\end{equation*}

\begin{remark} \label{rmk:bd of smooth}
We recall that $\de G_u|_{\Omega \x \RR^2} = 0$ when $u\in W^{1,2}(O;\RR^2)\subset\mathcal{A}^1(O;\RR^2)$, see \cite[3.2.1, Proposition 3]{Gia-Mod-Sou-I}. This property however fails for general functions $u\in\mathcal{A}^1(O;\RR^2)$.
\end{remark}

 In Lemma \ref{lemma:extension of currents} we need to interpret the graphs of $W^{1,1}(O;\SS^1)$ as currents. This can be done because of the following observation.

\begin{lemma}\label{lemma:W11aregraphs}
Let $O\subset\RR^d$	be an open, bounded set. Then $W^{1,1}(O;\SS^1)\subset\mathcal{A}^1(O;\RR^2)$.
\end{lemma}
\begin{proof}
It is well-known that Sobolev functions are approximately differentiable a.e.\ Moreover, all $1$-minors of $\nabla u$ are in $L^1(O)$. We argue that all $2$-minors vanish at a.e.\ point. To this end, denote by $P\colon \RR^2\setminus\{0\}\to\RR^2$ the smooth mapping $P(x)=x/|x|$. Since for $u\in W^{1,1}(O;\SS^1)$ we have $u=P\circ u$ almost everywhere, for a.e.\ $x\in O$ the chain rule for approximate differentials yields 
\begin{equation*}
\nabla u(x)= \nabla P(u(x))\nabla u(x).
\end{equation*}
Since $\nabla P(u(x))$ has at most rank $1$, also $\nabla u(x)$ has at most rank $1$ and therefore all $2$-minors have to vanish as claimed.
\end{proof}

We will use the orientation of the graph of a smooth function $u\colon O\subset\RR^2\to\SS^1$ (cf. \cite[2.2.4]{Gia-Mod-Sou-I}). For such maps we have $|G_u| = \H^2 \mres \M$, where $\M = (\mathrm{id} \x u)(\Omega)$, and, for every $(x,y) \in \M$,
\begin{equation} \label{eq:smooth components}
\begin{split}
\sqrt{1+|\nabla u(x)|^2} \ \vec{G}_u(x,y) = & \ e_1 \wedge e_2 \\
+ & \ \de_{x^2} u^1(x) e_1 \wedge \bar e_1 + \de_{x^2} u^2(x) e_1 \wedge \bar e_2 \\
- & \ \de_{x^1}  u^1(x)  e_2 \wedge \bar e_1 - \de_{x^1}  u^2(x) e_2 \wedge \bar e_2 \, .
\end{split} 
\end{equation}

\subsection{Cartesian currents}
Let $O\subset\RR^d$ be a bounded, open set. We recall that the class of {\em cartesian currents} in~$O\x \RR^2$ is defined by
\begin{equation*}
\begin{split}
\cart(O \x \RR^2) := \{ & T \in \D_d(O \x \RR^2) \ : \ T \text{ is i.m.\ rectifiable, } \de T|_{O \x \RR^2} = 0, \\
& \pi^O_\# T = \llbracket O \rrbracket \, , \ T |_{\d x} \geq 0 \, , \ |T| < \infty \, , \ \|T\|_1 < \infty \} \, ,
\end{split}
\end{equation*}
where $\pi^O \colon O \x \RR^2 \to O$ denotes the projection on the first component, $T|_{\d x} \geq 0$ means that $T(\phi(x,y) \d x) \geq 0$ for every $\phi \in C^\infty_c(O \x \RR^2)$ with $\phi \geq 0$, and 
\begin{equation*}
\| T \|_1 = \sup \{ T(\phi(x,y)|y| \d x ) \ : \ \phi \in C^\infty_c(O \x \RR^2) \, , \ |\phi| \leq 1 \}  \, . 
\end{equation*}
Note that, if for some function $u$
\begin{equation} \label{eq:norm 1}
T(\phi(x,y) \d x) = \integral{O}{\phi(x,u(x))}{\d x} \quad \text{then} \quad \| T \|_1  = \integral{O}{|u|}{\d x} \, .
\end{equation}

The class of {\em cartesian currents} in $O \x \SS^1$ is 
\begin{equation*}
\cart(O \x \SS^1) := \{ T \in \cart(O \x \RR^2) \ : \ \supp(T) \subset \ol O \x \SS^1 \} \,,
\end{equation*}
(cf.\ \cite[6.2.2]{Gia-Mod-Sou-II} for this definition). We recall the following approximation theorem which explains that cartesian currents in $O \x \SS^1$ are precisely those currents that arise as limits of graphs of $\SS^1$-valued smooth maps. The proof, based on a regularization argument on the lifting of $T$, can be found in~\cite[Theorem 7]{Gia-Mod-Sou-S1}.\footnote{Notice that some results in~\cite{Gia-Mod-Sou-S1} require $O$ to have smooth boundary. This is not the case for this theorem, which is based on a local construction.}

\begin{theorem}[Approximation Theorem] \label{thm:approximation}
Let $T \in \cart(O \x \SS^1)$. Then there exists a sequence of smooth maps $u_h \in C^\infty(O;\SS^1)$ such that 
\begin{equation*}
G_{u_h} \weak T \quad \text{in } \D_d(O \x \RR^2)\quad\text{and}\quad|G_{u_h}|(O\x \RR^2) \to |T|(O \x \RR^2) \, .
\end{equation*}
\end{theorem}

Using the above approximation result, we now prove an extension result for cartesian currents, which we could not find in the literature.
\begin{lemma}[Extension of cartesian currents] \label{lemma:extension of currents}
Let $O \subset \RR^d$ be a bounded, open set with Lipschitz boundary and let $T \in \cart(O \x \SS^1)$. Then there exist an open set $\tilde O \Supset O$ and a current $T \in \cart(\tilde O \x \SS^1)$ such that $\tilde T|_{O \x \RR^2} = T$ and $|\tilde T|(\de O \x \RR^2) = 0$. 
\end{lemma}
\begin{proof}
Applying Theorem \ref{thm:approximation} we find a sequence $u_k\in C^{\infty}(O;\SS^1)$ such that $G_{u_k}\weak T$ in $\D_d(O\x\RR^2)$ and $|G_{u_k}|(O\x\RR^2)\to |T|(O\x\RR^2)$. In particular, the sequence $|G_{u_k}|(O\x\RR^2)$ is bounded, which implies that
\begin{equation}\label{eq:massbound}
\sup_k \integral{O}{|\nabla u_k|}{\mathrm{d}x}<C\,.
\end{equation}	
Next we extend the functions $u_k$. To this end, note that there exists $t>0$ and a bi-Lipschitz map $\Gamma\colon (\partial O\x(-t,t))\to \Gamma(\partial O\x (-t,t))$ such that $\Gamma(x,0)=x$ for all $x\in\partial O$, $\Gamma(\partial O\x (-t,t))$ is an open neighborhood of $\partial O$ and
\begin{equation}\label{eq:bicollar}
\Gamma(\partial O\x (-t,0))\subset O,\quad\quad \Gamma(\partial O\times(0,t))\subset\RR^2\setminus\ol O\,.
\end{equation}
This result is a consequence of \cite[Theorem 7.4 \& Corollary 7.5]{Luu-Vae}; details can be found for instance in \cite[Theorem 2.3]{Lic}. The extension of $u_k$ is then achieved via reflection. More precisely, for a sufficiently small $t'>0$ we define it on $O'$ with $O'=O+B_{t'}(0)$ by
\begin{equation}\label{eq:extensionbyreflection}
\tilde{u}_k(x)=
\begin{cases}
u_k(\Gamma(P(\Gamma^{-1}(x)))) &\mbox{if $x\notin O$} \, ,
\\
u_k(x) &\mbox{otherwise,}
\end{cases}
\end{equation}
where $P(x,\tau)=(x,-\tau)$. Since $\Gamma$ is bi-Lipschitz, we have that $\tilde{u}_k\in W^{1,1}(O';\SS^1)$ and by a change of variables we can bound the $L^1$-norm of its gradient via
\begin{equation}\label{eq:gradbound}
\integral{O'}{|\nabla \tilde{u}_k|}{\d x}\leq \integral{O}{|\nabla u_k|}{\d x}+C_{\Gamma} \integral{O'\setminus O}{|(\nabla u_k)\circ \Gamma\circ P\circ\Gamma^{-1}|}{\d x}\leq C_{\Gamma}\integral{O}{|\nabla u_k|}{\d x} \, ,
\end{equation}
where the constant $C_{\Gamma}$ depends only on the bi-Lipschitz properties of $\Gamma$ and the dimension. Lemma \ref{lemma:W11aregraphs} implies that $\tilde{u}_k\in\mathcal{A}^1(O';\RR^2)$.
In particular, the current $G_{\tilde{u}_k}\in\mathcal{D}_d(O'\times\RR^2)$ is well-defined in the sense of
\begin{equation*}
G_{\tilde{u}_k}(\omega )=\integral{O'}{\langle\omega(x,\tilde{u}_k(x)),M(\nabla\tilde{u}_k(x))\rangle}{\d x}\,,
\end{equation*}
with $M(\nabla\tilde{u}_k)$ given by \eqref{eq:minors}. We next prove that $G_{\tilde{u}_k}\in\cart(O'\times\SS^1)$. First note that whenever $\omega\in\mathcal{D}^d(O'\times\RR^2)$ is a form with $\supp(\omega)\compact O'\times\RR^2\setminus (O'\times\SS^1)$, then the definition yields $G_{\tilde{u}_k}(\omega)=0$. It then suffices to prove that $\partial G_{\tilde{u}_k}|_{O'\times\RR^2}=0$. We will argue locally. For each $x\in O'$ we choose a rotation $Q_x$, radii $r_x>0$, and heights $h_x>0$ such that the cylinders $C_x:=x+Q_x \left((-r_x,r_x)^{d-1}\times (-h_x,h_x)\right)$ satisfy
\begin{itemize}
	\item[(i)] $C_x\compact O$ if $x\in O$;
	\item [(ii)] $C_x\compact O'\setminus \overline{O}$ if $x\in O'\setminus\ol{O}$;
	\item[(iii)] $C_x\compact O'$ if $x\in\partial O$ and 
	\begin{equation*}
		C_x\cap O =C_x\cap \left(x+Q_x\{(x',x_d)\in\RR^d:\,x'\in (-r_x,r_x)^{d-1}: -h_x<x_d<\psi(x')\}\right)
	\end{equation*}
		for some $\psi\in {\rm Lip}(\RR^{d-1})$. 
\end{itemize}
For $x\in O$ we have $\partial G_{\tilde{u}_k}|_{ C_x\times\RR^2}=\partial G_{u_k}|_{C_x\times\RR^2}=0$ since $u_k\in C^{\infty}(C_x;\SS^1)$. Next consider the second case, namely $x\in O'\setminus\ol{O}$. Since $C_x\compact O'\setminus\ol{O}$, the properties in~\eqref{eq:bicollar} imply that $\Gamma\circ P\circ\Gamma^{-1}(C_x)\compact O$. In particular, by the smoothness of $u_k$ on $O$ we have that $\tilde{u}_k\in W^{1,\infty}(C_x)$, so that by Remark \ref{rmk:bd of smooth} again $\partial G_{\tilde{u}_k}|_{C_x\times\RR^2}=0$. Finally, we consider $x\in\partial O$. Since $C_x\cap O$ is (up to a rigid motion) the subgraph of a Lipschitz function, it is in particular simply connected. By classical lifting theory, we find a sequence of scalar functions $\varphi_k\in C^{\infty}(C_x\cap O)$ such that $u_k(x)=\exp(\iota\varphi_k(x))$. In particular, using the chain rule we see that $\varphi_k\in W^{1,1}(C_x\cap O)$. Now fix $0<\delta_x < r_x$ small enough such that $B_{\delta_x}(x)\subset C_x$ and
\begin{equation*}
(\Gamma\circ P\circ \Gamma^{-1})(B_{\delta_x}(x)\cap (O'\setminus \ol{O}))\subset C_x\cap O\,,
\end{equation*}
which is possible due to \eqref{eq:bicollar}. We then extend the lifting $\varphi_k$ to a function $\tilde{\varphi}_k\in W^{1,1}(B_{\delta_x}(x))$ via the same reflection construction as in \eqref{eq:extensionbyreflection}, which is well-defined due to the above inclusion. Observe that this definition guarantees that $\tilde{u}_k(y)=\exp(\iota\tilde{\varphi}_k(y))$ for almost every $y\in B_{\delta_x}(x)$. Expressed in terms of currents this means that 
\begin{equation*}
G_{\tilde{u}_k}|_{B_{\delta_x}(x)\times\RR^2}=\chi_\# G_{\tilde{\varphi}_k}|_{B_{\delta_x}(x)\times \RR^2}\,, 
\end{equation*}
where $G_{\tilde{\varphi}_k} \in \D_d(B_{\delta_x}(x)\x \RR)$ is the current associated to the graph of $\tilde{\varphi}_k$ and $\chi \colon \RR^d \x \RR \to \RR^d \x \SS^1$ is the covering map defined by $\chi(x,\vartheta) := (x, \cos(\vartheta), \sin(\vartheta))$. In particular, by \cite[Theorem 2, p. 97 \& Proposition 1 (i), p. 100]{Gia-Mod-Sou-S1} we have $G_{\tilde{u}_k}|_{B_{\delta_x}(x)\times\RR^2}\in\cart(B_{\delta_x}(x)\times\SS^1)$, so that by the definition of cartesian currents we have $\partial G_{\tilde{u}_k}|_{B_{\delta_x}(x)\times\RR^2}=0$.
	
Thus we have shown that for every $x\in O'$ there exists a ball $B_{\delta_x}(x)\subset O'$ such that $\partial G_{\tilde{u}_k}|_{B_{\delta_x}(x)\times\RR^2}=0$. Using a partition of unity to localize the support of any form $\omega\in\mathcal{D}^{d-1}(O'\times\RR^2)$ with respect to the $x$-variable, we conclude that $\partial G_{\tilde{u}_k}|_{O'\times\RR^2}=0$ and therefore $G_{\tilde{u}_k}\in\cart(O'\times\SS^1)$. As seen in the proof of Lemma \ref{lemma:W11aregraphs}, all $2$-minors of $\DD u$ vanish, so that the bounds \eqref{eq:massbound} and \eqref{eq:gradbound} yield
\begin{equation*}
|G_{\tilde{u}_k}|(O'\x \RR^2)=\integral{O'}{|M(\nabla\tilde{u}_k)|}{\d x}\leq \integral{O'}{\left(1+ |\nabla\tilde{u}_k|^2\right)^{\frac{1}{2}}}{\d x}\leq C\,. 
\end{equation*}
Hence, up to a subsequence, we can assume that $G_{\tilde{u}_k}\weak\tilde{T}$ in $\mathcal{D}_d(O'\x\RR^2)$, see~\cite[2.2.4 Theorem~2]{Gia-Mod-Sou-I}. From \cite[4.2.2. Theorem 1]{Gia-Mod-Sou-I} it follows that $\tilde{T}\in\cart(O'\x \SS^1)$. Since $\tilde{u}_k=u_k$ on $O$, we find that $\tilde{T}|_{O\x\RR^2}=T$. It remains to show that $|\tilde{T}|(\partial O\times\RR^2)=0$. To this end, note that for $0<\eta < \eta' < 1$ (and $\eta$ small enough), by the bi-Lipschitz continuity of $\Gamma$ and \eqref{eq:bicollar} we have that
\begin{equation*}
(\Gamma\circ P\circ\Gamma^{-1})(O_{\eta}^{\rm out})\subset O_{\eta'}^{\rm in}\,,
\end{equation*}
where the sets $O_{\eta}^{\rm out}$ and $O_{\eta'}^{\rm in}$ are defined as
\begin{equation*}
O_{\eta}^{\rm out}:=\{x\in O'\setminus\ol{O}:\,\dist(x,\partial O)<\eta\},\quad O_{\eta'}^{\rm in}=\{x\in O:\,\dist(x,\partial O)<\eta'\}\,.
\end{equation*}
Hence, similar to \eqref{eq:gradbound} we obtain that
\begin{equation}\label{eq:boundtubularnbhd}
|G_{\tilde{u}_k}|((\partial O+B_{\eta}(0))\x\RR^2)\leq C_{\Gamma}\integral{O_{\eta'}^{\rm in}}{\left(1+|\nabla u_k|^2\right)^{\frac{1}{2}}}{\d x}=C_{\Gamma}|G_{u_k}|(O_{\eta'}^{\rm in}\x\RR^2)\,.
\end{equation}
Since $|G_{u_k}|(O\x\RR^2)\to |T|(O\x\RR^2)$ and $|T|$ is a finite measure, for a.e.\ $\eta' \in (0,1)$ we have $|G_{u_k}|(O_{\eta'}^{\rm in}\x\RR^2)\to |T|(O_{\eta'}^{\rm in}\x\RR^2)$.  Applying the lower semicontinuity of the mass with respect to weak convergence of currents in \eqref{eq:boundtubularnbhd}, we infer that 
\begin{equation*}
|\tilde{T}|((\partial O+B_{\eta}(0))\x\RR^2)\leq C_{\Gamma}|T|(O_{\eta'}^{\rm in}\x\RR^2)\,.
\end{equation*}
Sending first $\eta\to 0$ and then $\eta'\to 0$ we conclude that $|\tilde{T}|(\partial O\x\RR^2)=0$ as claimed. 
\end{proof}

We will also use the structure theorem for cartesian currents in $O \x \SS^1$ that has been proven in~\cite[Section 3, Theorems 1, 5, 6]{Gia-Mod-Sou-S1}.\footnote{As for the Approximation Theorem, no boundary regularity is required for this result.} However, to simplify notation, from now on we focus on dimension two. Recall that $\Omega\subset\RR^2$ is a bounded, open set with Lipschitz boundary. To state the theorem, we recall the following decomposition for a current $T \in \cart(\Omega \x \SS^1)$. Letting $\M$ be the countably $\H^2$-rectifiable set where $T$ is concentrated, we denote by $\M^{(a)}$ the set of points  $(x,y) \in \M$ at which the tangent plane $\mathrm{Tan}(\M,(x,y))$ does not contain vertical vectors (namely, the Jacobian of the projection $\pi^\Omega$ restricted to $\mathrm{Tan}(\M,(x,y))$ has maximal rank), by $\M^{(jc)} := (\M \sm \M^{(a)}) \cap (J_T \x \SS^1)$, where $J_T := \{ x \in \Omega \ : \ \frac{\d \pi^\Omega_\# |T|}{\d \H^1}(x) > 0 \}$, and by $\M^{(c)} := \M \sm (\M^{(a)} \cup \M^{(jc)})$. Then we can split the current as
\begin{equation*}
T = T^{(a)} + T^{(c)} + T^{(jc)} \, ,
\end{equation*}
where $T^{(a)} := T \mres \M^{(a)}$, $T^{(c)} := T \mres \M^{(c)}$, $T^{(jc)} := T \mres \M^{(jc)}$ are mutually singular measures, and we denote by $\mres$ the restriction of the Radon measure $T$.
Hereafter we use the notation $ \widehat x^1 =  x^2$ and $ \widehat x^2 = x^1$. 

\begin{theorem}[Structure Theorem for $\cart(\Omega \x \SS^1)$] \label{thm:structure} 
Let $T \in \cart(\Omega \x \SS^1)$. Then there exists a unique map $u_T \in BV(\Omega;\SS^1)$ and an (not unique) i.m.\ rectifiable 1-current $L_T =\tau(\L,k,\vec{L}_T) \in \D_1(\Omega)$  such that $T^{(jc)} = T^{(j)} + L_T \x \llbracket \SS^1 \rrbracket$ and
\begin{align}
T(\phi(x,y) \d x) & = T^{(a)}(\phi(x,y) \d x)  = \integral{\Omega}{\phi(x,u_T(x))}{\d x} \, ,   \label{eq:T horiz} \\
T^{(a)}(\phi(x,y) \d \widehat x^l \w \d y^m) & = (-1)^{2-l} \integral{\Omega}{\phi(x,u_T(x)) \de^{(a)}_{x^l} u_T^m(x)}{\d x} \, , \\
T^{(c)}(\phi(x,y) \d \widehat x^l \w \d y^m) &= (-1)^{2-l} \integral{\Omega}{\phi(x,\tilde u_T(x)) }{\d \de^{(c)}_{x^l} u_T^m(x)} \, , \label{eq:T cantor} \\
T^{(j)}(\phi(x,y) \d \widehat x^l \w \d y^m) &= (-1)^{2-l} \integral{J_{u_T}}{\bigg\{\integral{\gamma_x}{\phi(x,y)}{\d y^m}\bigg\} \nu_{u_T}^l(x)}{\d \H^1(x)} \label{eq:T jump}
\end{align}
for every $\phi \in C^\infty_c(\Omega \x \RR^2)$, $\gamma_x$ being the (oriented) geodesic arc in $\SS^1$ that connects $u_T^-(x)$ to $u_T^+(x)$ and $\tilde u_T$ being the precise representative of $u_T$.
% Moreover, the three measures $|T^{(a)}|$, $|T^{(c)}|$, and $|T^{(jc)}|$ are mutually singular.
\end{theorem}
 \begin{remark}\label{r.differentL}
In \cite[Theorem 6]{Gia-Mod-Sou-S1}	the structure of $T^{(j)}$ is formulated in a slightly different way, using the counter-clockwise arc $\gamma_{\varphi^-,\varphi^+}$ between $(\cos(\varphi^-),\sin(\varphi^-))$ and $(\cos(\varphi^+),\sin(\varphi^+))$ and replacing $J_{u_T}$ by $J_{\varphi}$, where $\varphi\in BV(\Omega)$ is a local lifting of $T$.
More precisely, the notion of lifting is understood in the sense that $T=\chi_{\#}G_{\varphi}$, where $\chi \colon \RR^2\x \RR\to\RR^2\x \SS^1$ is the covering map $(x,\vartheta)\mapsto (x,\cos(\vartheta),\sin(\vartheta))$  and $G_\varphi \in \cart(\Omega \x \RR)$ is the cartesian current given by the boundary of the subgraph of $\varphi$ (hence the push-forward via $\chi$ is well-defined as $G_{\varphi}$ has finite mass, see Section \ref{sec:currentsbasics}). 
To explain how to deduce~\eqref{eq:T jump}, we recall the local construction in~\cite{Gia-Mod-Sou-S1}: for every $x\in J_{\varphi}$ one chooses $p^+(x)\geq 0$ and $k'(x)\in\mathbb{N}\cup\{0\}$ such that
\begin{equation*}
	\varphi^+(x)=p^+(x)+2\pi k'(x),\quad\quad 0\leq p^+(x)-\varphi^-(x)<2\pi\,,
\end{equation*}
where we recall that in the scalar case the traces (and the normal to the jump set) are arranged to satisfy $\varphi^-<\varphi^+$ on $J_{\varphi}$. Then, locally, the $1$-current $L_T'$ in \cite[Theorem 6]{Gia-Mod-Sou-S1} is given by $L_T'=\tau(\mathcal{L}',k'(x),\vec{L}_T')$, where $\mathcal{L}'\subset J_{\varphi}$ denotes the set of points with $k'(x)\geq 1$ and $\vec{L}_T'$ is the orientation of $\mathcal{L}'$ defined via $\vec{L}_T'= \nu_{\varphi}^2e_1 - \nu_{\varphi}^1e_2$. To obtain the representation via geodesics, we let 
\begin{equation*}
	(q^+(x), k(x))=
	\begin{cases}
		(p^+(x),k'(x)) &\mbox{if $p^+(x)-\varphi^-(x)<\pi$} \, ,
		\\
		(p^+(x)-2\pi,k'(x)+1) &\mbox{if $p^+(x)-\varphi^-(x)>\pi$} \, ,
	\end{cases}
\end{equation*}
The case $p^+(x)-\varphi^-(x)=\pi$, i.e, antipodal points, needs special care. In this case we define~$q^+(x)$ and $k(x)$ according to the following rule: let $\tilde{\varphi}^\pm(x):=\varphi^{\pm}(x)\mod 2\pi \in [0,2 \pi)$. Then
\begin{equation*}
	(q^+(x),k(x))=
	\begin{cases}
		(p^+(x),k'(x)) &\mbox{if $\Psi(\tilde{\varphi}^+(x) - \tilde{\varphi}^-(x)) = \pi$} \, ,
		\\
		(p^+(x)-2\pi,k'(x)+1) &\mbox{if $\Psi(\tilde{\varphi}^+(x) - \tilde{\varphi}^-(x)) = -\pi$} \, ,
	\end{cases}
\end{equation*}
with the function $\Psi$ defined in \eqref{eq:def of Psi}. Replacing $(p^+(x),k'(x))$ by $(q^+(x), k(x))$, the modified structure of $T^{(j)}$ can be proven following exactly the lines of \cite[p.107-108]{Gia-Mod-Sou-S1}, noting that by the chain rule in $BV$ \cite[Theorem 3.96]{Amb-Fus-Pal} we have $J_{\varphi}=J_{u_T}\cup\{x\in J_{\varphi}:\,q^+(x)=\varphi^-(x)\}$. In particular, 
\begin{equation}\label{eq:localstructureofL}
	L_T=\tau(\mathcal{L}, k,\vec{L}_T) \, , \quad\quad\mathcal{L}=\{x\in J_{\varphi}:\, k(x)\geq 1\} \, ,\quad\quad\vec{L}_T=\nu_{\varphi}^2e_1-\nu_{\varphi}^1e_2
\end{equation} 
still depend on the local lifting $\varphi$, but in~\eqref{eq:T jump} the curves $\gamma_{\varphi^-,\varphi^+}$ are replaced by the more intrinsic geodesic arcs $\gamma_x$ connecting $u_T^{-}(x)=(\cos(\varphi^{-}(x)),\sin(\varphi^{-}(x)))$ to $u_T^{+}(x)=(\cos(\varphi^{+}(x)),\sin(\varphi^+(x)))$ (these formulas are consistent with the choice $\nu_{u_T}(x)=\nu_{\varphi}(x)$). In particular, exchanging $u_T^-(x)$ and $u_T^+(x)$ will change the orientation of the arc (also in the case of antipodal points) and of the normal $\nu_{u_T}(x)$,\footnote{More precisely, assume that $u_1, u_2, \nu \in \SS^1$ and assume that the geodesic arc from $u_1$ to $u_2$ is counterclockwise. If $(u_T^+(x), u_T^-(x),\nu_{u_T}(x)) = (u_2,u_1,\nu)$, then $\gamma_x$ is oriented counterclockwise. If, instead,  $(u_T^+(x), u_T^-(x),\nu_{u_T}(x)) = (u_1,u_2,-\nu)$ (equivalent to the first choice, according to the definition of jump point), then $\gamma_x$ is oriented clockwise.} so that the formula for $T^{(j)}$ is invariant, hence well-defined without the use of local liftings.	
\end{remark}

It is convenient to recast the jump-concentration part of $T \in \cart(\Omega \x \SS^1)$ in the following way. Let $L_T = \tau(\L,k,\vec{L}_T)$ as in Theorem~\ref{thm:structure}. We introduce for $\H^1$-a.e.\ $x \in J_T$ the normal $\nu_T(x)$ to the 1-rectifiable set $J_T = J_{u_T} \cup \L$ as 
\begin{equation}\label{eq:defnormal}
\nu_T(x)=
\begin{cases}
\nu_{u_T}(x) &\mbox{if $x\in J_{u_T}$} \, ,
\\
(-\vec{L}^2_T(x), \vec{L}^1_T(x)) &\mbox{if $x\in\mathcal{L}\setminus J_{u_T}$} \, ,
\end{cases}
\end{equation}
where we choose $\nu_{u_T}(x) = (-\vec{L}^2_T(x), \vec{L}^1_T(x))$ if $x \in \mathcal{L} \cap J_{u_T}$. For $\H^1$-a.e.\ $x \in J_T$ we consider the curve~$\gamma^T_x$ given by: the (oriented) geodesic arc $\gamma_x$ which connects~$u_T^-(x)$ to~$u_T^+(x)$ if $x \in J_{u_T} \sm \L$ (in the sense of Remark~\ref{r.differentL} in case of antipodal points); the whole~$\SS^1$ turning $k(x)$ times if $x \in \L \sm J_{u_T}$, $k(x)$ being the integer multiplicity of $L_T$; the sum (in the sense of currents)\footnote{In this case, a more elementary way of defining $\gamma_x^T$ is the following: let $\gamma_x \colon [0,1] \to \SS^1$ be the geodesic arc, and let $\varphi_x \colon [0,1] \to \RR$ be a continuous function (unique up to translations of an integer multiple of $2 \pi$) such that $\gamma_x(t) = \exp(\iota \varphi_x(t))$. Then $\gamma_x^T(t) = \exp\big(\iota (1-t) \varphi_x(0) + \iota t(\varphi_x(1) + 2 \pi k(x) )\big)$.} of the oriented geodesic arc~$\gamma_x$ and of $\SS^1$ with multiplicity $k(x)$ if $x \in J_{u_T} \cap \L$. Then
\begin{equation} \label{eq:jc part of T}
T^{(jc)}(\phi(x,y) \d \widehat x^l \w  \d y^m) = (-1)^{2-l}\integral{J_T}{\Big\{ \integral{ \gamma^T_x}{\phi(x,y)}{\d y^m} \Big\} \nu_T^{l}(x)}{\d \H^1(x)}  \, .
\end{equation}
The integration over $\gamma_x^T$ with respect to the form $\d y^m$ in the formula above is intended with the correct multiplicity of the curve $\gamma_x^T$ defined for $\H^1$-a.e.\ $x \in J_T$ by the integer number 
\begin{equation} \label{eq:jc multiplicity}
\mathfrak{m}(x,y) := \begin{cases}
\pm 1 \, , & \text{if } x \in J_{u_T} \sm \L \, ,\  y \in \supp(\gamma_x) \, , \\
k(x) \, , & \text{if } x \in \L \sm J_{u_T} \, ,\ y \in \SS^1,\\
k(x) \pm 1 \, , & \text{if } x \in \L \cap J_{u_T} \, ,\  y \in \supp(\gamma_x)  \, ,  \\
k(x) \, , & \text{if } x \in \L \cap J_{u_T} \, ,\  y \in \supp(\gamma^T_x) \sm \supp(\gamma_x)  \, ,
\end{cases}
\end{equation}
where $\pm = +/ -$ if the geodesic arc $\gamma_x$ is oriented counterclockwise/clockwise, respectively. More precisely,
\begin{equation} \label{eq:from dym to H1}
\integral{ \gamma^T_x}{\phi(y)}{\d y^m} = (-1)^m \,  \!\!\! \integral{ \supp(\gamma^T_x)}{\phi(y) \widehat y^m \mathfrak{m}(x,y)}{\d \H^1(y)}  \, .
\end{equation}
\begin{remark}\label{rmk:choiceoforientation}
Note that we constructed $\mathfrak{m}(x,y)$ based on the orientation \eqref{eq:defnormal} of $\nu_T$. As discussed in Remark \ref{r.differentL}, changing the orientation of $\nu_{u_T}$ changes the orientation of the geodesic $\gamma_x$, while a change of the orientation of $\vec{L}_T$ switches the sign of $k(x)$. Hence changing the orientation of $\nu_T(x)$ changes $\mathfrak{m}(x,y)$ into $-\mathfrak{m}(x,y)$. If we choose locally $\nu_T=\nu_{\varphi}$ as in Remark~\ref{r.differentL}, our construction above yields $\mathfrak{m}(x,y)\geq 0$.
\end{remark}
In the proposition below, we derive an explicit formula for the 2-vector $\vec{T}$ of a cartesian current. 

\begin{proposition} \label{prop:components of T}
Let $T \in \cart(\Omega \x \SS^1)$, let $u_T$ be the $BV$ function associated to $T$. Then $|T^{(a)}|=\H^2 \mres \M^{(a)}$, $|T^{(c)}|=\H^2 \mres \M^{(c)}$, $|T^{(jc)}|= |\mathfrak{m}| \H^2 \mres \M^{(jc)}$, and
\begin{equation} \label{eq:ac components}
\begin{split}
\sqrt{1+|\nabla u_T(x)|^2} \ \vec{T}(x,y) = & \ e_1 \wedge e_2 \\
+ & \ \de_{x^2}^{(a)} u_T^1(x) e_1 \wedge \bar e_1 + \de_{x^2}^{(a)} u_T^2(x) e_1 \wedge \bar e_2 \\
- & \ \de_{x^1}^{(a)} u_T^1(x)  e_2 \wedge \bar e_1 - \de_{x^1}^{(a)} u_T^2(x) e_2 \wedge \bar e_2 \, ,
\end{split} 
\end{equation}
for $\H^2$-a.e.\ $(x,y) \in \M^{(a)}$,
\begin{equation} \label{eq:cantor components}
\begin{split}
\vec{T}(x,y) = &  \frac{ \d \de_{x^2}^{(c)}u_T^1}{\d |\DD^{(c)} u_T|}(x)  e_1 \wedge \bar e_1 + \frac{\d \de_{x^2}^{(c)} u_T^2}{\d |\DD^{(c)} u_T|}(x) e_1 \wedge \bar e_2 \\
 - &\frac{\d \de_{x^1}^{(c)} u_T^1}{\d |\DD^{(c)} u_T|}(x)  e_2 \wedge \bar e_1 - \frac{\d \de_{x^1}^{(c)} u_T^2}{\d |\DD^{(c)} u_T|}(x) e_2 \wedge \bar e_2 \, ,
\end{split} 
\end{equation}
for $\H^2$-a.e.\ $(x,y) \in \M^{(c)}$, and 
\begin{equation} \label{eq:jc components}
\begin{split}
{\rm sign}(\mathfrak{m}(x,y))\vec{T}(x,y) = - & \ \nu^2_{T}(x) y^2 e_1 \wedge \bar e_1 + \nu^2_{T}(x) y^1 e_1 \wedge \bar e_2 \\
 + & \ \nu^1_{T}(x) y^2  e_2 \wedge \bar e_1 - \nu^1_{T}(x) y^1 e_2 \wedge \bar e_2 \, ,
\end{split} 
\end{equation}
for $\H^2$-a.e.\ $(x,y) \in \M^{(jc)}$, where $\mathfrak{m}(x,y)$ is the integer defined in~\eqref{eq:jc multiplicity}.
\end{proposition}

\begin{proof}
 Assume $\Omega$ simply connected (if not, the following arguments can be repeated locally). Let us consider the covering map $\chi \colon \Omega \x \RR \to \Omega \x \SS^1$ defined by $\chi(x,\vartheta) := (x, \cos(\vartheta), \sin(\vartheta))$. By \cite[Corollary 1, p.~105]{Gia-Mod-Sou-S1} there exists a lifting of $T$, i.e., there is a function $\varphi \in BV(\Omega;\RR)$ such that $T = \chi_\# G_\varphi$, where $G_\varphi \in \cart(\Omega \x \RR)$ is the cartesian current given by the boundary of the subgraph of $\varphi$. The fine structure of such currents is well known, compare~\cite[Theorem~4.5.9]{Fed},  \cite[4.1.5 \& 4.2.4]{Gia-Mod-Sou-I}. We recall here that, if we consider the subgraph $SG_\varphi := \{ (x,y) \in \Omega \x \RR \ : \ y < \varphi(x) \}$, then $SG_\varphi$ is a set of finite perimeter; $G_\varphi$ is the current $G_\varphi = \de \llbracket SG_\varphi \rrbracket$. The interior normal to $SG_\varphi$ is given by 
\begin{equation}\label{eq:normalvector}
\begin{split}
n(x,\varphi(x)) &= \frac{\d (\DD \varphi, -\L^2)}{\d |(\DD \varphi, - \L^2)|}(x) \, , \quad \text{for } x \in \Omega \sm J_\varphi \, , \\
n(x,\vartheta) &= (\nu_{\varphi}(x),0) \, , \quad \text{for } x\in J_\varphi \, , \ \vartheta \in [ \varphi^-(x), \varphi^+(x)] \, ,
\end{split}
\end{equation}
where $\nu_\varphi$ is the normal to the jump set $J_\varphi$  and $\L^2$ denotes the two-dimensional Lebesgue measure. Moreover, the current $G_\varphi$ can be represented as $G_\varphi = \vec{G}_\varphi |G_\varphi|$ where $|G_\varphi|$ is concentrated on the reduced boundary $\de^- SG_u$, $|G_\varphi| = \H^2 \mres \de^- S G_u$, and $\vec{G}_\varphi$ is the 2-vector in $\RR^3$ such that $-\vec{G}_\varphi(x,\vartheta) \wedge n(x,\vartheta) = e_1 \wedge e_2 \wedge  e_3$, i.e.,
\begin{equation*}
\vec{G}_\varphi = -n^3 e_1 \wedge e_2 + n^2 e_1 \wedge e_3 - n^1 e_2 \wedge e_3 \, .
\end{equation*}
Finally, letting
\begin{align*}
\Sigma^{(a)} & := \{ (x,  \tilde \varphi(x)) \ : \ x \in \Omega \sm J_\varphi , \ n^3(x,u(x)) \neq 0\} \, , \\
\Sigma^{(c)} & := \{ (x, \tilde \varphi(x)) \ : \ x \in \Omega \sm J_\varphi , \ n^3(x,u(x)) = 0\} \, , \\
\Sigma^{(j)} & := \{ (x, \vartheta) \ : \ x \in J_\varphi , \ \vartheta \in [\varphi^-(x), \varphi^+(x)], \  n^3(x,\vartheta) = 0\} \, ,
\end{align*}
we have $\de^- SG_\varphi = \Sigma^{(a)} \cup \Sigma^{(c)}  \cup \Sigma^{(j)}$  and, denoting $G_\varphi^{(a)} = G_\varphi \mres \Sigma^{(a)}$, $G_\varphi^{(c)} = G_\varphi \mres \Sigma^{(c)}$, $G_\varphi^{(j)} = G_\varphi \mres \Sigma^{(j)}$,
by \cite[formulas (2) and (16)]{Gia-Mod-Sou-S1} we have on the one hand that $ u_T = (\cos(\varphi),\sin(\varphi))$ a.e.\ and 
\begin{equation}\label{eq:liftedcurrents}
\chi_\# G^{(a)}_\varphi = T^{(a)} \, , \quad \chi_\# G^{(c)}_\varphi = T^{(c)} \, ,\quad \chi_\# G^{(j)}_\varphi = T^{(jc)}\, .
\end{equation}
On the other hand, observe that the Jacobian of $\d \chi \colon \mathrm{Tan}(\de^{-}SG_\varphi,x) \mapsto \RR^4$ equals 1 (indeed~$\d \chi$ maps any pair of orthonormal vectors of $\RR^3$ to a pair of orthonormal vectors in~$\RR^4$).   Hence by the area formula, for $\sigma\in \{a,c,j\}$ we obtain
\begin{align}\label{eq:afterareaformula}
\chi_{\#}G^{(\sigma)}_{\varphi}(\omega)&=\integral{\Sigma^{(\sigma)}}{ \langle \chi^{\#}\omega,\vec{G}_{\varphi}\rangle}{\d\mathcal{H}^2(x,\vartheta)}\nonumber
\\
&=\integral{\chi(\Sigma^{(\sigma)})}{\hspace{0em}\langle \omega(x,y),\hspace{-1.2em} \sum_{(x,\vartheta)\in\chi^{-1}(x,y)} \hspace{-1.2em} \d\chi(x,\vartheta)\vec{G}_{\varphi}(x,\vartheta)\rangle}{\d\mathcal{H}^2(x,y)} \, .
\end{align}
Next, note that for $\sigma\in\{a,c\}$ the map $\chi\colon \Sigma^{(\sigma)}\to\chi(\Sigma^{(\sigma)})$ is one-to-one and for any $(x,\tilde{\varphi}(x))\in\Sigma^{(a)}\cup\Sigma^{(c)}$ we have 
\begin{align}\label{eq:formuladirection}
\d\chi(x,\tilde{\varphi}(x))\vec{G}_{\varphi}(x,\tilde{\varphi}(x))=&-n^3(x,\tilde{\varphi}(x))e_1\wedge e_2\nonumber
\\
&-n^2(x,\tilde{\varphi}(x))\sin(\tilde{\varphi}(x))e_1\wedge \bar e_1+n^2(x,\tilde{\varphi}(x))\cos(\tilde{\varphi}(x))e_1\wedge \bar e_2\nonumber
\\
&+n^1(x,\tilde{\varphi}(x))\sin(\tilde{\varphi}(x))e_2\wedge \bar e_1-n^1(x,\tilde{\varphi}(x))\cos(\tilde{\varphi}(x)e_2\wedge \bar e_2\,.
\end{align}
Since $|n|=1$ we see that $|\d\chi(x,\tilde{\varphi}(x))\vec{G}_{\varphi}(x,\tilde{\varphi}(x))|=1$, too. Moreover, for $\mathcal{H}^2$-a.e. $(x,y)\in \chi(\Sigma^{(\sigma)})$ the vector $\d\chi(x,\tilde{\varphi}(x))\vec{G}_{\varphi}(x,\tilde{\varphi}(x))$ orients the tangent space at $(x,y)$. Hence \eqref{eq:liftedcurrents} and the uniqueness of the representation of i.m.\ rectifiable currents (cf. Section \ref{sec:currentsbasics}) implies $\chi(\Sigma^{(\sigma)})=\mathcal{M}^{(\sigma)}$ up to a $\mathcal{H}^2$-negligible set, $|T^{(\sigma)}|= \mathcal{H}^2\mres \mathcal{M}^{(\sigma)}$, and 
\begin{equation*}
\vec{T}(\chi(x,\tilde{\varphi}(x)))=\d\chi(x,\tilde{u}(x))\vec{G}_{\varphi}(x,\tilde{\varphi}(x))
\end{equation*}
for $\mathcal{H}^2$-almost every $(x,y)=\chi(x,\tilde{\varphi}(x))\in\Sigma^{(\sigma)}$. By the chain rule in $BV$~\cite[Theorem~3.96]{Amb-Fus-Pal} we deduce that 
\begin{align*}
\nabla u_T & = \begin{pmatrix}  - u_T^2 \\ u_T^1 \end{pmatrix}  \otimes \nabla \varphi  \, , 	\quad \DD^{(c)} u_T  = \begin{pmatrix} - \tilde u_T^2 \\ \tilde u_T^1  \end{pmatrix} \otimes \DD^{(c)} \varphi \, .
% \DD^{(j)} u_T  & = \big(\cos(\varphi^+) - \cos(\varphi^-), \sin(\varphi^+) - \sin(\varphi^-)\big) \otimes \nu_\varphi \H^1 \mres J_\varphi\, .
\end{align*}
Combined with the formula for $n$ given by \eqref{eq:normalvector}, the formulas \eqref{eq:ac components} and \eqref{eq:cantor components} then follow from \eqref{eq:formuladirection} by a straightforward calculation.

In order to treat the case $\sigma=j$, note that due to \eqref{eq:normalvector} we have for any $(x,y)=\chi(x,\vartheta)\in\chi(\Sigma^{(j)})$
\begin{equation}\label{eq:formuladirection_jump}
	\begin{split}
		\d\chi(x,\vartheta)\vec{G}_{\varphi}(x,\vartheta) & =  - \, \nu_{\varphi}^2(x)y^2e_1\wedge \bar e_1+\nu_{\varphi}^2(x)y^1e_1\wedge\bar e_2\nonumber \\
		& \hphantom{ = }  \, \, +  \nu_{\varphi}^1(x)y^2e_2\wedge\bar e_1-\nu_{\varphi}^1(x)y^1 e_2\wedge\bar e_2
		\\
		& =:  \xi(x,y) \, .
	\end{split}
\end{equation}
Again $|\xi(x,y)|=1$ and $\xi(x,y)$ orients the tangent space at $\mathcal{H}^2$-a.e.\ $(x,y)\in\chi(\Sigma^{(j)})$. Thus~\eqref{eq:afterareaformula} and the uniqueness of the representation of i.m.\ rectifiable currents imply (up to $\mathcal{H}^2$-negligible sets) that $\mathcal{M}^{(jc)}=\chi(\Sigma^{(j)})$, $\vec{T}=\xi$ on $\mathcal{M}^{(jc)}$, and $|T^{(jc)}|=N(x,y)\mathcal{H}^2\mres \M^{(jc)}$, with
\begin{equation*}
N(x,y)=\#\{\vartheta\in[\varphi^-(x),\varphi^+(x)] \ : \ (\cos(\vartheta),\sin(\vartheta))=y\}\,.
\end{equation*}
To conclude, we have to relate $\mathfrak{m}(x,y)$ to $N(x,y)$ and $\nu_{T}(x)$ to $\nu_{\varphi}(x)$. First note that the proof of the structure theorem (sketched in Remark~\ref{r.differentL}) yields $J_{u_T}\cup\mathcal{L}=J_{\varphi}$ and, combined with the definition of the curves $\gamma_x^T$ (cf.~\eqref{eq:jc part of T}), implies that $\chi[\varphi^-(x),\varphi^+(x)]=\supp(\gamma_{x}^T)$ for $x\in J_{\varphi}$. Hence
\begin{equation} \label{eq:Mjc is a product}
\M^{(jc)}=\chi(\Sigma^{(j)})=\{(x,y)\in\Omega\x\RR^2 \ : \ x\in J_{u_T}\cup\mathcal{L} 	\, ,\ y\in\supp(\gamma_x^T)\}\,.
\end{equation}
Moreover, provided we orient $J_{u_T}$ the same way as $J_{\varphi}$ and $\mathcal{L}$ according to \eqref{eq:localstructureofL},  equation~\eqref{eq:defnormal} also yields $\nu_T=\nu_{\varphi}$ and $\mathfrak{m}(x,y)=N(x,y)$ (a detailed proof of the latter requires to distinguish different cases, which we omit here).\footnote{As noted in Remark~\ref{rmk:choiceoforientation}, the choice $\nu_T(x) = \nu_\varphi(x)$ always yields $\mathfrak{m}(x,y) \geq 0$. The factor $\mathrm{sign}(\mathfrak{m}(x,y))$ in~\eqref{eq:jc components} makes the formula invariant under the change of $\nu_T(x)$.} Inserting this equality in \eqref{eq:formuladirection_jump} concludes the proof of \eqref{eq:jc components}.
\end{proof}

Finally, we recall the following result, proven in~\cite[Section 4]{Gia-Mod-Sou-S1}.
\begin{proposition} \label{prop:supporting BV}
If $u \in BV(\Omega;\SS^1)$, then there exists a $T \in \cart(\Omega \x \SS^1)$ such that $u_T = u$ a.e.\ in $\Omega$.
\end{proposition}

\subsection{Currents associated to discrete spin fields}\label{s.discretecurrents}

 We introduce the piecewise constant interpolations of spin fields. For a set $S$, we define 
\begin{equation*}
\PC_\e(S) := \{u \colon \RR^2 \to S \ : \ u(x) = u(\e i)   \text{ if } x \in \e i + [0, \e)^2 \text{ for some } i \in  \e \ZZ^2 \} 
\end{equation*}
Given $u \colon \Omega_\e \to \SS^1$, we can always identify it with its piecewise constant interpolation belonging to $\PC_\e(\SS^1)$, arbitrarily extended to $\RR^2$. Note that the piecewise constant interpolation of $u$ coincides with $u$ on the bottom-left corners of the squares of the lattice $\e \ZZ^2$.

We associate to $u\in \PC_\e(\SS^1)$ the current $G_{u} \in \D_2(\Omega \x \RR^2)$ defined by
\begin{align}
G_{u}(\phi(x,y) \d x^1 \w  \d x^2) &:= \integral{\Omega}{\phi(x,u(x))}{\d x} \, , \label{eq:Gu ac} \\
G_{u}(\phi(x,y) \d \widehat x^l \w  \d y^m) &:= (-1)^{2-l}\integral{J_{u}}{\bigg\{ \integral{\gamma_x}{ \phi(x,y)}{\d y^m} \bigg\} \nu^l_{u}(x) }{\d \H^1(x)} \, , \label{eq:Gu j}\\
G_{u}(\phi(x,y) \d y^1 \w  \d y^2) &:= 0 \, ,  \label{eq:Gu v}
\end{align}
for every $\phi \in C^\infty_c(\Omega \x \RR^2)$, where $J_{u}$ is the jump set of $u$, $\nu_{u}(x)$ is the normal to $J_{u}$ at $x$, and $\gamma_x \subset \SS^1$ is the (oriented) geodesic arc which connects the two traces $u^-(x)$ and $u^+(x)$. If $u^+(x)$ and $u^-(x)$ are opposite vectors, the choice of the geodesic arc $\gamma_x \subset \SS^1$ is done consistently with the choice made in~\eqref{eq:def of projection} for the values $\Psi(\pi)$ and $\Psi(-\pi)$ as follows: let $\varphi^\pm(x) \in [0,2 \pi)$ be the phase of $u^\pm(x)$; if  $\Psi(\varphi^+(x) - \varphi^-(x)) = \pi$, then $\gamma_x$ is the arc that connects $u^-(x)$ to $u^+(x)$ counterclockwise; if $ \Psi(\varphi^+(x) - \varphi^-(x)) =  -\pi$, then $\gamma_x$ is the arc that connects $u^-(x)$ to~$u^+(x)$ clockwise. Note that the choice of the arc $\gamma_x$ is independent of the orientation of the normal~$\nu_{u}(x)$. 

We define for $\H^1$-a.e.\ $x \in J_{u}$  the integer number $\mathfrak{m}(x) = \pm 1$, where $\pm = +/ -$ if the geodesic arc $\gamma_x$ is oriented counterclockwise/clockwise, respectively. Then
\begin{equation} \label{eq:from dym to H1 2}
\integral{ \gamma_x}{\phi(x,y)}{\d y^m} = (-1)^m \, \mathfrak{m}(x) \integral{  \supp(\gamma_x)}{\phi(x,y) \widehat y^m}{\d \H^1(y)}  \, .
\end{equation}
 
\begin{figure}[H]
\scalebox{0.7}{
    \includegraphics{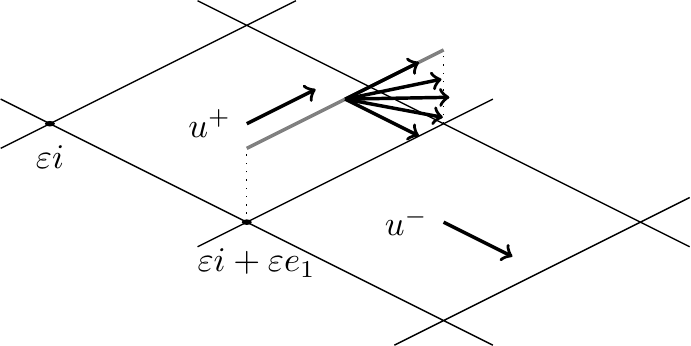}
% \begin{tikzpicture}[scale=1,y={(-1cm,0.5cm)},x={(1cm,0.5cm)}, z={(0cm,1cm)}]

% \coordinate (O) at (0, 0, 0);

% \draw (-0.5,0,0) -- (2+0.5,0,0)
%       (2,0,0) -- (2,2+0.5,0)
%       (2+0.5,2,0) -- (-0.5,2,0)
% 	  (0,2+0.5,0) -- (0,0,0);
% \draw (2,0,0) -- (2,-2-0.5,0)
%       (2+0.5,-2,0) -- (-0.5,-2,0)
% 	  (0,-2-0.5,0) -- (0,0,0);
% \draw[->,line width=1pt] (1,1,0) -- (1+0.7,1,0);
% \draw[->,line width=1pt] (1,-1,0) -- (1,-1-0.7,0);

% \draw (1,1,0) node[anchor=east] {$u^+$};
% \draw (1,-1,0) node[anchor=east] {$u^-$};

% \draw[fill=black] (0,0,0) circle(0.03);
% \draw[fill=black] (0,2,0) circle(0.03);
% \draw (0,2,-0.1) node[anchor=north] {$\varepsilon i$};
% \draw (0.1,0,-0.2) node[anchor=north] {$\varepsilon i + \varepsilon e_1$};

% \draw[dotted] (0,0,0) -- (0,0,0.75)
%               (2,0,0) -- (2,0,0.75);
			  
% \draw[color=gray, line width=1pt] (0,0,0.75) -- (2,0,0.75);

% \begin{scope}[xshift=1cm, yshift=0.5cm+0.75cm]
% \draw[->, line width=1pt] (0,0,0) -- (-22.5:0.75);
% \draw[->, line width=1pt] (0,0,0) -- (-43:0.75);
% \draw[->, line width=1pt] (0,0,0) -- (-65.5:0.75);
% \draw[->, line width=1pt] (0,0,0) -- (0:0.75);
% \draw[->, line width=1pt] (0,0,0) -- (-90:0.75);
% \end{scope}

% \end{tikzpicture}
}
\caption{The current $G_{u}$ has vertical parts concentrated on the jump set $J_{u}$, where a transition from $u^-$ to $u^+$ occurs.}
\end{figure}
In the proposition below we characterize the current $G_u$ associated to a discrete spin field in terms of the decomposition $G_u=\vec{G}_u|G_u|$.
\begin{proposition} \label{prop:Gu is im}
Let $u \in \PC_\e(\SS^1)$ and let $G_{u} \in \D_2(\Omega \x \RR^2)$ be the current defined in \eqref{eq:Gu ac}--\eqref{eq:Gu v}. Then $G_{u}$ is an i.m.\ rectifiable current and, according to the representation formula~\eqref{eq:representation}, $G_{u} = \vec{G}_u   |G_{u}|$, where $|G_{u}| = \H^2 \mres \M$,
\begin{equation*}
\M = \M^{(a)} \cup \M^{(j)} = \{(x,u(x)) \ : \ x \in \Omega \sm J_{u} \} \cup  \{(x,y) \ : \ x \in J_{u}, \ y \in \gamma_x \} \, ,
\end{equation*}
and 
\begin{equation} \label{eq:orientation of Gu 1}
\begin{split}
\vec{G}_{u}(x,y) =  e_1  \wedge e_2
\end{split}
\end{equation}
for $\H^2$-a.e.\ $(x,y) \in \M^{(a)}$ and
\begin{equation} \label{eq:orientation of Gu 2}
\begin{split}
\vec{G}_{u}(x,y) = \mathrm{sign}(\mathfrak{m}(x)) \big[  - & \  \nu^2_{u}(x) y^2 e_1 \wedge \bar e_1 +  \nu^2_{u}(x) y^1 e_1 \wedge \bar e_2  \\
 + & \ \nu^1_{u}(x) y^2 e_2 \wedge \bar e_1 - \nu^1_{u}(x) y^1 e_2 \wedge \bar e_2 \big ] \, 
\end{split}
\end{equation}
for $\H^2$-a.e.\ $(x,y) \in \M^{(j)}$.
\end{proposition}
\begin{proof}
 First note that the set $\M$ is countably $\mathcal{H}^2$-rectifiable. Since $u$ is piecewise constant, for horizontal forms we have
\begin{equation*}
G_{u}(\phi(x,y) \d x) = \integral{\Omega}{\phi(x,u(x))}{\d x}  = \integral{\Omega \x \RR^2}{\phi(x,y)}{\d \H^2\mres \M^{(a)}(x,y)} \, .
\end{equation*}
By~\eqref{eq:from dym to H1 2} we deduce that for $l,m=1,2$
\begin{equation*}
\begin{split}
G_{u}(\phi(x,y) \d \widehat x^l \w \d y^m) & = (-1)^{2-l}\integral{J_{u}}{\bigg\{ \integral{\gamma_x}{ \phi(x,y)}{\d y^m} \bigg\} \nu^l_{u}(x) }{\d \H^1(x)} \\
& = (-1)^{2-l+m}\integral{J_{u}}{\bigg\{ \integral{\supp(\gamma_x)}{ \phi(x,y) \widehat y^m}{\d \H^1(y)} \bigg\} \nu^l_{u}(x) \mathfrak{m}(x)}{\d \H^1(x)} \\
& = (-1)^{2-l+m}\integral{\Omega \x \RR^2}{\phi(x,y) \widehat y^m \nu^l_{u}(x)\mathfrak{m}(x) }{\d \H^2\mres \M^{(j)}(x,y)}  \, .
\end{split}
\end{equation*}
Then for every $\omega \in \D_2(\Omega \x \RR^2)$ we have
\begin{equation*}
G_{u}(\omega) = \integral{\Omega \x \RR^2}{\langle \omega, \vec{G}_{u} \rangle}{ \d  \H^2 \mres \M} \, 
\end{equation*}
for $\vec{G}_{u}$ defined as in~\eqref{eq:orientation of Gu 1}--\eqref{eq:orientation of Gu 2} and moreover $\vec{G}_u(x,y)$ is associated to the tangent space at $(x,y)\in\mathcal{M}$. Since also $|\vec{G}_{u}(x,y)| = 1$ for $|G_{u}|$-a.e.\ $(x,y) \in \Omega \x \RR^2$, we conclude the proof.
\end{proof}
The next proposition is crucial since it relates the boundary of the current $G_u$ associated to a discrete spin field to the vorticity measure $\mu_u$.
\begin{proposition} \label{prop:bd of Gu is mu}
Let $u \in \PC_\e(\SS^1)$ and let $G_{u} \in \D_2(\Omega \x \RR^2)$ be the current defined in \eqref{eq:Gu ac}--\eqref{eq:Gu v}. Then 
\begin{equation*}
\de G_{u}|_{\Omega \x \RR^2} = - \mu_{u} \x \llbracket \SS^1 \rrbracket  \, ,
\end{equation*}
where $\mu_{u}$ is the discrete vorticity measure defined in~\eqref{eq:discrete vorticity measure} for $u|_{\e \ZZ^2} \colon \e \ZZ^2 \to \SS^1$.
\end{proposition}
\begin{proof}
Let us fix $0 < \rho < \min\{\e/4, \dist(\Omega_\e, \de \Omega)\}$ and $\eta \in \D^1(\Omega \x \RR^2)$. With a partition of unity we can split $\eta$ into the sum of 1-forms depending on their supports. We discuss here all the possibilities for the supports. 

{\ul{Case 1}}: $\supp(\eta) \subset (\e i + (0,\e)^2) \x\RR^2$ for some $ i \in \ZZ^2$. Since $u$ is constant in $(\e i + (0,\e)^2)$, we get automatically $\de G_{u} (\eta) = 0$ by Remark~\ref{rmk:bd of smooth}.

{\ul{Case 2}}: Let $H$ be the side of the square $\e i + [0,\e]^2$ connecting two vertices $p, q \in \e \ZZ^2$ and let $U$ be the $\rho/2$-neighborhood of $H \sm \big(B_\rho(p) \cup B_\rho(q)\big)$. Assume that $\supp(\eta) \subset U \x \RR^2$. We claim that 
\begin{equation} \label{eq:bd of Gu close to side}
\de G_{u}(\eta) = 0 \, .
\end{equation}
To prove this, we approximate the pure-jump function $u$ by means of a sequence of Lipschitz functions $u_j$. Let $u^\pm$ be the traces of $u$ on the two sides of $H$ and let $\nu_H$ be the normal to~$H$ oriented as $\nu_{u}$. We let $\widehat \varphi^\pm \in [0,2 \pi)$ be the phases of $u^\pm$ defined by $u^\pm = \exp(\iota \widehat \varphi^\pm)$. We set $\varphi^- := \widehat \varphi^-$ and $\varphi^+ := \widehat \varphi^- + \Psi(\widehat \varphi^+ - \widehat \varphi^-) \in (-\pi, 3\pi)$, where $\Psi$ is the function given by~\eqref{eq:def of Psi}. We then define
\begin{equation*}
\varphi(t) := \begin{cases}
\varphi^- \, , & \text{ if } t \leq - \tfrac{1}{2}\\
\varphi^- + \big(\varphi^+ - \varphi^-\big)(t+\tfrac{1}{2})  & \text{ if } -\tfrac{1}{2} < t < \tfrac{1}{2} \\
\varphi^+ & \text{ if } t \geq  \tfrac{1}{2} \, ,
\end{cases}
\end{equation*}
and $\varphi_k(s) := \varphi(k s)$ for $k$ large enough. Note that the curve $t \in (-1/2,1/2) \mapsto \exp(\iota \varphi(t))$ parametrizes the geodesic arc $\gamma_{\pm} \subset \SS^1$ which connects $u^-$ to $u^+$, consistently with the choice done in formula~\eqref{eq:Gu j}. Then we put
\begin{equation*}
u_k(x) := \exp\big( \iota \varphi_k(\nu_H \cdot  (x-p)) \big) \quad \text{for } x \in U \, .
\end{equation*}
We prove that $G_{u_k} \weak G_{u}$ in $\D_2(U \x \RR^2)$. Let us fix $\phi \in C^\infty_c(U \x \RR^2)$. Since $u_k \to u$ in measure, we have that
\begin{equation*}
G_{u_k}(\phi(x,y) \d x) = \integral{U}{\phi(x,u_k(x))}{\d x} \to \integral{U}{\phi(x,u(x))}{\d x} = G_{u}(\phi(x,y) \d x) \, .
\end{equation*} 
Writing $x \in U$ as $x = x' + s \nu_H$ with $x' \in H$, $s \in \RR$, for $l=1,2$ we further obtain that 
\begin{equation*}
\begin{split}
& G_{u_k}(\phi(x,y) \d \widehat x^l \w \d y^1) = (-1)^{2-l}\integral{U}{\phi(x,u_k(x)) \de_{x_l} u_k^1(x)}{\d x} \\
& \quad = (-1)^{3-l}\integral{U}{\phi(x,u_k(x)) \sin(\varphi_k(\nu_H \cdot  (x-p))) \varphi'_k(\nu_H \cdot (x-p)) \nu_H^l}{\d x}  \\
&\quad = (-1)^{3-l}\integral{H}{\bigg\{ \int \limits_{-1/2k}^{1/2k} \phi\Big(x' + s \nu_H,\exp\big(\iota \varphi_k(s)\big)\Big) \sin\big(\varphi_k(s)\big) \varphi'_k(s) \d s \bigg\} \nu_H^l }{\d \H^1(x')} \\
&\quad = (-1)^{3-l}\integral{H}{\bigg\{ \int \limits_{-1/2}^{1/2} \phi\Big(x' + \tfrac{t}{k} \nu_H,\exp\big(\iota \varphi(t)\big)\Big) \sin\big(\varphi(t)\big) \varphi'(t) \d t \bigg\} \nu_H^l }{\d \H^1(x')} \\
&\quad = (-1)^{2-l}\integral{H}{\bigg\{ \int \limits_{\gamma_{\pm}} \phi\Big(x' + \tfrac{t}{k} \nu_H,y\Big) \d y^1 \bigg\} \nu_H^l }{\d \H^1(x')} \\
&\quad \to (-1)^{2-l}\integral{H}{\bigg\{ \int \limits_{\gamma_{\pm}} \phi\big(x',y\big) \d y^1 \bigg\} \nu_H^l}{\d \H^1(x')}  = G_{u}(\phi(x,y) \d \widehat x^l \w \d y^1) \, ,
\end{split}
\end{equation*}
where $\gamma_{\pm} \subset \SS^1$ is the geodesic arc connecting $u^-$ to $u^+$. With analogous computations one proves $G_{u_k}(\phi(x,y) \d \widehat x^l \w \d y^2) \to G_{u}(\phi(x,y) \d \widehat x^l \w \d y^2)$. %Finally,  $G_{u_k}(\phi(x,y) \d y^1 \w \d y^2) = 0$ as~$\det(\nabla u_k) = 0$.

Hence, due to Stokes' Theorem we have that
\begin{equation*}
0 = \de G_{u_k}(\eta) = G_{u_k}(\! \d \eta) \to G_{u}(\! \d \eta) = \de G_{u}(\eta) \, ,
\end{equation*}
which proves~\eqref{eq:bd of Gu close to side}.

{\ul{Case 3}}: $\supp(\eta) \subset B_\rho(p) \x \RR^2$, where $p = \e i + \e e_1 + \e e_2$ for some $i \in \ZZ^2$. In this case we will approximate the current $G_{u}$ with graphs of a sequence of functions $u_k$ which are Lipschitz outside the point $p$. For notation simplicity we let $\widehat \varphi_1, \widehat \varphi_2, \widehat \varphi_3, \widehat \varphi_4 \in [0,2\pi)$ be the phases defined by the relations 
\begin{equation} \label{eq:phases of vortex}
\begin{aligned}
& u(\e i + \e e_1 + \e e_2) =: u_1 = \exp(\iota \widehat \varphi_1)\, , && u(\e i + \e e_2) =: u_2 = \exp(\iota \widehat \varphi_2) \, ,\\
&  u( \e i) =: u_3 = \exp(\iota \widehat \varphi_3) \, ,  && u( \e i+\e e_1 ) =: u_4 =  \exp(\iota \widehat \varphi_4) \, .
\end{aligned}
\end{equation}
We define the auxiliary angles 
\begin{equation} \label{eq:auxiliary angles}
\widetilde{\varphi}_{\sigma(h+1)} := \widehat \varphi_{\sigma(h)} + \Psi(\widehat \varphi_{\sigma(h + 1)} - \widehat \varphi_{\sigma(h)})  \, ,
\end{equation}
for $h = 1,2,3,4$, where $\sigma(h) \in \{1,2,3,4\}$ is such that $\sigma(h) \equiv h$ $\mathrm{mod} \ 4$ (the term $\Psi(\widehat \varphi_{\sigma(h + 1)} - \widehat \varphi_{\sigma(h)})$ is the oriented angle in $[-\pi,\pi]$ between the two vectors $u_{\sigma(h)}$ and $u_{\sigma(h+1)}$). We introduce the $2\pi$-periodic function $\varphi_k \colon \RR \to \RR$
\begin{equation*}
\varphi_k(\vartheta) :=  \begin{cases}
\widehat \varphi_{\sigma(h)} \, , & \text{if } -\tfrac{\pi}{4}+h \tfrac{\pi}{2} < \vartheta \leq h \tfrac{\pi}{2} - \tfrac{1}{2k} \, , \\
\widehat \varphi_{\sigma(h)} + k\big(\widetilde{\varphi}_{\sigma(h+1)} - \widehat \varphi_{\sigma(h)}\big)(\vartheta - h \tfrac{\pi}{2} + \tfrac{1}{2k} ) \, , & \text{if } h \tfrac{\pi}{2} - \tfrac{1}{2k} < \vartheta < h \tfrac{\pi}{2} + \tfrac{1}{2k} \\
\widetilde{\varphi}_{\sigma(h+1)} \, , & \text{if } h \tfrac{\pi}{2} + \tfrac{1}{2k} \leq \vartheta < \tfrac{\pi}{4}+h \tfrac{\pi}{2} \, .
\end{cases}
\end{equation*}
for $\vartheta \in \big(-\tfrac{\pi}{4}+h \tfrac{\pi}{2},\tfrac{\pi}{4}+h \tfrac{\pi}{2})$, $h \in \ZZ$. The function $\varphi_k$ might have jumps at the points $\tfrac{\pi}{4}+h \tfrac{\pi}{2}$, $h \in \ZZ$; note, however, that according to~\eqref{eq:def of projection} the amplitude of the jump is given  by
\begin{equation*}
\widehat \varphi_{\sigma(h+1)} - \widetilde\varphi_{\sigma(h+1)} = \widehat \varphi_{\sigma(h+1)} - \widehat \varphi_{\sigma(h)} - \Psi(\widehat \varphi_{\sigma(h+1)} - \widehat \varphi_{\sigma(h)}) = Q(\widehat \varphi_{\sigma(h+1)} - \widehat \varphi_{\sigma(h)}) \in 2 \pi \ZZ \, .
\end{equation*}
\begin{figure}[H]
\scalebox{0.7}{
    \includegraphics{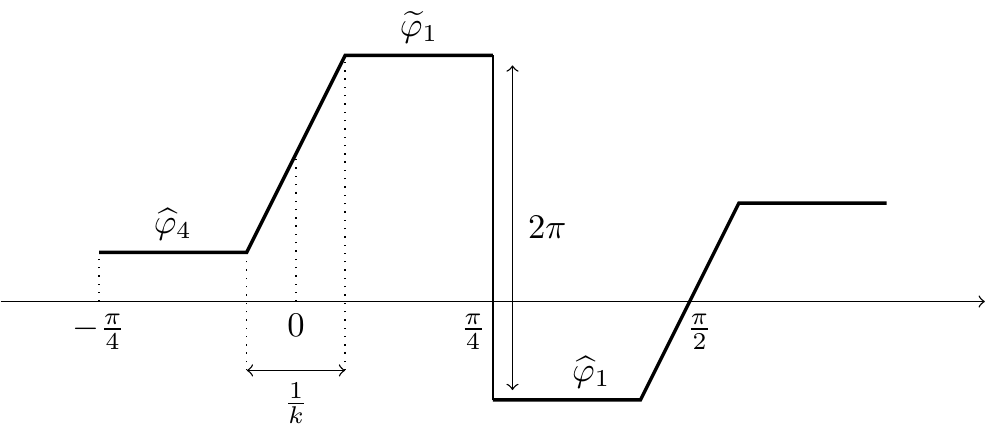}
% \begin{tikzpicture}
% \draw[->] (-3,0) -- (7,0);
% %\draw (0,0) -- (0,5);
% \draw[line width=1pt] (-2,0.5)--(-0.5,0.5) -- (0.5,2.5) -- (2,2.5);
% \draw (2,2.5) -- (2,-1);
% \draw[line width=1pt] (2,-1) -- (3.5,-1) -- (4.5,1) -- (6,1);
% \draw[dotted] (0,0) -- (0,1.5);
% \draw (0,0) node[anchor=north] {$0$};
% \draw (4.1,0) node[anchor=north] {$\frac{\pi}{2}$};
% \draw[dotted] (-2,0) -- (-2,0.5);
% \draw (-2,0) node[anchor=north] {$-\frac{\pi}{4}$}; 
% \draw (1.8,0) node[anchor=north] {$\frac{\pi}{4}$};
% \draw[<->] (-0.5,-0.7) --(0.5,-0.7);
% \draw[dotted] (-0.5,-0.7) -- (-0.5,0.5);
% \draw[dotted] (0.5,-0.7) -- (0.5,2.5);
% \draw (0,-0.7) node[anchor=north] {$\frac{1}{k}$};
% \draw (-1.25,0.5) node[anchor=south] {$\widehat{\varphi}_4$};
% \draw (1.25,2.5) node[anchor=south] {$\widetilde \varphi_1$};
% \draw (3,-1) node[anchor=south] {$\widehat{\varphi}_1$};
% \draw[<->] (2.2,2.4) -- (2.2,-1+0.1);
% \draw (2.2,0.75) node[anchor=west] {$2 \pi$};
% \end{tikzpicture}
}
\caption{Example for the definition of $\varphi_k$ for $h = 0$.}
\end{figure}

We now define a map $v_k \colon \SS^1 \to \SS^1$. Given $y \in \SS^1$, let $\vartheta(y) \in [0, 2\pi)$ be the angle such that $y = \exp(\iota \vartheta(y))$ and set 
\begin{equation*}
v_k(y) := \exp\big(\iota \varphi_k(\vartheta(y))\big) \, .
\end{equation*} 
The definition actually does not depend on the choice of the phase $\vartheta(y)$, due to the $2\pi$-periodicity of~$\varphi_k$. Thus we could also choose $\vartheta(y) \in [2\pi h, 2\pi(h+1))$ for any $h \in \ZZ$. Note that $v_k$ is continuous: indeed the possible jumps of $\varphi_k$ have amplitude in $2 \pi \ZZ$, and thus are not seen by $v_k$. In particular, we can compute the degree of the map $v_k$ via the formula 
\begin{equation*}
\begin{split}
\deg(v_k) 2 \pi &= \deg(v_k) \int \limits_{\SS^1} \omega_{\SS^1}  = \int \limits_{\SS^1} v_k^\# \omega_{\SS^1} = \sum_{h=0}^3 \widetilde\varphi_{\sigma(h+1)} - \widehat \varphi_{\sigma(h)} \\
& = \sum_{h=0}^3 \Psi(\widehat \varphi_{\sigma(h+1)} - \widehat \varphi_{\sigma(h)}) = d_{u}(\e i) 2\pi \,,
\end{split}
\end{equation*}
where $\omega_{\SS^1}$ is the volume form on $\SS^1$ and  $d_{u}(\e i)$ is the discrete vorticity defined in~\eqref{eq:discrete vorticity}.

We now define the map $u_k \colon B_\rho(p) \to \SS^1$ by
\begin{equation*}
u_k(x) := v_k \big(\tfrac{x-p}{|x-p|} \big) \, .
\end{equation*}
Note that, if $(r,\vartheta)$ are polar coordinates for the point $x-p$, then the polar coordinates of $u_k(x)$ are $(1, \varphi_k(\vartheta))$. 

\begin{figure}[H]
\scalebox{0.7}{
    \includegraphics{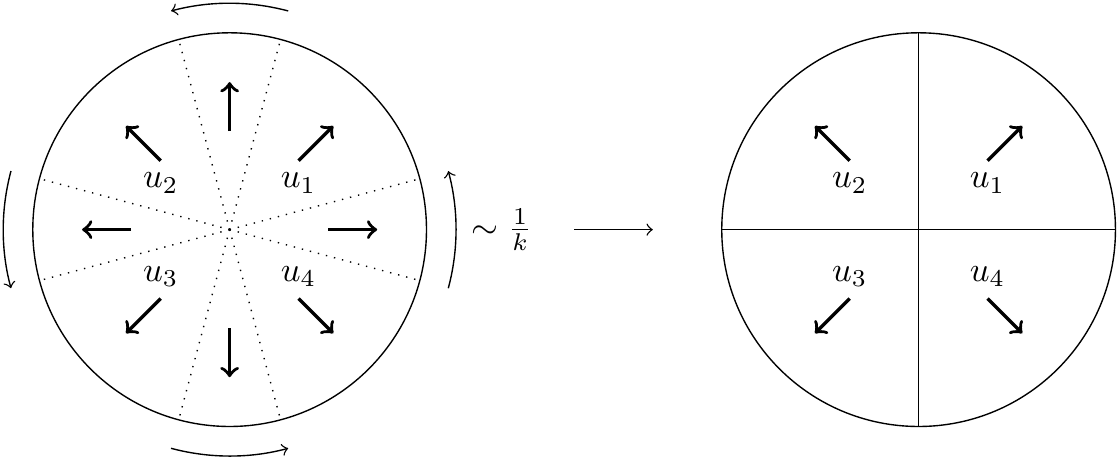}
}
\caption{Example of the approximation $u_k$ (on the left) of the function $u$ (on the right). The jump set of the function $u$ is expanded and a transition between the jumps of $u$ is constructed using the geodesic arcs in $\SS^1$ between the traces. If $u$ has a nontrivial discrete vorticity as in the picture, then the graph $G_{u_k}$ of the function $u_k$ has a hole in the center, as it happens for the graph of the map $x \mapsto \frac{x}{|x|}$. The hole is then preserved in the passage to limit to $G_{u}$, see formula~\eqref{eq:hole}.}
\end{figure}

By \cite[3.2.2, Example 2]{Gia-Mod-Sou-I} we get that
\begin{equation} \label{eq:hole}
\de G_{u_k}|_{B_\rho(p) \x \RR^2} = - \deg(v_k) \delta_{p} \x \llbracket \SS^1 \rrbracket = - d_{u}(\e i) \delta_{p} \x \llbracket \SS^1 \rrbracket = - \mu_{u} \x \llbracket \SS^1 \rrbracket|_{B_\rho(p) \x \RR^2} \, .
\end{equation}
Therefore, to conclude the proof it suffices to show the convergence $G_{u_k} \weak G_{u}$ in $\D_2(\Omega \x \RR^2)$, so that 
\begin{equation*}
- \mu_{u} \x \llbracket \SS^1 \rrbracket (\eta) = \de G_{u_k}(\eta) \to \de G_{u}(\eta) \, .
\end{equation*}
To do so, let us fix $\phi \in C^\infty_c(B_\rho(p) \x \RR^2)$. Since $u_k \to u$ in measure, we have that
\begin{equation*}
G_{u_k}(\phi(x,y) \d x) = \integral{B_\rho(p)}{\phi(x,u_k(x))}{\d x} \to \integral{B_\rho(p)}{\phi(x,u(x))}{\d x} = G_{u}(\phi(x,y) \d x) \, .
\end{equation*} 
To compute the limit on forms of the type $\phi(x,y) \d \widehat x^l \w \d y^m$, observe that $u_k$ is not constant only in the 4 sectors of $B_\rho(p)$ given in polar coordinates by 
\begin{equation*}
A_k^h  -p  := \big\{(r,\vartheta) \ : \ r \in (0,\rho), \ \vartheta \in \big(h \tfrac{\pi}{2} - \tfrac{1}{2k}, h \tfrac{\pi}{2} + \tfrac{1}{2k}\big) \big\} \, , \quad h \in \{0,1,2,3\} \, ,
\end{equation*}
thus, for $l,m=1,2$,
\begin{equation*}
\begin{split}
 (-1)^{2-l} G_{u_k}(\phi(x,y) \d \widehat x^l \w \d y^m)& = \integral{B_\rho(p)}{\phi(x, u_k(x)) \de_{x^l} u^m_k(x)}{\d x} 
\\
& = \sum_{h=0}^3 \integral{A_k^h}{\phi(x, u_k(x)) \de_{x^l} u^m_k(x)}{\d x} \, .
\end{split} 
\end{equation*}
The integrals on the sets $A_k^h$ can be computed in polar coordinates. We show the computations for $h=0$ and $m=1$, the other cases being analogous. Changing variables in the integral on the interval $(-1/2k,1/2k)$ we obtain  
\begin{equation*}
\begin{split}
& \integral{A_k^0}{\phi(x, u_k(x))  \de_{x^2} u^1_k(x)}{\d x} \\
&\quad = - \int \limits_0^\rho \int \limits_{-1/2k}^{1/2k} \phi\big(p+ r \exp\big(\iota \vartheta\big), \exp\big(\iota \varphi_k(\vartheta)\big) \big)\sin(\varphi_k(\vartheta)) \varphi'_k(\vartheta)\cos(\vartheta) \d \vartheta  \d r \\
& \quad = - \int \limits_0^\rho \int \limits_{-1/2}^{1/2} \phi\big(p+ r \exp\big(\iota \tfrac{t}{k}\big), \exp(\iota\varphi_1(t)) \big) \sin(\varphi_1(t)) \varphi'_{1}(t)\cos\big(\tfrac{t}{k}\big) \d t  \d r \\
& \quad \to - \int \limits_0^\rho \int \limits_{-1/2}^{1/2} \phi\big(p+(r,0), \exp(\iota\varphi_1(t))\big)\sin(\varphi_1(t)) \varphi'_{1}(t)  \d t  \d r 
\\
& \quad = \int \limits_{J_{41}} \bigg\{ \int \limits_{\gamma_{41}} \phi(x, y)  \d y^1 \bigg\} \nu^2 \d \H^1(x) \, ,
\end{split}
\end{equation*}
where $t \in (-1/2,1/2) \mapsto \gamma_{41}(t) := \exp\big(\iota (\widehat \varphi_4 + (\widetilde\varphi_1 - \widehat \varphi_4)(t + \tfrac{1}{2}))$ is a parametrization of the geodesic arc $\gamma_{41} \in \SS^1$ which connects $u_4$ to $u_1$ (cf.\ the definition of~$\widehat \varphi_4$ in~\eqref{eq:phases of vortex} and of $\widetilde\varphi_1$ in~\eqref{eq:auxiliary angles}) and $J_{41}$ is the subset of $J_{u} \cap B_\rho(p)$ where $u$ jumps from $u_4$ to $u_1$, oriented with normal $\nu = (0,1)$. Moreover,
\begin{equation*}
\begin{split}
\integral{A_k^0}&{\phi(x, u_k(x))  \de_{x^1} u^1_k(x)}{\d x}
\\
&=\int \limits_0^\rho \int \limits_{-1/2k}^{1/2k} \phi\big(p+ r \exp\big(\iota \vartheta\big), \exp\big(\iota \varphi_k(\vartheta)\big) \big)\sin(\varphi_k(\vartheta)) \varphi'_k(\vartheta)\sin(\vartheta) \d \vartheta  \d r \\
& = \int \limits_0^\rho \int \limits_{-1/2}^{1/2}\phi\big(p+  r \exp\big(\iota \tfrac{t}{k}\big), \exp(\iota\varphi_1(t)) \big) \sin(\varphi_1(t)) \varphi'_{1}(t) \sin\big(\tfrac{t}{k}\big) \d t  \d r  \to  0 \, .
\end{split}
\end{equation*}
This concludes the proof.
\end{proof}

In the elementary lemma below we show that the flat convergence of the vorticity measure implies convergence of the boundaries of the graphs associated to the corresponding spin field.

\begin{lemma} \label{lemma:flat implied D1}
Let $\mu_\e, \mu  \in \M_b(\Omega)$ and assume that $\mu_\e \flat \mu$ in $\Omega$. Then $\mu_\e \x \llbracket \SS^1 \rrbracket \weak \mu \x \llbracket \SS^1 \rrbracket$ in $\D_1(\Omega \x \RR^2)$.
\end{lemma}
\begin{proof}
Let us fix $\phi \in C^\infty_c(\Omega\x \RR^2)$. For $l =1,2$, by the very definition of the product of a $0$-current and a $1$-current we infer
\begin{equation*}
\mu_\e \x \llbracket \SS^1 \rrbracket (\phi(x,y) \d x^l) = 0 = \mu \x \llbracket \SS^1 \rrbracket(\phi(x,y) \d x^l) \, .
\end{equation*}
Next note that $\psi^m(x) := \int_{ \SS^1} \phi(x,y) \d y^m$ belongs to $C^{0,1}_c(\Omega)$ for $m=1,2$. Hence
\begin{equation*}
\begin{split}
\mu_\e \x \llbracket \SS^1 \rrbracket (\phi(x,y) \d y^m) & = \integral{\Omega}{ \llbracket \SS^1 \rrbracket (\phi(x,y) \d y^m) }{\d \mu_\e(x)} = \integral{\Omega}{ \bigg\{ \integral{\SS^1}{\phi(x,y)}{\d y^m} \bigg\} }{\d \mu_\e(x)} \\
& = \integral{\Omega}{ \psi^m(x) }{\d \mu_\e(x)} \to \integral{\Omega}{ \psi^m(x) }{\d \mu(x)}.
\end{split}
\end{equation*}
\end{proof}

\section{A compactness result}\label{sec:compactness}

In this section we prove a general compactness result that includes the statement in Theorem~\ref{thm:theta equal e log}-{\em i)} but can be also applied in other regimes, as in~\cite{CicOrlRuf}.  For this reason, in each result we give precisely the assumptions on $\theta_{\e}$ for which the statements hold true. The notation $\theta_{\e}\lesssim\e|\log\e|$ stands for $\lim_{\e \to 0} \frac{\theta_\e}{\e |\log \e|} \in [0,+\infty)$.  Given a measure $\mu = \sum_{h=1}^M d_h \delta_{x_h}$ and an open set $A$, we adopt the notation 
\begin{equation*}
A_\mu := A \sm \supp(\mu) = A \sm \{x_1, \dots, x_M\}
\end{equation*}
and $A_\mu^\rho := A\sm \bigcup_{h=1}^M B_\rho(x_h)$.

Our first goal is to prove a compactness result for the graphs $G_{u_{\e}}$ in the class of i.m.\ rectifiable currents. To state the result, given $\mu = \sum_{h=1}^M d_h \delta_{x_h}$ with $d_h \in \ZZ$ and~$u \in BV(\Omega;\SS^1)$, we introduce the set of admissible currents  
\begin{equation} \label{eq:def of Adm}
	\begin{split}
		\mathrm{Adm}(\mu,u;\Omega)  :=  \bigg\{ & T \in \D_2(\Omega \x \RR^2)\ : \ T \in \cart(\Omega_\mu \x \SS^1)\, , \\[-1em]
		& \hspace{2em} \de T|_{\Omega \x \RR^2} = - \mu \x \llbracket \SS^1 \rrbracket \, , \ u_T = u \text{ a.e.\ in } \Omega  \bigg\}.
	\end{split}
\end{equation} 

This is the main result in this section. 

\begin{proposition}[Compactness in the sense of currents] \label{prop:current compactness}
	Assume that $u_\e \colon \e \ZZ^2 \to \SS^1$ satisfies $\frac{1}{\e\theta_{\e}}E_{\e}(u_{\e})\leq C$ with $\theta_{\e}\lesssim\e|\log\e|$. Let $G_{u_\e} \in \D_2(\Omega \x \RR^2)$ be the current associated to $u_{\e}$ as in~\eqref{eq:Gu ac}--\eqref{eq:Gu v} and let~$\mu_{u_{\e}}$ the discrete vorticity measure associated to $u_{\e}$ as in~\eqref{eq:discrete vorticity measure}. Then there exists a subsequence (not relabeled) and
	\begin{itemize}
		\item[i)] $\mu = \sum_{h=1}^M d_h \delta_{x_h}$ with $d_h \in \ZZ$ such that $\mu_{u_\e} \flat \mu$;
		\item[ii)] $u \in BV(\Omega;\SS^1)$ such that $u_\e \to u$ in $L^1(\Omega;\RR^2)$ and $u_\e \wstar u$ in $BV_{\mathrm{loc}}(\Omega;\RR^2)$; 
		\item[iii)] $T \in \mathrm{Adm}(\mu,u;\Omega) $ such that  $G_{u_\e} \weak T$ in $\D_2(\Omega \x \RR^2)$.
	\end{itemize}
	In particular, if $\theta_\e \ll \e |\log \e|$, then $\mu = 0$ and $T \in \cart(\Omega \x \SS^1)$. 
	\end{proposition}

We postpone the proof, since we need some preliminary results. To deduce a bound on the mass $|G_{u_\e}|$ (and thus compactness in $\D_2(\Omega \x \RR^2)$), we rewrite the energy as a parametric integral of the currents $G_{u_\e}$. Specifically, defining the convex and positively $1$-homogeneous function $\Phi \colon \Lambda_2 (\RR^2 \x \RR^2) \mapsto \RR$ by
\begin{equation} \label{eq:Phi is 2,1}
\Phi(\xi) := \sqrt{(\xi^{21})^2 + (\xi^{22})^2 }  + \sqrt{(\xi^{11})^2 + (\xi^{12})^2 }
\end{equation}
for every
\begin{equation*}
\xi = \xi^{\ol 0 0} e_1 \wedge e_2 + \xi^{21} e_1 \wedge \bar e_1 + \xi^{22} e_1 \wedge \bar e_2 + \xi^{11} e_2 \wedge \bar e_1 + \xi^{12} e_2 \wedge \bar e_2 + \xi^{0 \ol 0} \bar e_1 \wedge \bar e_2  \in \Lambda_2 (\RR^2 \x \RR^2) \, ,
\end{equation*}
we have the following representation proven in~\cite[Lemma 4.3]{CicOrlRuf}.  
\begin{lemma} \label{lemma:bound with parametric integral}
Assume that $\theta_\e \ll 1$. Let $\sigma \in (0,1)$ and $A \compact \Omega$. Then for $\e$ small enough
\begin{equation*}
 \frac{1}{\e \theta_\e}E_\e(u_\e) \geq (1-\sigma)\integral{J_{u_\e} \cap A}{\geo(u_\e^-,u_\e^+)|\nu_{u_\e}|_1}{\d \H^1} = (1-\sigma) \integral{A \x \RR^2}{\Phi(\vec G_{u_\e})}{\d |G_{u_\e}|} \, .
\end{equation*}
\end{lemma}

One of the features of the limit current $T$ is that it is an i.m.\ rectifiable current. This will follow from the Closure Theorem~\cite[2.2.4, Theorem 1]{Gia-Mod-Sou-I}. However, we first need a technical lemma to circumvent the fact that, in general, the masses $|\de G_{u_\e}|$ are not equibounded. By Proposition~\ref{prop:bd of Gu is mu} the boundaries $\de G_{u_\e}$ are indeed related to the vorticity~$\mu_{u_\e}$, which thanks to the well-known ball construction is equivalent to a sequence of measures with equibounded masses. The precise statement suited for our purposes is the following.

\begin{lemma} \label{lemma:modification}
	Assume that $u_\e \colon \e \ZZ^2 \to \SS^1$ satisfies $\frac{1}{\e^2 |\log \e|} XY_\e(u_\e;\Omega) \leq C$. Let $G_{u_\e} \in \D_2(\Omega \x \RR^2)$ be the current associated to $u_\e$ defined as in~\eqref{eq:Gu ac}--\eqref{eq:Gu v}. Let $\Omega' \subset \subset \Omega$. Then there exist a subsequence (not relabeled), points $z_1,\dots,z_J \in \Omega'$, and $\overline u_\e \colon \e \ZZ^2 \to \SS^1$ such that
	\begin{itemize}
		\item[(i)] $G_{\overline u_\e} - G_{u_\e} \weak 0$ in $\D_2(\Omega' \x \RR^2)$;
		\item[(ii)] $\sup_\e  |G_{\overline u_\e}|(\Omega'\x \RR^2) \leq \sup_\e  |G_{u_\e}|(\Omega'\x \RR^2) +1 $ for $\e$ small enough;
		\item[(iii)] $\de G_{\overline u_\e}|_{(\Omega' \sm \{z_1,\dots,z_J\}) \x \RR^2} = 0$ for $\e$ small enough.
	\end{itemize}
\end{lemma}
\begin{proof}
	The proof relies on some arguments for the discrete vorticity measure that can be adapted in this from~\cite{BacCicKreOrl}. We provide some details for the sake of completeness. 
	
	We consider an auxiliary discrete vorticity measure $\mu_{u_\e}^{\triangle}$ defined through a triangulation with respect to the lattice $\e \ZZ^2$. (We do this since the results in~\cite{BacCicKreOrl} are stated on the triangular lattice.) More precisely, let $\varphi_\e \colon \e \ZZ^2 \to [0,2\pi)$ be such that $u_\e(x) = \exp(\iota \varphi_\e(x))$. As in~\eqref{eq:discrete vorticity}, for $\pm \in \{+,-\}$ in the triangle $\mathrm{conv}\{\e i, \e i \pm \e e_1, \e i \pm \e e_2\}$ we set 
	\begin{equation*}
		d_{u_\e}^\pm(\e i) := \frac{1}{2\pi} \Big[ \Psi\big( \varphi(\e i \pm \e e_1) - \varphi(\e i) \big) +   \Psi\big( \varphi(\e i \pm \e e_2) - \varphi(\e i \pm \e e_1) \big) + \Psi\big( \varphi(\e i) - \varphi(\e i \pm \e e_2) \big)   \Big]
	\end{equation*}
	and 
	\begin{equation*}
		\mu_{u_\e}^{\triangle} \mres \mathrm{conv}\{\e i, \e i \pm \e e_1, \e i \pm \e e_2\} := d_{u_\e}^\pm(\e i) \delta_{\e i \pm (1-\frac{\sqrt{2}}{2})\e e_1 \pm (1-\frac{\sqrt{2}}{2})\e e_2} \, ,
	\end{equation*}
	where $\e i \pm (1-\frac{\sqrt{2}}{2})\e e_1 \pm (1-\frac{\sqrt{2}}{2})\e e_2$ is the incenter of the triangle $\mathrm{conv}\{\e i, \e i \pm \e e_1, \e i \pm \e e_2\}$. Note that, if $d_{u_\e}^\pm(\e i) \neq 0$, then $\frac{1}{\e^2} XY_\e(u_\e;\mathrm{conv}\{\e i, \e i \pm \e e_1, \e i \pm \e e_2\}) \geq c_0$ for some universal constant $c_0$. Thus $|\mu_{u_\e}^\triangle|(\Omega') \leq \frac{C}{\e^2} XY_\e(u_\e;\Omega) \leq C |\log \e|$ for every $\Omega' \subset \subset \Omega$. Moreover,  
	\begin{equation} \label{eq:mu on triangles}
		|\mu_{u_\e}^\triangle|(A') = 0 \implies |\mu_{u_\e}|(A) = 0 \quad \text{for } A \subset \subset A'  \text{ and } \e \text{ small enough.}
	\end{equation}
	Indeed, if $\mu_{u_\e}^{\triangle} \mres \mathrm{conv}\{\e i, \e i + \e e_1, \e i + \e e_2\} = 0$ and $\mu_{u_\e}^{\triangle} \mres \mathrm{conv}\{\e i + \e e_1 + \e e_2, \e i + \e e_1, \e i + \e e_2\} = 0$, then $\mu_{u_\e}  \mres (\e i + (0,\e]^2)=0$. % Finally, as in Lemma~\ref{lemma:vorticityequiv}, one can show that $\mu_{u_\e}^\triangle \mres \Omega' - \mu_{u_\e} \mres \Omega' \flat 0$ for every $\Omega' \subset \subset \Omega$.

	Let us fix $\Omega' \subset \subset \Omega$. We define the family of balls
	\begin{equation*}
		\mathcal{B}_\e := \{B_{(1-\frac{\sqrt{2}}{2})\e}(x) : x \in \supp(\mu_{u_\e}^\triangle) \cap \Omega' \}
	\end{equation*}
	and we let $\mathcal{R}(\mathcal{B}_\e) := \sum_{B_r(x) \in \mathcal{B}_\e} r$. Each ball in $\mathcal{B}_\e$ is contained in a triangle of the lattice~$\e \ZZ^2$. Since $|\mu_{u_\e}^\triangle|(\Omega') \leq C |\log \e|$, we have that $\# \mathcal{B}_\e \leq C |\log \e|$. For every $0<r<R$ and for every $x \in \mathbb{R}^2$ we set $A_{r,R}(x) := B_R(x) \sm \overline B_r(x)$. If $A_{r,R}(x) \cap \bigcup_{B \in \mathcal{B}_\e} B =\emptyset$ we set 
	\begin{align*}
		\mathcal{E}_\e(A_{r,R}(x)) := |\mu_{u_\e}^\triangle(B_r(x))| \log\frac{R}{r}\,,
	\end{align*}
	and we extend $\mathcal{E}_\e$ to every open set $A$ by  
	\begin{align}\label{def: setfunction}
		\begin{split}
			\mathcal{E}_\e(A) := \sup\Big\{ \sum_{j=1}^N \mathcal{E}_\e(A^j)  : & \ N \in \mathbb{N} \, , \ A^j =A_{r_j,R_j}(x_j)\, , \  A^j \cap \bigcup_{B \in \mathcal{B}_\e} B =\emptyset \, ,\\
			&  \hspace{2em} A^j \cap A^k =\emptyset \text{ for } j \neq k \, , \ A^j \subset A  \text{ for all } j \Big\}\,.
		\end{split}
	\end{align}
	The set function $\mathcal{E}_\e$ is increasing, superadditive and equals $-\infty$ iff $A\subset\bigcup_{B \in \mathcal{B}_\e} B$.  Let $\Omega''$ be such that $\Omega' \subset \subset \Omega'' \subset \subset \Omega$. As in \cite[Lemma~7.1]{BacCicKreOrl} one can prove that
	\begin{equation} \label{eq:bound on E}
		\mathcal{E}_\e(\Omega'') \leq \frac{C}{\e^2} XY_\e(u_\e;\Omega) \leq C |\log \e| \, .
	\end{equation}
	
	We apply the ball construction to the triplet $(\mathcal{E}_\e,\mu_{v_\e}^\triangle,\mathcal{B}_\e)$. The form which suits most the arguments here is the one stated in~\cite[Lemma~6.1]{BacCicKreOrl}. To keep track of the constants, we let $\overline C$ be such that $XY_\e(u_\e;\Omega) \leq \overline C \e^2 |\log \e|$. We fix $p\in(\frac{3}{4},1)$ and we set $\alpha_\e := \overline C \e^p |\log \e|$. Then there exists a family $\{\mathcal{B}_\e(t)\}_{t\geq 0}$ which satisfies
	\begin{itemize}
		\setlength\itemsep{1em}
		\item[{\rm (1)}]
		
		$\displaystyle\bigcup_{B \in \mathcal{B}_\e} B \subset \bigcup_{B \in \mathcal{B}_\e(t_1)} \!  B \subset \bigcup_{B \in \mathcal{B}_\e(t_2)}\! B \, , \quad  \text{for every } 0 \leq t_1 \leq t_2$ ;
		
		\item[{\rm (2)}]   $\overline B\cap \overline B^\prime = \emptyset$  for every $B,B^\prime \in \mathcal{B}_\e(t)$, $B\neq B^\prime$,   and $t \geq 0$;  
		\item[{\rm (3)}] for every $0\leq t_1 \leq t_2$ and every open set $U$ we have that
		\begin{align*} 
			\mathcal{E}_\e\Big(U \cap \Big( \bigcup_{B \in \mathcal{B}_\e(t_2)}B \setminus \bigcup_{B \in \mathcal{B}_\e(t_1)}\overline{B}\Big)\Big) \geq  {\sum_{ \substack{ B \in\mathcal{B}_\e(t_2) \\ B \subset U}}} |\mu_{u_\e}^\triangle(B)| \log\frac{1+t_2}{1+t_1}\, ; 
		\end{align*}
		\item[{\rm (4)}] for $B=B_r(x) \in \mathcal{B}_\e(t)$ and $t\geq 0$, we have that $r > \alpha_\e$ and $|\mu_{u_\e}^\triangle|(B_{r+\alpha_\e}(x)\setminus \overline B_{r-\alpha_\e}(x))=0$;
		\item[{\rm (5)}] for every $t\geq 0$ we have that $\mathcal{R}(\mathcal{B}_\e(t)) \leq (1+t) \big( \mathcal{R}(\mathcal{B}_\e) + \#\mathcal{B}_\e  \alpha_\e \big)$, where $\mathcal{R}(\mathcal{B}_\e(t)) := \sum_{B_r(x) \in \mathcal{B}_\e(t)} r$;
	\end{itemize}
	Note that in general $\mathcal{B}_\e(0)$ is not $\mathcal{B}_\e$. We let $t_\e:=\e^{p-1}-1$ and we define the measures 
	\begin{equation*}
		\tilde \mu_\e := \sum_{B_r(x) \in \mathcal{B}_\e(t_\e)} \mu_{u_\e}^\triangle(B_r(x)) \delta_{x} \, .
	\end{equation*}
	Since $\# \mathcal{B}_\e \leq C |\log \e|$, by property (5) above we have that
	\begin{equation} \label{eq:bound on radii}
		\mathcal{R}(\mathcal{B}_\e(t_\e)) \leq  (1+t_\e) \big( \mathcal{R}(\mathcal{B}_\e) + \#\mathcal{B}_\e \alpha_\e \big) \leq C\e^{2p-1}  |\log \e|^2 .
	\end{equation}
	Moreover, since $\mathcal{E}_\e$ is an increasing set function, by~\eqref{eq:bound on E}, and property~(3) in the ball construction, for $\e$ small enough we have that 
	\begin{equation*}
		C|\log \e| \geq \mathcal{E}_\e(\Omega'') \geq \sum_{\substack{B \in \mathcal{B}_\e(t_\e) \\ B \subset \Omega''}} |\mu_{u_\e}^\triangle(B)| \log (1+t_\e) \geq |\tilde \mu_\e|(\Omega') (1-p) |\log \e|  
	\end{equation*}
	and thus
	\begin{equation} \label{eq:mu is bounded}
		|\tilde \mu_\e|(\Omega') \leq C \, .
	\end{equation}
	We consider the following two subclasses of $\mathcal{B}_\e(t_{\e})$: 
	\begin{align}\label{def:subclasses}
		\begin{split}
			&\mathcal{B}_\e^{=0} := \{ B_r(x) \in \mathcal{B}_\e(t_{\e}) : \mu_{u_\e}^\triangle(B_r(x))=0 \, , \ x \in \Omega^{\prime}\}\,, \\
			& \mathcal{B}_\e^{\neq 0} := \{ B_r(x) \in \mathcal{B}_\e(t_{\e}) : \mu_{u_\e}^\triangle(B)\neq 0  \, , x \in \Omega^\prime \}\,.
		\end{split}
	\end{align}
	Let $B_{r_\e}(x_\e) \in \mathcal{B}_\e^{=0}$. Thanks to property~(4) in the ball construction, $|\mu_{u_\e}^\triangle|(B_{r_\e+\alpha_\e}(x_\e)\setminus \overline B_{r_\e-\alpha_\e}(x_\e))=0$. We set $K_\e := \lfloor \frac{\alpha_\e}{4\e^p} \rfloor \geq 1$ and $r_k := r_\e-\frac{\alpha_\e}{2} + k \e^p $ for $k = 0, \dots , K_\e$. Note that $r_{k} \leq r_\e - \frac{\alpha_\e}{4}$. For every~$\e$ there exists $k_\e \in \{1, \dots, K_\e\}$ such that
	\begin{equation*}
		\begin{split}
			\overline C \e^2 |\log \e| & \geq XY_\e(u_\e;B_{r_\e+\alpha_\e}(x_\e)\setminus \overline B_{r_\e-\alpha_\e}(x_\e)) \geq \sum_{k=1}^{K_\e}  XY_\e(u_\e;A_{r_{k-1},r_k}(x_\e)) \\
			&  \geq K_\e  XY_\e(u_\e;A_{r_{k_\e-1},r_{k_\e}}(x_\e)) \geq   \frac{\alpha_\e}{8\e^p}   XY_\e(u_\e;A_{r_{k_\e-1},r_{k_\e}}(x_\e))\, .
		\end{split}
	\end{equation*}
	Rearranging terms the definition of $\alpha_{\e}$ yields the bound $XY_\e(u_\e;A_{r_{k_\e-1},r_{k_\e}}(x_\e)) \leq C_1 \e^2$, with $C_1 := 8$. Hence we can apply Lemma~\ref{lemma:extension} below to find $\overline u_\e \colon \e \ZZ^2 \to \SS^1$  such that,  for $\e < (r_{k_\e}-r_{k_\e-1})\big(\frac{2\pi}{3}\big)^2 \frac{1}{C_0 C_1}$, we have
	\begin{equation*}
	\overline u_\e = u_\e\text{ on }\e \ZZ^2 \sm \overline B_{\frac{r_{k_\e-1}+r_{k_\e}}{2}}(x_\e) \quad\text{and}\quad |\mu_{\overline u_\e}^\triangle|( B_{r_{k_\e}}(x_\e)) = 0 \, .	
	\end{equation*}
	The condition $\e < (r_{k_\e}-r_{k_\e-1})\big(\frac{2\pi}{3}\big)^2 \frac{1}{C_0 C_1}$ is satisfied as $\e < \e^p \big(\frac{2\pi}{3}\big)^2 \frac{1}{8 C_0}$ for $\e$ small enough. Since $r_{k_\e} \leq r_\e - \frac{\alpha_\e}{4} \ll r_\e - \sqrt{2} \e$, we also have that the piecewise constant functions $\overline u_\e$ and $u_\e$ coincide outside~$B_{r_\e}(x_\e)$.
	
	We apply the modification described above to every $B_{r_\e}(x_\e) \in \mathcal{B}_\e^{=0}$. In this way, for every~$\e$ we construct   $\overline u_\e \colon \e \ZZ^2 \to \SS^1$  such that $\overline u_\e = u_\e$ (as piecewise constant functions)  in $\RR^2 \sm \bigcup_{B \in \mathcal{B}_\e^{=0}} B$ and $|\mu_{\overline u_\e}^\triangle|(A) = 0$ for every open set $A$ such that $A \subset \subset \Omega' \sm \bigcup_{B \in \mathcal{B}_\e^{\neq 0}} B$. By~\eqref{eq:mu is bounded}--\eqref{def:subclasses} and by the definition of $\tilde \mu_\e$, we have that $\# \mathcal{B}_\e^{\neq 0}$ is equibounded and, up to a subsequence, we can assume $\mathcal{B}_\e^{\neq 0} = \{B_{r_1^\e}(x_1^\e), \dots , B_{r_M^\e}(x_M^\e)\}$. There exists a set of points $\{z_1, \dots, z_J\} \subset \overline{\Omega}'$ with $J \leq M$ such that, up to a subsequence, each sequence $x_m^\e$ converges to a point $z_h$ as $\e \to 0$. (The points belonging to $\de \Omega'$ are actually not relevant for the following discussion.) From now on, we work with this subsequence.
	
	Let us show that $G_{\overline u_\e} - G_{u_\e} \weak 0$ in $\D_2(\Omega' \x \RR^2)$. Let $\phi \in C^\infty_c(\Omega' \x \RR^2)$. By~\eqref{eq:bound on radii} we have  
	\begin{equation*}
		\begin{split}
			& | (G_{\overline u_\e} - G_{u_\e})(\phi(x,y) \d x)| \leq \int_{\Omega'}\! | \phi(x,\overline u_\e(x)) - \phi(x, u_\e(x)) | \d x  \\
			& \quad \leq \sum_{B \in \mathcal{B}_\e^{=0}}\int_{B} \! | \phi(x,\overline u_\e(x)) - \phi(x, u_\e(x)) | \d x \leq 2 \|\phi\|_{L^\infty} \pi  \mathcal{R}(\mathcal{B}_\e(t_\e))^2 \leq C \|\phi\|_{L^\infty}  \e^{4p-2} |\log \e|^4.
		\end{split}
	\end{equation*}
	For the next estimate we observe that, given a ball $B_r(x)$, we have $\H^1(J_{u_\e} \cap B_r(x)) \leq C\frac{r^2}{\e}$. Indeed, since $u_\e$ is piecewise constant on the squares of $\e \ZZ^2$, the measure of its jump set in $B_r(x)$ can be roughly estimated by $4 \e$ times the number of squares that intersect $B_r(x)$, which is of the order of $\frac{r^2}{\e^2}$, at least for $\e\lesssim r$. The same holds true for $\overline u_\e$. For $x \in J_{\overline u_\e}$ (resp., $x \in J_{u_\e}$), let $\overline \gamma_x$ (resp., $\gamma_x$) be the (oriented) geodesic arc that connects~$\overline u_\e^+(x)$ and~$\overline u_\e^-(x)$ (or $u_\e^+(x)$ and $u_\e^-(x)$). Then, using that $\overline u_\e = u_\e$ on $\e \ZZ^2 \sm \bigcup_{B \in \mathcal{B}_\e^{=0}} B$, 
	\begin{equation*}
		\begin{split}
			& |(G_{\overline u_\e} - G_{u_\e})(\phi(x,y) \d \hat x^l \w \d y^m)| \leq \\
			& \quad  \leq  \bigg| \sum_{B \in \mathcal{B}_\e^{=0}}  \int \limits_{J_{\overline u_\e}\cap B} \bigg\{ \int \limits_{\overline \gamma_x} \phi(x,y) \d y^m \bigg\} \nu^l_{\overline u_\e}(x) \d \H^1(x)  - \hspace{-0.5em}\int \limits_{J_{u_\e} \cap B} \bigg\{ \int \limits_{ \gamma_x} \phi(x,y) \d y^m \bigg\} \nu^l_{ u_\e}(x) \d \H^1(x)  \bigg| \\
			& \quad  \leq 2 \pi \|\phi\|_{L^\infty} \sum_{B \in \mathcal{B}_\e^{=0}} \big( \H^1(J_{\overline u_\e} \cap B) + \H^1(J_{u_\e} \cap B) \big) \leq C  \|\phi\|_{L^\infty}  \frac{2}{\e} \mathcal{R}(\mathcal{B}_\e(t_\e))^2 \\
			& \quad \leq C \|\phi\|_{L^\infty}   \e^{4p-3} |\log \e|^4.
		\end{split}
	\end{equation*}
	 The two previous inequalities and the fact that $p\in(\frac{3}{4},1)$ imply that $G_{\overline u_\e} - G_{u_\e} \weak 0$ in $\D_2(\Omega' \x \RR^2)$. Moreover, taking the supremum over 2-forms with $L^\infty$-norm less than 1, they also imply that $ |G_{\overline u_\e}|( \Omega' \x \RR^2) \leq  |G_{u_\e}|(\Omega' \x \RR^2) + 1$ for $\e$ small enough.
	
	Let $A \subset \subset A' \subset \subset \Omega' \sm  \{z_1,\dots, z_J\}$.  By~\eqref{eq:bound on radii}, for $\e$ small enough it follows that $A' \subset \subset \Omega' \sm \bigcup_{B \in \mathcal{B}_\e^{\neq 0}} B$ and thus $|\mu_{\overline u_\e}^\triangle|(A') = 0$. By~\eqref{eq:mu on triangles}, we obtain that $|\mu_{\overline u_\e}|(A) = 0$, i.e., $\de G_{\overline u_\e}|_{A \x \RR^2} = 0$. By the arbitrariness of $A$ and $A'$ we get $\de G_{\overline u_\e}|_{(\Omega' \sm \{z_1,\dots,z_J\}) \x \RR^2} = 0$.
\end{proof}

In the proof we applied the following extension lemma proven  in~\cite[Lemma 3.5 and Remark 3.6]{BacCicKreOrl}.  

\begin{lemma}\label{lemma:extension}
	There exists a constant $C_0>0$ such that the following holds true. Let $\e >0$,  $x_0 \in \RR^2$,and $R > r > \e$, let $C_1 > 1$ and $u_\e \colon \e \ZZ^2 \to \SS^1$ with $XY_\e(u_\e;B_R(x_0) \sm \overline B_r(x_0)) \leq C_1 \e^2$, $\mu_{u_\e}^\triangle(B_{r}(x_0)) = 0$, and $|\mu_{u_\e}^\triangle|(\e i + (0,\e]^2)= 0$ whenever $(\e i + (0,\e]^2) \cap (B_R(x_0) \sm \overline B_r(x_0)) \neq 0$. Then there exists $\overline u_\e \colon \e \ZZ^2 \to \SS^1$  such that  for $\e < \frac{R-r}{C_0C_1}\big(\frac{2\pi}{3}\big)^2$: 
	\begin{itemize}
		\item  $\overline u_\e = u_\e$ on $\e \ZZ^2  \sm \overline B_{\frac{r+R}{2}}(x_0)$;  
		\item $|\mu_{\overline u_\e}^\triangle|( B_{R}(x_0)) = 0$;
		%\item $XY_\e(\overline u_\e; B_R(x_0)) \leq C(r,R) XY_\e(u_\e; B_R(x_0) \sm \overline B_r(x_0) )$, where $C(r,R) = C_0 \frac{R}{R-r}$.
	\end{itemize}
\end{lemma}

We are finally in a position to prove Proposition~\ref{prop:current compactness}.

\begin{proof}[Proof of Proposition~\ref{prop:current compactness}]
	From the assumptions $\frac{1}{\e\theta_{\e}}E_{\e}(u_{\e})\leq C$ and $\theta_{\e}\lesssim\e|\log\e|$ it follows that 
\begin{equation*}
\frac{1}{\e^2 |\log \e|} E_\e(u_\e) =  \frac{\theta_\e}{\e |\log \e|}\frac{1}{\e \theta_\e} E_\e(u_\e) \leq C \, ,
\end{equation*}
so that by Proposition~\ref{prop:XY classical} we get that (up to a subsequence) $\mu_{u_\e} \flat \mu = \sum_{h=1}^M d_h \delta_{x_h}$, proving the first point. 

Applying Lemma~\ref{lemma:bound with parametric integral} with $\sigma = \frac{1}{2}$ we deduce that for every $A \subset \subset \Omega$ 
\begin{equation*}
	C \geq \frac{1}{\e \theta_\e} E_\e(u_\e) \geq \frac{1}{2}\integral{J_{u_\e} \cap A}{\geo(u_\e^-,u_\e^+)|\nu_{u_\e}|_1}{\d \H^1} \geq \frac{1}{2} |\DD u_\e|(A) \, .
\end{equation*}
Hence $u_{\e}$ is bounded in $BV(A;\SS^1)$ and we conclude that (up to a subsequence) $u_{\e}\to u$ in $L^1(A)$  and $u_{\e}\wstar u$ in $BV(A;\RR^2)$ for some $u\in BV(A;\SS^1)$ with $|\DD u|(A)\leq C$. Since $A\compact\Omega$ was arbitrary and the constant $C$ does not depend on $A$, the second point follows from a diagonal argument and the equiintegrability of $u_{\e}$.

Applying Lemma~\ref{lemma:bound with parametric integral} with $\sigma = \frac{1}{2}$ and since $\Phi(\xi) \geq \sqrt{(\xi^{21})^2 + (\xi^{22})^2 + (\xi^{11})^2 + (\xi^{12})^2 }$,  we obtain by Proposition \ref{prop:Gu is im}  that for every $A \subset \subset \Omega$
\begin{equation*} 
\begin{split}
|G_{u_\e}|(A \x \RR^2) & = |G_{u_\e}|(\M^{(a)} \cap A \x \RR^2) + |G_{u_\e}|(\M^{(j)} \cap A \x \RR^2) \\
& \leq  |A| + \integral{A \x \RR^2}{\Phi(\vec{G}_{u_\e})}{\d |G_{u_\e}|} \leq |\Omega| + \frac{2}{\e\theta_\e} E_\e(u_\e) \leq C \, .
\end{split}
\end{equation*}
By the Compactness Theorem for currents~\cite[2.2.3, Proposition~2 and Theorem 1-(i)]{Gia-Mod-Sou-I} we deduce that there exists a subsequence (not relabeled) and a current $T \in \D_2(\Omega \x \RR^2)$ with $|T| < \infty$ such that $G_{u_\e} \weak T$ in $\D_2(\Omega \x \RR^2)$.

It remains to prove that $T \in \mathrm{Adm}(\mu,u;\Omega)$:
\begin{itemize}
	\item $T$ is an i.m.\ rectifiable current: Let $\Omega' \subset \subset \Omega$. We consider the subsequence (not relabeled), the points  $z_1,\dots, z_J \in \Omega'$, and the spin field $\overline u_\e \colon \e \ZZ^2 \to \SS^1$ given by Lemma~\ref{lemma:modification}. By Lemma~\ref{lemma:modification}-{\em i)} we have that $G_{\overline u_\e} \weak T$ in $\D_2(\Omega' \x \RR^2)$. Let us fix $A \subset \subset \Omega' \sm  \{x_1,\dots,x_M,z_1,\dots,z_J\}$.  We have that $\sup_\e |G_{u_\e}|(A \x \RR^2) < +\infty$ and $\de G_{u_\e}|_{A \x \RR^2} = 0$. By the Closure Theorem~\cite[2.2.4, Theorem 1]{Gia-Mod-Sou-I}, the limit $T|_{A \x \RR^2}$ is an i.m.\ rectifiable current. By the arbitrariness of $A$ and $\Omega'$ and since $\{x_h\}\x \SS^1$ and $\{z_j\} \x \SS^1$ are $\H^{2}$-negligible sets, this proves that $T$ is an i.m.\ rectifiable current in $\Omega \x \RR^2$.
	\item $\de T|_{\Omega \x \RR^2} = -\mu \x \llbracket \SS^1 \rrbracket$: by Proposition~\ref{prop:bd of Gu is mu} we have $\de G_{u_\e}|_{\Omega \x \RR^2} = - \mu_{u_\e} \x \llbracket \SS^1 \rrbracket$. Since $\mu_{u_\e} \flat \mu$, by Lemma~\ref{lemma:flat implied D1}, and since $\de G_{u_\e} \weak \de T$ in $\D_1(\Omega \x \RR^2)$, we conclude that $\de T|_{\Omega \x \RR^2} = -\mu \x \llbracket \SS^1 \rrbracket$. In particular, $\de T|_{\Omega_\mu \x \RR^2} = 0$.
	\item $T|_{\d x} \geq 0$: let $\omega \in \D_2(\Omega \x \RR^2)$ be of the form $\omega(x,y) = \phi(x,y) \d x$ with $\phi \in C^\infty_c(\Omega \x \RR^2)$  and  $\phi \geq 0$. Then $G_{u_\e}(\omega) = \int_\Omega \phi(x,u_\e(x)) \d x \geq 0$. Passing to the limit as $\e \to 0$ we get $T(\omega) \geq 0$.
	\item $|T| < \infty$: this is a consequence of the Compactness Theorem for currents (see above).
	\item $\|T\|_1 < \infty$: note that, by~\eqref{eq:norm 1}, $\| G_{u_\e} \|_1 = \int_\Omega |u_\e| \d x = |\Omega|$. By the lower semicontinuity of $\| \cdot \|_1$ with respect to the convergence in $\D_2(\Omega \x \RR^2)$ we deduce that $\|T\|_1 \leq |\Omega|$. 
	\item $\pi^\Omega_\# T = \llbracket \Omega \rrbracket$: let us fix $\omega \in \D^2(\Omega)$, i.e., a 2-form of the type $\omega(x) = \phi(x) \d x$ with $\phi \in C^\infty_c(\Omega)$. Then $G_{u_\e}(\omega) = \int_\Omega \phi(x) \d x$. Thus $\pi^\Omega_\# G_{u_\e} =  \llbracket \Omega \rrbracket$. Passing to the limit as $\e \to 0$ we get the desired condition (cf.\ also \cite[4.2.1, Proposition~3]{Gia-Mod-Sou-I}). 
	\item $\supp(T) \subset \ol \Omega \x \SS^1$: let us fix $\omega \in \D^2(\Omega \x \RR^2)$ with $\supp(\omega) \compact (\Omega \x \RR^2) \sm (\Omega \x \SS^1)$. Then $G_{u_\e}(\omega) = 0$. Passing to the limit as $\e \to 0$, we conclude that $T(\omega) = 0$.
\end{itemize}

To prove that $u_T = u$ a.e.\ we observe that $u_\e \to u$ implies
\begin{equation*}
G_{u_\e}(\phi(x,y) \d x) = \integral{\Omega}{\phi(x,u_\e(x))}{\d x}  \to \integral{\Omega}{\phi(x,u(x))}{\d x} 
\end{equation*}
for every $\phi \in C^\infty_c(\Omega \x \RR^2)$. On the other hand,  due to Theorem \ref{thm:structure}
\begin{equation*}
G_{u_\e}(\phi(x,y) \d x) \to T(\phi(x,y) \d x) = \int_\Omega \phi(x,u_T(x)) \d x \, .
\end{equation*}
By the arbitrariness of $\phi$, we get $u_T = u$ a.e.\ in $\Omega$.

Finally, if $\theta_\e \ll \e |\log \e|$, then Proposition \ref{prop:XY classical} and the assumed energy bound yield that $\mu = 0$, whence $T \in \cart(\Omega \x \SS^1)$.
\end{proof}

\section{Proofs in the regime \texorpdfstring{$\theta_{\e}\sim \e |\log \e|$}{}}\label{sec:theta<<eloge}
In this section we prove Theorem~\ref{thm:theta equal e log}. We remark that the compactness result Theorem~\ref{thm:theta equal e log}-{\em i)} is already covered by Proposition~\ref{prop:current compactness}. Thus, we only need to prove Theorem~\ref{thm:theta equal e log}-{\em ii)} and~{\em iii)}.

From the lower semicontinuity of parametric integrals with respect to the mass bounded weak convergence of currents, \cite[1.3.1, Theorem~1]{Gia-Mod-Sou-II}, we obtain the following asymptotic lower bound. 

\begin{proposition}[Lower bound for the parametric integral]\label{prop:lb for parametric}  Assume that $\theta_{\e} \lesssim \e|\log \e|$ and that $\frac{1}{\e \theta_\e}E_\e(u_\e) \leq C$. Let $G_{u_\e} \in \D_2(\Omega \x \RR^2)$ be the currents associated to $u_\e$ defined as in~\eqref{eq:Gu ac}--\eqref{eq:Gu v} and assume that $G_{u_{\e}}\rightharpoonup T$ with $T \in \D_2(\Omega \x \RR^2)$ represented as $T = \vec{T} |T|$. Then
	\begin{equation} \label{eq:parametric lower bound}
	\integral{A \x \RR^2}{\Phi( \vec{T} )}{\d |T|} \leq \liminf_{\e \to 0} \integral{A \x \RR^2}{\Phi( \vec{G}_{u_\e} )}{\d |G_{u_\e}|} \, .
	\end{equation}
	for every open set $A \compact \Omega$.
	\end{proposition}

	We can write explicitly the parametric integral in the left-hand side of~\eqref{eq:parametric lower bound} in terms of the $BV$ function $u$, limit of the sequence $u_\e$. We recall that by~\eqref{eq:jc part of T} the jump-concentration part of $T$ is given by
	\begin{equation*}
	T^{(jc)}(\phi(x,y) \d \widehat x^l \w \d y^m) = (-1)^{2-l} \integral{J_T}{\bigg\{\integral{ \gamma^T_x}{\phi(x,y)}{\d y^m} \bigg\} \nu_{T}^l(x)}{\d \H^1(x)} \, .
	\end{equation*}
	For $\H^1$-a.e.\ $x \in J_T$ we define the number
	\begin{equation}\label{eq:deflT}
	\ell_T(x) := \mathrm{length}(\gamma^T_x) = \integral{\supp (\gamma^T_x)}{|\mathfrak{m}(x,y)|}{\d \H^1(y)} \, ,
	\end{equation}
	where $\mathfrak{m}(x,y)$ is the integer defined in~\eqref{eq:jc multiplicity}. Notice that by $\mathrm{length}(\gamma^T_x)$ we mean the length of the curve $\gamma_x^T$ counted with its multiplicity and not the $\H^1$ Hausdorff measure of its support. Observe that, in particular, $\ell_T(x) = \geo \big( u^-(x), u^+(x) \big)$ if $x \in J_u \sm \L$, whilst $\ell_T(x) = 2 \pi |k(x)|$ if $x \in \L \sm J_u$. The full form of the parametric integral is contained in the lemma below.  

	\begin{lemma}\label{lemma:parametric in terms of u}
		Let~$\Phi$ be the parametric integrand defined in~\eqref{eq:Phi is 2,1}. Let $\mu = \sum_{h=1}^M d_h \delta_{x_h}$ with $d_h \in \ZZ$ and~$u \in BV(\Omega;\SS^1)$. Let $T \in \mathrm{Adm}(\mu, u;\Omega)$. Then
	\begin{equation*}
	\begin{split}
	\integral{\Omega \x \RR^2}{\Phi( \vec{T} )}{\d |T|}  & = \integral{\Omega}{|\nabla u|_{2,1} }{\d x} + |\DD^{(c)} u|_{2,1}(\Omega) +  \integral{J_T \cap \Omega}{\ell_T(x) |\nu_T(x)|_1}{\d \H^1(x)} \, . 
	\end{split}
	\end{equation*}
	\end{lemma}
	\begin{proof}
		We first prove the statement in the case $\mu = 0$, namely $T \in \cart(\Omega\x\SS^1)$. We employ the mutually singular decomposition given by Theorem~\ref{thm:structure} and Proposition~\ref{prop:components of T}, so that $|T| = \H^2 \mres \M^{(a)} +  \H^2 \mres \M^{(c)} + |T^{(jc)}|$. First of all note that by~\eqref{eq:ac components} and \eqref{eq:T horiz} for every $\psi \in C^\infty_c(\Omega)$ we have
		\begin{equation} \label{eq:area for ac}
		\integral{\M^{(a)}}{ \psi(x) \frac{1}{\sqrt{1 + |\nabla u(x)|^2}}}{\d \H^2(x,y)} = \integral{\Omega}{ \psi(x)}{\d x} \, ,
		\end{equation}
		since both integrals are equal to $T^{(a)}(\psi(x) \d x)$. By approximation, the above equality is true for every $\psi \colon \Omega \to \RR$ such that $(x,y) \mapsto \psi(x)$ is $\H^2\mres \M^{(a)}$-measurable and $x \mapsto \psi(x)$ is $\L^2$-measurable. Note that $x \mapsto |\nabla u(x)|_{2,1}$ satisfies these measurability properties thanks to~\eqref{eq:ac components}. In particular, we deduce that 
		\begin{equation*}
		\begin{split}
		\integral{\Omega \x \RR^2}{\Phi( \vec{T}(x,y) )}{\d \H^2\mres \M^{(a)}(x,y)} & = \integral{ \M^{(a)}}{|\nabla u(x)|_{2,1} \frac{1}{\sqrt{1 + |\nabla u(x)|^2}}}{\d \H^2(x,y)} \\
		& = \integral{\Omega}{|\nabla u(x)|_{2,1}}{\d x} \, .
		\end{split}
		\end{equation*}
		
		Next note that, by~\eqref{eq:T cantor}, for every function $\psi \in C_c(\Omega)$ it holds true
		\begin{equation} \label{eq:two suprema}
			\sup_{\substack{|\tilde \omega(x,y)| \leq \psi(x) \\ \tilde \omega \in \D^2(\Omega\x \RR^2)}} T^{(c)}(\tilde \omega) = \hspace{-2em} \sup_{\substack{|\omega(x)| \leq \psi(x)\\ \omega \ |\DD^{(c)}u|-\text{measurable}}} \hspace{-2em} T^{(c)}(\omega)  \, .
		\end{equation}
		Indeed given $\tilde \omega \in \D^2(\Omega\x \RR^2)$ such that $|\tilde{\omega}(x,y)| \leq \psi(x)$, one can define the $|\DD^{(c)}u|$-measurable 2-form $\omega(x) := \tilde \omega(x,\tilde u_T(x))$ to prove that the left-hand side is greater than or equal to the right-hand side. For the reverse inequality, given a $|\DD^{(c)}u|$-measurable 2-form~$\omega$ such that $|\omega(x)| \leq \psi(x)$, one can regularize it and then define the 2-form $\tilde \omega(x,y) := \omega(x) \zeta(y)$, where $\zeta \in C^\infty_c(B_2)$ is such that $\zeta(y) = 1$ for $|y| \leq 1$ (note that~$\zeta$ does not affect the value of $T^{(c)}(\omega)$ thanks to~\eqref{eq:T cantor}).
		
		Since $|T^{(c)}|=\H^2 \mres \M^{(c)}$ and by~\eqref{eq:T cantor}, equality~\eqref{eq:two suprema} implies that 
		\begin{equation} \label{eq:area cantor}
			\integral{\M^{(c)}}{\psi(x)}{\d \H^2(x,y)} = \integral{\Omega}{\psi(x)}{\d |\DD^{(c)} u|(x)} \, ,
		\end{equation}
		for every function $\psi \in C_c(\Omega)$. By approximation, \eqref{eq:area cantor} holds true for every $\psi \colon \Omega \to \RR$ such that $\psi$ is $\H^2 \mres \M^{(c)}$-measurable and $|\DD^{(c)}u|$-measurable. The function $x \mapsto \big|\frac{\d \DD^{(c)} u}{\d |\DD^{(c)} u|} (x) \big|_{2,1}$ satisfies these measurability properties, cf.~\eqref{eq:cantor components}, thus \eqref{eq:cantor components} implies
		\begin{equation*}
		\begin{split}
		\integral{\Omega \x \RR^2}{\Phi( \vec{T}(x,y) )}{\d \H^2\mres \M^{(c)}(x,y)}  & = \integral{\M^{(c)}}{\Big|\frac{\d \DD^{(c)} u}{\d |\DD^{(c)} u|} (x) \Big|_{2,1}}{\d \H^2(x,y)} \\
		& = |\DD^{(c)} u|_{2,1}(\Omega) \, .
		\end{split}
		\end{equation*}
		
		Finally, by~\eqref{eq:jc components} we get that 
		\begin{equation*}
		\begin{split}
		\integral{\Omega \x \RR^2}{\Phi( \vec{T}(x,y) )}{\d |T^{(jc)}|(x,y)} &= \integral{\Omega \x \RR^2}{ |\mathfrak{m}(x,y)| |\nu_T(x)|_1}{\d \H^2\mres \M^{(jc)}(x,y)} \\
		& = \integral{J_T}{ \bigg\{   \integral{\SS^1}{ \mathds{1}_{\M^{(jc)}}(x,y) |\mathfrak{m}(x,y)| |\nu_T(x)|_1  }{\d \H^1(y)} \bigg\}  }{\d \H^1(x)}   \\
		& = \integral{J_T}{ \bigg\{   \integral{\supp(\gamma_x^T)}{\!  |\mathfrak{m}(x,y)| }{\d \H^1(y)} \bigg\} |\nu_T(x)|_1  }{\d \H^1(x)}   \\
		& = \integral{J_T}{\ell_T(x) |\nu_T(x)|_1}{\d \H^1(x)} \, .
		\end{split}
		\end{equation*}
		In the second equality we employed the coarea formula for rectifiable sets~\cite[Theorem~3.2.22]{Fed} (applied with $W = J_T \x \SS^1$, $Z = J_T$, $f$ given by the projection $J_T \x \SS^1 \to J_T$, and $g = \mathds{1}_{\M^{(jc)}} |\mathfrak{m}| |\nu_T|_1$) and in the third equality we used~\eqref{eq:Mjc is a product}.

		Let us prove the general case $T \in \mathrm{Adm}(\mu,u;\Omega)$. We observe that a current $T \in \cart(\Omega_\mu \x \SS^1)$ can be extended to a current $T \in \D_2(\Omega  \x \RR^2)$. Indeed, since $T \in \cart(\Omega_\mu \x \SS^1)$, it can be represented as 
		\begin{equation*}
			T(\omega) = \integral{\Omega_\mu \x \RR^2}{\langle \omega, \xi \rangle \theta}{\d \H^2 \mres \M}\, , \quad \text{for } \omega \in \D^2(\Omega_\mu \x \RR^2) \, ,
		\end{equation*}
		according to the notation in~\eqref{eq:im rectifiable}, where $\M \subset \Omega_\mu \x \SS^1$ $\H^2$-a.e. (cf.\ the proof of Proposition~\ref{prop:components of T} for the last fact). The integral above can be extended to a linear functional on forms $\omega \in \D^2(\Omega \x \RR^2)$, namely
		\begin{equation*}
			T(\omega) = \integral{\Omega \x \RR^2}{\langle \omega, \xi \rangle \theta}{\d \H^2 \mres \M}\, , \quad \text{for } \omega \in \D^2(\Omega \x \RR^2) \, .
		\end{equation*}
		To prove the continuity of this functional, let us fix a form $\omega \in \D^2(\Omega \x \RR^2)$ with $ \sup_{x}| \omega(x)| \leq 1$. We have the bound
		\begin{equation} \label{eq:comparing full with punctured}
			\begin{split}
				|T(\omega)| & \leq |T((1-\zeta) \omega)| + \Big|\integral{\Omega \x \RR^2}{\zeta \langle \omega, \xi \rangle \theta}{\d \H^2 \mres \M} \Big| \\
				&  \leq |T|(\Omega_\mu \x \RR^2) + \sum_{h=1}^N \integral{B_\rho(x_h) \x \RR^2}{|\theta|}{\d \H^2 \mres \M}
			\end{split}
		\end{equation}
		where $\zeta \in C^\infty_c(\Omega)$ is such that $0 \leq \zeta \leq 1$, $\supp(\zeta) \subset \bigcup_{h=1}^N B_\rho(x_h)$, and $\zeta \equiv 1$ on $B_{\rho/2}(x_h)$ for every $h=1,\dots,N$. Letting $\rho \to 0$ in the inequality above, we get $|T(\omega)| \leq |T|(\Omega_\mu \x \RR^2)$ since $\H^2\big(\M \cap (\{x_h\} \x \RR^2)\big) \leq  \H^2\big(\{x_h\} \x \SS^1\big) = 0$ for $h=1,\dots,N$ and $\theta$ is $\mathcal{H}^2\mres\M$-summable. This shows that $T \in \D_2(\Omega \x \RR^2)$.
		
		Moreover, by the arbitrariness of $\omega$ in \eqref{eq:comparing full with punctured} we deduce that $|T|(\Omega \x \RR^2)=|T|(\Omega_\mu \x \RR^2)$ and, in particular, from the first step of the proof applied to $\Omega_{\mu}$ we infer that
		\begin{align*}
			\integral{\Omega \x \RR^2}{\Phi( \vec{T} )}{\d |T|} =& \integral{\Omega_\mu \x \RR^2}{\Phi( \vec{T} )}{\d |T|} 
			\\
			=& \integral{\Omega}{|\nabla u|_{2,1} }{\d x} + |\DD^{(c)} u|_{2,1}(\Omega) +  \integral{J_T}{\ell_T(x) |\nu_T(x)|_1}{\d \H^1(x)} \, .
		\end{align*}
		\end{proof}

We are now in a position to prove the lower bound in the regime $\theta_\e \sim \e |\log \e|$. We recall that the asymptotic lower bound is written in terms of the energy
\begin{equation} \label{eq:def of surface L}
\mathcal{J}(\mu,u;\Omega) := \inf \bigg\{ \integral{J_T}{\ell_T(x) |\nu_T(x)|_1}{\d \H^1(x)} \ : \ T \in \mathrm{Adm}(\mu,u;\Omega)  \bigg\}
\end{equation}
with $\ell_T(x)$ defined in \eqref{eq:deflT} and $\mathrm{Adm}(\mu,u;\Omega)$ in~\eqref{eq:def of Adm}. We remark that $\mathrm{Adm}(\mu,u;\Omega)$ is non-empty.\footnote{By Proposition~\ref{prop:supporting BV} there exists a current $T \in \cart(\Omega \x \SS^1)$ such that $u_T = u$. Let $\gamma_1 , \dots, \gamma_M$ be pairwise disjoint unit speed Lipschitz curves  such that $\gamma_h$ connects $x_h$ to $\de \Omega$. Define $L_h$ to be the 1-current $\tau(\supp(\gamma_h), -d_h, \dot \gamma_h)$, so that $\de L_h = d_h \delta_{x_h}$. Then $T + \sum_{h=1}^M L_h \x \llbracket \SS^1 \rrbracket \in \mathrm{Adm}(\mu, u;\Omega)$.} Moreover,
\begin{equation} \label{eq2:J greater than jump}
	\mathcal{J}(\mu,u;\Omega) \geq \integral{J_{u} \cap \Omega}{\geo(u^-,u^+)|\nu_u|_1}{\d \H^1}.
\end{equation}
Indeed, for $\H^1$-a.e.\ $x \in J_T$ we have $\geo(u^-(x),u^+(x)) \leq \mathrm{length}(\gamma^T_x) = \ell_T(x)$, since $\gamma^T_x$ is a curve connecting $u^-(x)$ and $u^+(x)$ in $\SS^1$.

Using the previous results, we obtain the $\Gamma$-liminf inequality.

\begin{proof}[Proof of Theorem~\ref{thm:theta equal e log}-ii)]
	Let $u_\e \colon \Omega_\e \to \S_\e$ (extended arbitrarily to $\e \ZZ^2$), $u \in BV(\Omega;\SS^1)$, and $\mu = \sum_{h=1}^M d_h \delta_{x_h}$ be as in the statement of the theorem. Let $T \in \mathrm{Adm}(\mu,u;\Omega)$ be given by Proposition~\ref{prop:current compactness} and fix a set $A \subset \subset \Omega_\mu$.  Let $\rho > 0$ be such that the balls $\{B_\rho(x_h)\}_{h=1}^M$ are pairwise disjoint and $A \subset \subset \Omega_\mu^\rho$. 
	Let $\sigma \in (0,1)$. Then, by Lemma~\ref{lemma:bound with parametric integral}
	\begin{equation} \label{eq:split}
		\frac{1}{\e\theta_\e} E_\e(u_\e) \geq \sum_{h=1}^M \frac{1}{\e^2 |\log \e|} E_\e(u_\e;B_\rho(x_h)) + (1-\sigma) \integral{A \x \RR^2}{\Phi(\vec G_{u_\e})}{\d |G_{u_\e}|} \, .
	\end{equation}
	To estimate the first term, we exploit the localized lower bound for the $XY$-model~\cite[Theorem 3.1]{Ali-DL-Gar-Pon}, which yields the existence of a constant~$\tilde C \in \RR$ such that 
	\begin{equation*} 
	\liminf_{\e \to 0} \bigg[ \frac{1}{\e^2} E_\e(u_\e; B_\rho(x_h)) - 2 \pi |d_h| \log \frac{\rho}{\e} \bigg] \geq \tilde C 
	\end{equation*}
	and, in particular, 
	\begin{equation}\label{eq:localXYlb}
		\liminf_{\e \to 0}  \frac{1}{\e^2|\log \e|} E_\e(u_\e; B_\rho(x_h)) \geq 2 \pi |d_h|  \, .
		\end{equation}
	By~\eqref{eq:split}, Proposition~\ref{prop:lb for parametric}, letting $\sigma \to 0$ and $A \nearrow \Omega_\mu$, and by Lemma~\ref{lemma:parametric in terms of u} we infer that  
		\begin{equation*}
			\begin{split}
				\liminf_{\e \to 0}\frac{1}{\e\theta_\e} E_\e(u_\e) & \geq \sum_{h=1}^M 2 \pi |d_h| + \integral{\Omega_\mu \x \RR^2}{\Phi( \vec{T} )}{\d |T|}  \\
				& = 2 \pi |\mu|(\Omega) + \integral{\Omega}{|\nabla u|_{2,1} }{\d x} + |\DD^{(c)} u|_{2,1}(\Omega) +  \integral{J_T \cap \Omega}{\ell_T(x) |\nu_T(x)|_1}{\d \H^1(x)} \,  .
			\end{split}
		\end{equation*}
		Taking the infimum over all $T \in \mathrm{Adm}(\mu,u;\Omega)$ we deduce the claim. 
\end{proof}

Let us prove the $\Gamma$-limsup inequality. In the definition of the recovery sequence we use a map that projects vectors of $\SS^1$ on~$\S_\e$. Given $u\in\SS^1$ we let $\varphi_u\in [0,2\pi)$ be the unique angle such that $u=\exp(\iota\varphi_u)$. We define $\mathfrak{P}_{\e}\colon \SS^1\to\S_\e$ by 
\begin{equation}\label{eq:defproj}
\mathfrak{P}_{\e}(u)=\exp\left(\iota \theta_{\e}\left\lfloor\tfrac{\varphi_u}{\theta_{\e}}\right\rfloor\right).
\end{equation}

\begin{proof}[Proof of Theorem~\ref{thm:theta equal e log}-iii)]
To construct the recovery sequence, we closely follow the proof of~\cite[Proposition~4.22]{CicOrlRuf} done for the regime $\e \ll \theta_\e \ll \e |\log \e|$. Most of the arguments hold true also when $\theta_\e = \e |\log \e|$, see \cite[Remark~4.23]{CicOrlRuf}. Here we will sketch the proof and provide more details for the steps that need to be adapted.  

Let us fix $\mu = \sum_{h=1}^M d_h \delta_{x_h}$ and $u \in BV(\Omega;\SS^1)$ as in the statement. The function $u$ is gradually approximated as explained in the following.

\ul{Step 1} (Approximation with currents) 
Fix $\sigma > 0$. By the definition~\eqref{eq:def of surface L} of $\mathcal{J}$ there exists $T \in \mathrm{Adm}(\mu,u;\Omega)$ such that 
\begin{equation} \label{eq:from BV to cart}
\integral{\Omega \x \RR^2}{\Phi(\vec{T})}{\d |T|} \leq \integral{\Omega}{|\nabla u|_{2,1} }{\d x} + |\DD^{(c)} u|_{2,1}(\Omega) + \mathcal{J}(\mu, u;\Omega) + \sigma \, .
\end{equation}

\ul{Step 2} (Approximation with $\SS^1$-valued maps with finitely many singularities) 
Exploiting the extension Lemma~\ref{lemma:extension of currents} and the approximation Theorem~\ref{thm:approximation}, we find an open set $\tilde \Omega \supset  \supset \Omega$ and a sequence of maps $u_k \in C^\infty(\tilde \Omega_\mu;\SS^1) \cap W^{1,1}( \tilde{\Omega};\SS^1)$ such that $u_k \to u$ in $L^1(\Omega;\RR^2)$, $|G_{u_k}|(\Omega \x \RR^2) \to |T|(\Omega \x \RR^2)$, and $\deg(u_k)(x_h) = d_h$ for $h = 1,\dots,N$. We refer to~\cite[Lemma~4.17]{CicOrlRuf} for a detailed proof. Reshetnyak's Continuity Theorem implies that 
\begin{equation} \label{eq:from cart to smooth}
\integral{\Omega}{|\nabla u_k|_{2,1}}{\d x} = \integral{\Omega \x \RR^2}{\Phi(\vec{G}_{u_k})}{\d |G_{u_k}|} \leq \integral{\Omega \x \RR^2}{\Phi(\vec{T})}{\d |T|} + \sigma \, ,
\end{equation}
for $k$ large enough. In the first equality we applied Lemma~\ref{lemma:parametric in terms of u}. Thanks to this step (and via a diagonal argument as $\sigma \to 0$), it is enough to prove the $\Gamma$-limsup inequality assuming the stronger regularity $u \in C^\infty(\tilde \Omega_\mu;\SS^1) \cap W^{1,1}( \tilde{\Omega};\SS^1)$. 

\ul{Step 3} (Splitting of the degree) Without loss of generality, hereafter we shall assume that $|\deg(u)(x_h)| = 1$ for $h=1,\dots,N$. If this is not the case, we split each singularity $x_h$ with degree $d_h$ into~$|d_h|$ singularities of degree with modulus 1, without increasing the energy asymptotically. More precisely, by~\cite[Lemma~4.18]{CicOrlRuf}, for $0 < \tau \ll 1$ there exist measures $\mu^\tau$ and $u^\tau \in C^\infty(\tilde \Omega_{\mu^\tau};\SS^1) \cap W^{1,1}( \tilde{\Omega};\SS^1)$  with $\mu^\tau = \sum_{h=1}^{N^\tau} \deg(u^\tau)(x_h^\tau) \delta_{x_h^\tau}$ and $|\deg(u^\tau)(x_h^\tau)|=1$ such that $u^\tau \to u$ in $L^1(\Omega;\RR^2)$, $\mu^\tau \flat \mu$, and $\int_\Omega |\nabla u^\tau|_{2,1} \d x \to \int_\Omega |\nabla u|_{2,1} \d x$, as $\tau \to 0$.

\ul{Step 4} (Moving singularities on a lattice) We introduce the additional parameter $\lambda_n := 2^{-n}$, $n \in \NN$, which will be used later to obtain a piecewise constant approximation. Without loss of generality, we shall assume that $x_h \in \lambda_n \ZZ^2$ for $h=1,\dots,N$. If this is not the case, we find an approximation of~$u$ in the $W^{1,1}(\tilde \Omega)$-norm satisfying that property as follows. For every $n$ and $h=1,\dots,N$ we choose $x^{n}_h \in \lambda_n \ZZ^2 \cap \Omega$ such that $x^{n}_h \to x_h$ as $n \to +\infty$. For every $n$ there exists a diffeomorphism~$\psi_n \colon \tilde \Omega \to \tilde \Omega$ such that $\psi_n(x^n_h) = x_h$ for $h=1, \dots, N$ (see, e.g., \cite[p.\ 210]{Guo-Ran-Jos} for an explicit construction). We remark that it is possible to construct $\psi_n$ in such a way that~$\|\psi_n - \mathrm{id}\|_{C^1}$ and $\|\psi_n^{-1} - \mathrm{id}\|_{C^1}$ are controlled by $ \max_{h} |x^n_h - x_h|$ for every $n$. In particular, $\|\psi_n - \mathrm{id}\|_{C^1}, \|\psi_n^{-1} - \mathrm{id}\|_{C^1} \to 0$ for $n \to +\infty$. We define $u^n := u \circ \psi_n \in C^\infty(\tilde \Omega \sm \{x^n_1, \dots, x^n_N\};\SS^1) \cap W^{1,1}(\tilde \Omega;\SS^1)$. Then $u^n \to u$ strongly in $W^{1,1}(\tilde \Omega;\SS^1)$ as $n \to +\infty$. Let us fix $\rho>0$ such that the balls $\ol B_{\rho}(x_h)$ are pairwise disjoint and contained in $\tilde \Omega$. For $n$ large enough, we have that $x^n_h \in B_{\rho/4}(x_h)$ for $h=1,\dots,N$. Let $\zeta \in C^\infty_c(B_{\rho}(x_h))$ such that $\zeta \equiv 1$ on $B_{\rho/2}(x_h)$. By~\cite[Theorem~B.1]{BreCorLie} we have that 
\begin{equation*}
	2\pi \deg(u^n)(x^n_h) = - \hspace{-1em}\integral{B_{\rho}(x_h)}{\hspace{-0.5em}(u^n \x \nabla u^n )^\perp \cdot \nabla \zeta}{\d x} \overset{n\to +\infty}{\to}
	 - \hspace{-1em}\integral{B_{\rho}(x_h)}{\hspace{-0.5em}(u  \x \nabla u  )^\perp \cdot \nabla \zeta}{\d x} = 2\pi \deg(u)(x_h)
\end{equation*}
where $(u \x \nabla u)^\perp = ( u_1 \de_2 u_2 - u_2 \de_2 u_1, u_2 \de_1 u_1 - u_1 \de_1 u_2)$.

\ul{Step 5} (Modification near singularities) Let us fix $\sigma > 0$. Then there exists $\eta_0 > 0$ (small enough) and $u^\sigma \in C^\infty(\tilde \Omega_{\mu};\SS^1) \cap W^{1,1}( \tilde{\Omega};\SS^1)$ such that $\int_\Omega |\nabla u^\sigma|_{2,1} \d x \leq \int_\Omega |\nabla u|_{2,1} \d x + \sigma$, $u^\sigma(x) = u(x)$ in $\tilde \Omega \sm \bigcup_{h=1}^M \overline  B_{\sqrt{\eta_0}}(x_h)$, and $u^\sigma(x) = \big( \frac{x-x_h}{|x-x_h|} \big)^{d_h}$ in $B_{\eta_0}(x_h) \sm \{x_h\}$ (where the power is meant in the sense of complex functions). We refer to~\cite[Lemma~4.21]{CicOrlRuf} for a proof of this modification result. Thus, up to a diagonal argument as $\sigma \to 0$, we assume that $u$ has the structure of $u^\sigma$ with singularities $x_h \in \lambda_n \ZZ^2$. 

\ul{Step 6} (Recovery sequence near singularities) By the assumption in Step~5, $u(x) := \big(\frac{x-x_h}{|x-x_h|}\big)^{d_h}$ in $B_{\eta_0}(x_h)$, where $d_h = \pm 1$. In~\cite[Formula~(4.75)]{CicOrlRuf} we showed that the projection $\mathfrak{P}_\e(u)$ is concentrating the energy of a vortex near the singularity. More precisely, for every $\eta \in (0, \eta_0)$ we have that 
\begin{equation} \label{eq2:step6final}
	\limsup_{\e \to 0} \Big( \frac{1}{\e \theta_\e} E_\e(\mathfrak{P}_\e(u); B_\eta(x_h)) - 2 \pi |\log \e| \frac{\e}{\theta_\e} \Big) \leq C \eta \, ,
\end{equation}
for some universal constant $C$. In the regime $\theta_\e = \e |\log \e|$, this yields 
\begin{equation*}
	\limsup_{\e \to 0} \frac{1}{\e \theta_\e} E_\e\Big(\mathfrak{P}_\e(u); \bigcup_{h=1}^M B_\eta(x_h)\Big)  \leq 2\pi M + C \eta = 2 \pi |\mu|(\Omega) + C \eta \, .
\end{equation*}

\ul{Step 7} (Recovery sequence far from singularities)
Fix $\eta \in (0,\eta_0)$. We consider a suitable square centered at the singularities $x_h$ and with corners on $\lambda_n \ZZ^2$. More precisely, let $m(\lambda_n) \in \NN$ be the maximal integer such that $Q(\lambda_n,x_h) := x_h + [-2^{m(\lambda_n)} \lambda_n , 2^{m(\lambda_n)} \lambda_n]^2 \subset B_{\eta/2}(x_h)$, so that the estimate in Step~5 holds true in $\bigcup_{h=1}^M Q(\lambda_n,x_h)$. We also consider the square $Q_0(\lambda_n,x_h) := x_h + [(-2^{m(\lambda_n)}+1) \lambda_n , (2^{m(\lambda_n)}-1) \lambda_n]^2$. The squares $Q(\lambda_n,x_h)$ and~$Q_0(\lambda_n,x_h)$ differ by a frame made by 1 layer of squares of $\lambda_n \ZZ^2$. By the choice of $m(\lambda_n)$, one can prove that $B_{\eta/16}(x_h) \subset\subset Q_0(\lambda_n,x_h)$. Far from the singularities, we discretize $u$ on the lattice $\lambda_n \ZZ^2$. Specifically, we exploit the fact that $u \in C^\infty(\tilde \Omega \sm \bigcup_{h=1}^M \overline B_{\eta/32}(x_h);\SS^1)$ to find a sequence of piecewise constant functions $u_n \in \PC_{\lambda_n}(\SS^1)$ such that
	 \begin{gather}
	 u_n \to u \text{ strongly in } L^1\big(\Omega \sm \bigcup_{h=1}^M \overline B_{\eta/16}(x_h)\big) \, ,  \label{eq2:convergence of un} \\
	\limsup_{n \to +\infty}  \integral{J_{u_n} \cap O^{\lambda_n}}{ \geo(u_n^-,u_n^+) |\nu_{u_n}|_1}{ \d \H^1} \leq \integral{\Omega}{|\nabla u|_{2,1}}{\d x}  \label{eq2:limit of piecewise} \, ,
\end{gather}
where $O^{\lambda_n}$ is the union of half-open squares $I_{\lambda_n}(\lambda_n z)$, with $z \in \ZZ^2$, that intersect $\Omega \sm \bigcup_{h=1}^M \overline B_{\eta/16}(x_h)$. Note that, since $B_{\eta/16}(x_h) \subset Q_0(\lambda_n,x_h)$, 
\begin{equation*}
	\Omega \sm \bigcup_{h=1}^M Q_0(\lambda_n,x_h) \subset \Omega \sm \bigcup_{h=1}^M B_{\eta/16}(x_h) \subset  O^{\lambda_n} \subset \tilde \Omega \sm \bigcup_{h=1}^M \overline B_{\eta/32}(x_h) \, .
\end{equation*}
For a detailed proof of this discretization result see~\cite[Lemma~4.13]{CicOrlRuf}. 

We consider a recovery sequence $u_\e' \in \PC_\e(\S_\e)$ for $u_n$ satisfying $u_\e' \to u_n$ strongly in $L^1(O^{\lambda_n})$ and 
\begin{equation} \label{eq2:estimate outside Qlambda}
	\limsup_{\e \to 0} \frac{1}{\e\theta_\e} E_\e \Big(u_\e' ; \Omega \sm \bigcup_{h=1}^M  Q_0(\lambda_n,x_h) \Big) \leq  \integral{J_{u_n} \cap O^{\lambda_n}}{ \geo(u_n^-,u_n^+) |\nu_{u_n}|_1}{ \d \H^1}.
\end{equation}
The recovery sequence $u_\e'$ is defined as in the regime $\e \ll \theta_\e \ll \e |\log \e|$ in the case of no vortex-like singularities, exploiting the piecewise constant structure of $u_n$. The details of this construction can be found in~\cite[Proposition~4.16]{CicOrlRuf}.

\ul{Step 8} (Joining the two constructions) A careful dyadic decomposition of the square $Q(\lambda_n,x_h)$ leads to the construction of a spin field $u_\e \colon \e \ZZ^2 \cap Q(\lambda_n,x_h) \to \S_\e$ such that $u_\e = \mathfrak{P}_\e(u)$ on $B_{\eta/16}(x_h) \subset Q_0(\lambda_n,x_h)$,  while $u_\e(\e i) = u_\e'(\e i)$ if $\e i \in \e \ZZ^2 \cap Q(\lambda_n,x_h)$ satisfies $\dist(\e i, \de Q(\lambda_n,x_h)) \leq \e$, and 
 \begin{equation} \label{eq2:estimate inside Qlambda}
	\limsup_{\e \to 0} \frac{1}{\e \theta_\e} E_\e\Big(u_\e; \bigcup_{h=1}^M Q(\lambda_n,x_h)\Big)  \leq 2 \pi |\mu|(\Omega) + C \eta \, .
 \end{equation}
 Since $u_\e$ and $u_\e'$ agree close to $\de Q(\lambda,x_h)$, we define a global spin field $u_\e \in \PC_\e(\S_\e)$ (equal to $u_\e'$ outside $\bigcup_{h=1}^M Q(\lambda_n,x_h)$) satisfying, thanks to~\eqref{eq2:estimate outside Qlambda} and~\eqref{eq2:estimate inside Qlambda},
 \begin{equation} \label{eq2:split limsup}
	\begin{split}
		 \limsup_{\e \to 0} \frac{1}{\e \theta_\e} E_\e(u_\e) & \leq \limsup_{\e \to 0} \Big[ \frac{1}{\e\theta_\e} E_\e \Big(u_\e' ; \Omega \sm \bigcup_{h=1}^M  Q_0(\lambda_n,x_h) \Big) + \frac{1}{\e \theta_\e} E_\e\Big(u_\e; \bigcup_{h=1}^M Q(\lambda_n,x_h)\Big) \Big] \\
		 &  \leq \integral{J_{u_n} \cap O^{\lambda_n}}{ \geo(u_n^-,u_n^+) |\nu_{u_n}|_1}{ \d \H^1} +  2 \pi |\mu|(\Omega) + C \eta \, .
	\end{split}
 \end{equation}
 We refer to~\cite[Steps 2--4 in the proof of Proposition~4.22]{CicOrlRuf} for the details about this construction.

 \ul{Step 9} (Identifying the $L^1$-limit of $u_\e$) As $\e \to 0$, the spin fields $u_\e$ converge in $L^1(\Omega;\RR^2)$ to the map $\mathring u_n \in L^1(\Omega;\SS^1)$ given by 
 \begin{equation*}
	\mathring u_n(x) := \begin{cases}
		u_n(x) \, , & \text{if } x \in \Omega \sm \bigcup_{h=1}^M  Q(\lambda_n,x_h) \, , \\
		\big(\frac{x-x_h}{|x-x_h|}\big)^{d_h} , & \text{if } x \in Q_\infty(\lambda_n,x_h) \text{ for } h=1,\dots,M \, , \\
		u_0^{\lambda_n}(x) \, , & \text{if } x \in Q(\lambda_n,x_h) \sm Q_\infty(\lambda_n,x_h) \text{ for } h=1,\dots,M \, ,
	\end{cases}
 \end{equation*}
 where $u_0^{\lambda_n}$ is an $\SS^1$-valued map whose value is not relevant here and $Q_\infty(\lambda_n,x_h) := x_h + [(-2^{m(\lambda_n)}+2) \lambda_n , (2^{m(\lambda_n)}-2) \lambda_n] \subset Q(\lambda_n,x_h)$. (This notation is used in agreement to~\cite[Proposition~4.22]{CicOrlRuf}.) Since $Q(\lambda_n,x_h) \sm Q_\infty(\lambda_n,x_h)$ is a frame of width $2 \lambda_n$ and $Q(\lambda_n,x_h) \subset B_{\eta/2}(x_h)$, where $u(x) = \big(\frac{x-x_h}{|x-x_h|}\big)^{d_h}$ by our assumptions on $u$ in Step~5, by~\eqref{eq2:convergence of un} we get that 
 \begin{equation} \label{eq2:lambdalimit}
	\mathring u_n \to u \text{ in } L^1(\Omega;\RR^2) \, .
 \end{equation}
Hence, via a diagonal argument as $n\to+\infty$ we find a subsequence such that $u_\e \to u$ in $L^1(\Omega;\RR^2)$ and, by~\eqref{eq2:limit of piecewise} and \eqref{eq2:split limsup}, that
 \begin{equation*}
	\limsup_{\e \to 0} \frac{1}{\e \theta_\e} E_\e(u_\e) \leq  \integral{\Omega}{|\nabla u|_{2,1}}{\d x} + 2 \pi |\mu|(\Omega) + C \eta \, ,
 \end{equation*}
 which will give the claim after an additional diagonal argument as $\eta \to 0$. 
 
 \ul{Step 10} (Identifying the flat limit of $\mu_{u_\e}$)
 In order to implement the diagonal arguments proven in Step~9, we need to identify the flat limit of $\mu_{u_\e}$ for $n$ fixed. After the diagonal argument we will obtain the desired convergence $\mu_{u_\e} \flat \mu$. As this is the major difference in the regime $\theta_\e = \e |\log \e|$, we provide all the details. Since $\theta_{\e}=\e|\log \e|$, the energy bound~\eqref{eq2:split limsup} reads
 \begin{equation}\label{eq2:criticalscaling}
 \limsup_{\e \to 0}\frac{1}{\e^2|\log \e|}E_{\e}(u_{\e};\Omega)\leq \int_{J_{u_{n}}\cap O^{\lambda_n}}{ \geo(u_{n}^-,u_{n}^+)|\nu_{u_n}|_1}{\mathrm{d}\mathcal{H}^1}  + 2\pi |\mu|(\Omega) +  C \eta \,.
 \end{equation}
 Note that the left hand side agrees with the unconstrained scaled $XY$-model. In particular, by Proposition \ref{prop:XY classical} we deduce that there exists a measure $\mu_{n}\in\mathcal{M}_b(\Omega)$ of the form $\mu_{n}=\sum_{k=1}^{K_{n}}d_{k,n}\delta_{x_{k,n}}$ with $d_{k,n}\in\ZZ\setminus\{0\}$ such that, up to a subsequence, $\mu_{u_{\e}}\flat\mu_{n}$. Thus the (already proven) lower bound in Theorem~\ref{thm:theta equal e log}-{\em ii)}, the convergence $u_\e \to \mathring{u}_n$, and~\eqref{eq2:criticalscaling} yield
	\begin{align*}
	&\integral{\Omega}{|\nabla \mathring u_{n}|_{2,1} }{\d x} + |\DD^{(c)} \mathring u_{n}|_{2,1}(\Omega)+\mathcal{J}(\mu_{n},\mathring u_{n};\Omega)+2\pi|\mu_{n}|(\Omega)
	\\
	& \quad \leq \int_{J_{u_{n}}\cap O^{\lambda_n}}{\hspace{-1em}\geo(u_{n}^-,u_{n}^+)|\nu_{u_n}|_1}{\mathrm{d}\mathcal{H}^1} + 2\pi |\mu|(\Omega)+ C \eta \,.
	\end{align*}
	Since the last term in the above estimate is controlled via \eqref{eq2:limit of piecewise}, we deduce that $|\mu_{n}|(\Omega)$ is equibounded in $n$. In particular, up to a subsequence, there exists a measure $\mu_0\in\mathcal{M}_b(\Omega)$ that has the structure $\mu_0=\sum_{k=1}^{K}d_k\delta_{x_k}$ for some $d_k\in\ZZ\setminus\{0\}$ such that $\mu_{n}\flat\mu_0$ (and weakly* in the sense of measures). We next want to use a lower semicontinuity property of the left hand side. However, due to the mixed term $\mathcal{J}(\mu,u;\Omega)$, this is not straightforward, so we estimate the left hand side from below with a negligible error when $\lambda_n\to 0$. Indeed, by~\eqref{eq2:J greater than jump} we have that
	\begin{equation*}
	\mathcal{J}(\mu_{n},\mathring u_{n};\Omega)\geq \integral{J_{\mathring u_{n}} \cap \Omega}{\geo(\mathring u_{n}^-,\mathring u_{n}^+)|\nu_{\mathring u_{n}}|_1}{\mathrm{d}\mathcal{H}^1}.
	\end{equation*} 
	Inserting this lower bound in the previous estimate we obtain that
	\begin{align*}
	&\integral{\Omega}{|\nabla  \mathring u_{n}|_{2,1} }{\d x} + |\DD^{(c)} \mathring u_{n}|_{2,1}(\Omega)+\integral{J_{\mathring u_n} \cap \Omega}{\geo(\mathring u_{n}^-,\mathring u_{n}^+)|\nu_{\mathring u_{n}}|_1}{\mathrm{d}\mathcal{H}^1}+2\pi|\mu_{n}|(\Omega)
	\\
	& \quad \leq  C\,\eta+2\pi |\mu|(\Omega)+\int_{J_{u_{n}}\cap O^{\lambda_n}}{\hspace{-1em} \geo(u_{n}^-,u_{n}^+)|\nu_{u_n}|_1}{\mathrm{d}\mathcal{H}^1}.
	\end{align*}
	The left hand side is lower semicontinuous with respect to the $L^1(\Omega;\RR^2)$-convergence of $\mathring u_n$ (see Proposition~\ref{prop:H is lsc} below) and the weak*-convergence of $\mu_{n}$. The limit of the right hand side is given by~\eqref{eq2:limit of piecewise}. Due to~\eqref{eq2:lambdalimit} and the fact that $u\in W^{1,1}(\Omega;\SS^1)$ (recall the standing assumption in Step~5) we conclude that
	\begin{equation*}
	\int_{\Omega}|\nabla u|_{2,1}\,\mathrm{d}x+2\pi|\mu_0|\leq C\eta+\int_{\Omega}|\nabla u|_{2,1}\,\mathrm{d}x+2\pi |\mu|(\Omega)\,,
	\end{equation*}
	which implies that $|\mu_0|(\Omega)\leq |\mu|(\Omega)$ (recall that $\eta<\eta_0$ can be chosen arbitrary small, while the constant $C$ is bounded uniformly).
	
	We finish the proof by showing that all measures~$\mu_{n}$ have mass $1$ in a uniform neighborhood of each of the points $x_h$ given by the target measure $\mu=\sum_{h=1}^M\deg(u)(x_h)\delta_{x_h}$. Indeed, by Step~6, $u_{\e}=\mathfrak{P}_\e(u)$ on each $B_{\eta/16}(x_h)$. Due to~\eqref{eq2:step6final} we have for $\e$ small enough
	\begin{equation*}
	\frac{1}{\e^2}E_{\e}(\mathfrak{P}_{\e}(u);B_{\eta}(x_h))\leq 2C\eta\frac{\theta_{\e}}{\e}+2\pi|\log \e|\leq C|\log \e|\,.
	\end{equation*}
	This allows us to apply~\cite[Proposition 5.2]{Ali-Cic-Pon}, so that the flat convergence of discrete vorticities is equivalent to the flat convergence of the (normalized) Jacobians of the piecewise affine interpolations. Denote by $\widehat v_\e$ and $\widehat u(\e)$ the piecewise affine interpolations associated to $\mathfrak{P}_\e(u)$ and $u$ on $B_{\eta/16}(x_h)$, respectively. (We adopt the notation $\widehat u(\e)$ to stress that the interpolated function is independent of $\e$.) We have that
    \begin{equation*}  
        \begin{split}
            & \|\widehat{v}_{\e}-\widehat{u}(\e)\|_{L^2(B_{\eta/16}(x_h))}\left(\|\nabla\widehat{v}_{\e}\|_{L^2(B_{\eta/16}(x_h))}+\|\nabla\widehat{u}(\e)\|_{L^2(B_{\eta/16}(x_h))}\right)\\ 
            & \leq C\eta \theta_{\e}\left(\frac{1}{\e^2}E_{\e}(\mathfrak{P}_\e(u);B_\eta(x_h))+\frac{1}{\e^2}E_{\e}\big(\big(\tfrac{x-x_h}{|x-x_h|}\big)^{d_h};B_\eta(x_h)\big)\right)^{\! \frac{1}{2}}\leq C \theta_{\e}|\log \e|^{\frac{1}{2}}\,.
        \end{split}
	\end{equation*} 
	As $\theta_{\e}=\e|\log \e|$, the right hand side vanishes when $\e\to 0$. Hence \cite[Lemma 3.1]{Ali-Cic-Pon} yields that $\mathrm{J}\widehat{v}_{\e}-\mathrm{J}\widehat{u}(\e)\flat 0$ on $B_{\eta/16}(x_h)$. Since $u=\big(\tfrac{x-x_h}{|x-x_h|}\big)^{\pm 1}$ on $B_{\eta/16}(x_h)$, Step 1 of the proof of \cite[Theorem 5.1 (ii)]{Ali-Cic} implies that $\tfrac{1}{\pi}\mathrm{J}\widehat{u}(\e)\flat \deg(u)(x_h)\delta_{x_h}$. Choosing an arbitrary $\varphi\in C_c^{0,1}(B_{\eta/16}(x_h))$, the above arguments imply 
	\begin{equation*}
	\langle \mu_{n}\mres B_{\eta/16}(x_h),\varphi\rangle=\lim_{\e \to 0}\langle \mu_{u_{\e}},\varphi\rangle = \lim_{\e \to 0}\langle \tfrac{1}{\pi}\mathrm{J}\widehat{v}_{\e},\varphi\rangle=\lim_{\e \to 0}\langle \tfrac{1}{\pi}\mathrm{J}\widehat{u}(\e),\varphi\rangle=\deg(u)(x_h)\varphi(x_h) \, .
	\end{equation*}
	Letting $n\to +\infty$ in the above equality we infer that $\mu_0\mres B_{\eta/16}(x_h)=\deg(u)(x_h)\delta_{x_h}$. Now consider the decomposition of $\mu_0$ into the mutually singular measures
	\begin{equation*}
	\mu_0=\sum_{h=1}^M\deg(u)(x_h)\delta_{x_h}+\mu_0\mres\Big(\Omega\setminus\bigcup_{h=1}^M B_{\eta/16}(x_h)\Big)\,.
	\end{equation*} 
	From mutual singularity we deduce that
	\begin{align*}
	|\mu|(\Omega)&\geq |\mu_0|(\Omega)=\sum_{h=1}^M|\deg(u)(x_h)|+\Big|\mu_0\mres\Big(\Omega\setminus\bigcup_{h=1}^M B_{\eta/16}(x_h)\Big)\Big|
	\\
	&\geq |\mu|(\Omega)+\Big|\mu_0\mres\Big(\Omega\setminus\bigcup_{h=1}^M B_{\eta/16}(x_h)\Big)\Big|\,.
	\end{align*}
	Hence $\mu_0\mres\big(\Omega\setminus\bigcup_{h=1}^M B_{\eta/16}(x_h)\big)=0$ and therefore $\mu_0=\sum_{h=1}^M\deg(u)(x_h)\delta_{x_h}=\mu$. In conclusion, we have proved that $\mu_{u_{\e}}\flat\mu_{n}$ as $\e \to 0$ and $\mu_{n}\flat\mu$ as $n \to +\infty$ which justifies the diagonal arguments in Step 9.
\end{proof}

\section{Proofs in the regime \texorpdfstring{$\e |\log \e| \ll \theta_\e \ll 1$}{}} \label{sec:e log smaller theta}
In the present scaling regime the discrete vorticity measures $\mu_{u_{\e}}$ for sequences with bounded energy are not necessarily compact. Hence we cannot use the parametric integral as a comparison, but we will work directly with the spin variable $u_{\e}$. We recall the following lower-semicontinuity result proven in the $d$-dimensional case in \cite[Lemma 3.2]{CicOrlRuf3} via a slicing argument.
\begin{proposition} \label{prop:H is lsc}
For every open set $A \subset \Omega$ define the functional $E(\cdot;A):L^1(A;\RR^2)\to [0,+\infty]$ with domain $BV(A;\SS^1)$, on which it is given by 
\begin{equation*}
E(u;A) := \integral{A}{|\nabla w|_{2,1}}{\d x} + |\DD^{(c)} w|_{2,1}(A) + \hspace{-0.3em}\integral{J_{u} \cap A}{\geo\big(u^+,u^-\big)|\nu_u|_1}{\mathrm{d}\mathcal{H}^1} \,.
\end{equation*}
Then $E(\cdot;A)$ is lower semicontinuous with respect to strong convergence in $L^1(A;\RR^2)$.
\end{proposition}

\begin{proof}[Proof of Theorem~\ref{thm:e log smaller theta}]
The compactness result in {\em i)} follows by Lemma~\ref{lemma:bound with parametric integral} as in the proof of Proposition~\ref{prop:current compactness}. Indeed, Lemma~\ref{lemma:bound with parametric integral} with $\sigma = \frac{1}{2}$ yields that $u_{\e}$ is bounded in $BV(A;\SS^1)$ for every $A \subset \subset \Omega$.

To prove {\em ii)}, it suffices to consider $u\in BV(\Omega;\SS^1)$ and a sequence $u_\e\in\mathcal{PC}_{\e}(\S_{\e})$ such that $u_\e \to u$ strongly in $L^1(\Omega;\RR^2)$ and
\begin{equation*}
\liminf_{\e \to 0}\frac{1}{\e\theta_{\e}}E_{\e}(u_{\e})\leq C\,.
\end{equation*}
Let us fix an open set $A\compact\Omega$ and $\sigma>0$. By Lemma~\ref{lemma:bound with parametric integral} , for $\e$ small enough we have that 
\begin{equation*}
\frac{1}{\e\theta_{\e}}E_\e(u_\e) \geq (1 - \sigma) \integral{J_{u_\e} \cap A}{\geo(u_\e^+,u_\e^-)|\nu_{u_\e}|_1}{\d \H^1} = (1 - \sigma) E(u_\e; A)  \, .
\end{equation*}
Hence Proposition~\ref{prop:H is lsc} yields that 
\begin{equation*}
\liminf_{\e \to 0} \frac{1}{\e\theta_{\e}}E_\e(u_\e) \geq (1-\sigma) E(u; A)\, .
\end{equation*}
We conclude the proof of the lower bound letting $\sigma \to 0$ and $A \nearrow \Omega$.

\medskip
 	
In order to show the $\Gamma$-limsup inequality in {\em iii)}, we first remark that, following the approximation procedure in \cite[Proof of Proposition 4.3 (Step 1-3)]{CicOrlRuf3}, it suffices to prove the upper bound for a target function $u\in C^{\infty}(\tilde \Omega \sm V;\SS^1) \cap W^{1,1}(\tilde \Omega;\SS^1)$, where $\widetilde{\Omega}\supset\supset\Omega$ has Lipschitz-boundary and  $V=\{x_1,\ldots,x_M\}\subset\widetilde{\Omega}$ is a finite set. Moreover, following Steps~3-5 in the proof of Theorem~\ref{thm:theta equal e log}-{\em iii)} above, we can assume, without loss of generality, that $V \subset \lambda_n \ZZ^2 \cap \tilde \Omega$ with $\lambda_n = 2^{-n}$, that there exists $\Omega \subset \subset \Omega' \subset \subset \tilde \Omega$ such that $V \cap (\Omega' \sm \Omega) = \emptyset$, and that there exists $\eta_0>0$ such that for every $x_i \in V$ we have $u(x) = \big(\frac{x-x_i}{|x-x_i|}\big)^{\pm 1}$ for $x \in \ol B_{\eta_0}(x_i)$. 
 
We are finally in a position to apply \cite[Proposition~4.22]{CicOrlRuf}  (which is valid also in the case $\e \ll \theta_\e$) to the modified field $u  \in C^\infty(\Omega' \sm V;\SS^1) \cap W^{1,1}(\Omega';\SS^1)$. It gives the existence of a recovery sequence $u_\e \in \PC_\e(\S_\e)$ such that $u_\e \to u$ in $L^1(\Omega;\RR^2)$ and 
\begin{equation*}
	\limsup_{\e \to 0}\Big( \frac{1}{\e \theta_\e} E_\e(u_\e) - 2\pi |\mu|(\Omega) |\log \e| \frac{\e}{\theta_\e}  \Big) \leq \int_{\Omega}|\nabla u|_{2,1}\,\mathrm{d} x \, .
\end{equation*} 
Since $ |\log \e| \frac{\e}{\theta_\e} \to 0$, we obtain that
\begin{equation*}
\Gamma\hbox{-}\limsup_{\e \to 0} \frac{1}{\e\theta_{\e}}E_{\e}(u) \leq  \integral{ \Omega}{|\nabla u|_{2,1} }{\d x}\,.
\end{equation*}
Note that the right hand side coincides with the functional claimed to be the $\Gamma$-limit in Theorem \ref{thm:e log smaller theta} since $u\in W^{1,1}(\Omega;\SS^1)$.
\end{proof}
\section{Proofs in the regime \texorpdfstring{$\theta_\e \ll \e$}{}}
We now come to the scaling regime which yields a discretization of $\SS^1$ that is fine enough to commit asymptotically no error compared to the $XY$-model up to the first order development. Throughout this section we shall always assume that
\begin{equation} \label{eq:theta less eps}
\theta_\e \ll \e \, .
\end{equation}
Moreover, we will use the following elementary estimate: for any $x,y\in\RR^2\setminus\{0\}$ it holds that
\begin{equation}\label{eq:xmoxcontinuous}
\left|\frac{x}{|x|}-\frac{y}{|y|}\right|\leq 2\frac{|x-y|}{|y|}.	
\end{equation}
\subsection{Renormalized and core energy} We recall that the {\em renormalized energy} corresponding to the configuration of vortices $\mu = \sum_{h=1}^M d_h \delta_{x_h}$ is defined by 
\begin{equation*} 
	\WW(\mu) = -2 \pi \sum_{h\neq k} d_h d_k \log |x_h - x_k| - 2 \pi \sum_{h} d_h R_0(x_h) \, ,
\end{equation*}
where $R_0$ is harmonic in $\Omega$ and $R_0(x) = - \sum_{h=1}^M d_h \log |x - x_h|$ for $x \in \de \Omega$. The renormalized energy can also be recast as 
\begin{equation} \label{eq:renormalized energy as limit}
	\WW(\mu) = \lim_{\eta \to 0} \big[ \tilde m(\eta,\mu) -  2  \pi |\mu|(\Omega) |\log \eta | \big] \, ,
\end{equation}
where
\begin{equation} \label{eq:min for renormalized energy}
	\tilde m(\eta,\mu) := \min \Big\{   \integral{\Omega_\mu^\eta}{|\nabla w|^2}{\d x} \colon w(x) = \alpha_h \odot \big(\tfrac{x-x_h}{|x-x_h|}\big)^{d_h} \text{ for } x \in \de B_\eta(x_h)\, , |\alpha_h| = 1 \Big\} \, .
\end{equation}

To define the {\it core energy}, we introduce the discrete minimization problem in a ball $B_r$ 
\begin{equation} \label{eq:def of core}
	\gamma(\e,r) := \min \Big\{ \frac{1}{\e^2}E_\e(v;B_r)  \ : \ \ v \colon \e \ZZ^2 \cap B_r \to \SS^1 , \  v(x) = \tfrac{x}{|x|} \text{ for } x \in \de_\e B_r \Big\} \, ,
\end{equation}
where $\de_\e B_r$ is the discrete boundary of $B_r$, defined for a general open set by 
\begin{equation*}
	\de_\e A = \{ \e i \in \e \ZZ^2 \cap A \colon \dist(\e i,\partial A)\leq\e \} \, .
\end{equation*}
Note that $\de_\e B_r \subset \e \ZZ^2\cap B_r \sm B_{r-\e}$. Then the {\em core energy} of a vortex is the number $\gamma$ given by the following lemma.
The result is analogous to~\cite[Theorem~4.1]{Ali-DL-Gar-Pon} with some differences: here we consider $r_\e \to 0$ depending on $\e$ and we use a different notion of discrete boundary of a set. The modifications in the proof are minor, but we give the details for the convenience of the reader. 
 
\begin{lemma}  \label{lemma:core energy eps dependent}
Let $r_\e$ be a family of radii such that $\e \ll r_\e \leq C$. Then there exists
\begin{equation} \label{eq:new core energy}
	\lim_{\e \to 0} \Big[ \gamma(\e,r_\e) - 2 \pi \big| \log \tfrac{\e}{r_\e} \big| \Big] =: \gamma \in \RR \, ,
\end{equation}
where $\gamma$ is independent of  the sequence $r_\e$.
\end{lemma}
\begin{proof}
We introduce the function
\begin{equation*}
I(t) = \min \Big\{ E_1(v;B_\frac{1}{t}) \ :  \ v \colon \ZZ^2 \cap B_\frac{1}{t} \to \SS^1 , \, v(x) = \tfrac{x}{|x|} \text{ for } x \in \de_1 B_\frac{1}{t} \Big\} \, .
\end{equation*}
Let us show that 
\begin{equation} \label{eq:monotonicity of I}
I(t_1)  \leq I(t_2) + 2 \pi \log \tfrac{t_2}{t_1} + \varrho_{t_2} \quad \text{for } 0 < t_1 \leq t_2 \, ,
\end{equation}
where $\varrho_{t_2}$ is a generic sequence (which may change from line to line) satisfying $\varrho_{t_2} \to 0$ as $t_2 \to 0$. To this end, let $v_2 \colon \ZZ^2 \cap B_\frac{1}{t_2} \to \SS^1$ be such that $v_2(x) = \tfrac{x}{|x|}$ for $x \in \de_1 B_\frac{1}{t_2}$ and $E_1(v_2;B_\frac{1}{t_2}) = I(t_2)$. We extend $v_2$ to $B_\frac{1}{t_1}$ setting
\begin{equation*}
	v_1(i) := 
	\begin{cases}
		v_2(i) \, , & \text{if } i \in \ZZ^2\cap B_\frac{1}{t_2} \, , \\
		\tfrac{i}{|i|} \, , & \text{if } i \in \ZZ^2\cap B_\frac{1}{t_1} \sm B_\frac{1}{t_2} \, .
	\end{cases}
\end{equation*}
To reduce notation, we define $A(t_1,t_2):=B_{\tfrac{1}{t_1}}\setminus\overline{B}_{\tfrac{1}{t_2}-2}$. Next, note that if $i \in B_\frac{1}{t_2}$ and $j \notin B_\frac{1}{t_2}$ with $|i-j|=1$, then $|i| \geq \tfrac{1}{t_2} - 1 > \tfrac{1}{t_2} - 2$. Hence
\begin{equation*}
	I(t_1) \leq E_1(v_1;B_\frac{1}{t_1})  \leq E_1(v_2;B_\frac{1}{t_2}) + E_1(v_1;A(t_1,t_2))  = I(t_2) + \frac{1}{2} \sum_{\substack{\langle i,j \rangle \\ i,j \in A(t_1,t_2)}} \big| \tfrac{i}{|i|} - \tfrac{j}{|j|} \big|^2 .
\end{equation*}
To control the last sum, we derive an estimate for the finite differences away from the singularity. Set $u(x)=\frac{x}{|x|}$. For any $t\in[0,1]$ and $i,j\in\ZZ^2$ with $|i-j|=1$ we have
\begin{equation*}
|(1-t) \e i+t\e j|\geq |\e i|-\e \, .
\end{equation*}
Hence, by the regularity of $u$ in $\RR^2\setminus\{0\}$, for any $\e i,\e j\in\e\ZZ^2\setminus B_{2\e}$ with $|i-j|=1$, 
\begin{equation*}
|u(\e i)-u(\e j)|\leq \int_0^1|\nabla u(t \e i+(1-t)\e j)(\e i-\e j)|\,\mathrm{d}t \, .
\end{equation*}
Since $i-j\in\{\pm e_1,\pm e_2\}$, a direct computation yields the two cases
\begin{equation}\label{eq:estimatecases}
|u(\e i)-u(\e j)|\leq
\begin{cases}\displaystyle
\int_0^1 \frac{|i\cdot e_2|}{|t i+(1-t) j|^2}\,\mathrm{d}t &\mbox{if $(i-j)\parallel e_1$} \, ,
\\
\\
\displaystyle\int_0^1 \frac{|i \cdot e_1|}{|t i+(1-t) j|^2}\,\mathrm{d}t &\mbox{if $(i-j)\parallel e_2$} \, ,
\end{cases}
\end{equation}
Next note that $|ti+(1-t)j|\geq |i|-1\geq \max\{|i\cdot e_1|,|i\cdot e_2|\}-1$. The left hand side is non-negative if $i\neq 0$, so that we can take the square of this inequality. Using Jensen's inequality, we infer from \eqref{eq:estimatecases} that
\begin{equation}\label{eq:estimatenocases}
|u(\e i)-u(\e j)|^2\leq \frac{k^2}{(k-1)^4}\Big|_{k=\max\{|i\cdot e_1|,|i\cdot e_2|\}} \, .
\end{equation}
We shall use both \eqref{eq:estimatecases} and \eqref{eq:estimatenocases} to estimate the sum of interactions in $A(t_1,t_2)$. To do so, we split the annulus $A(t_1,t_2)$ using trimmed quadrants defined as follows: given a tuple of signs $s=(s_1,s_2)\in\{(+,+), (-,+), (-,-), (+,-)\}$ and $n\in\mathbb{N}$ we define the trimmed quadrants $\mathcal{Q}_n^{s}$ as
\begin{equation} \label{eq:trimmed quadrants}
\mathcal{Q}_n^{s}:=\{x\in\RR^2:s_1 \, x \cdot e_1 \geq n,\ s_2\,  x\cdot e_2 \geq n\}\,.
\end{equation}
\begin{figure}[H]
	\hspace{0cm}\scalebox{0.8}{
        \includegraphics{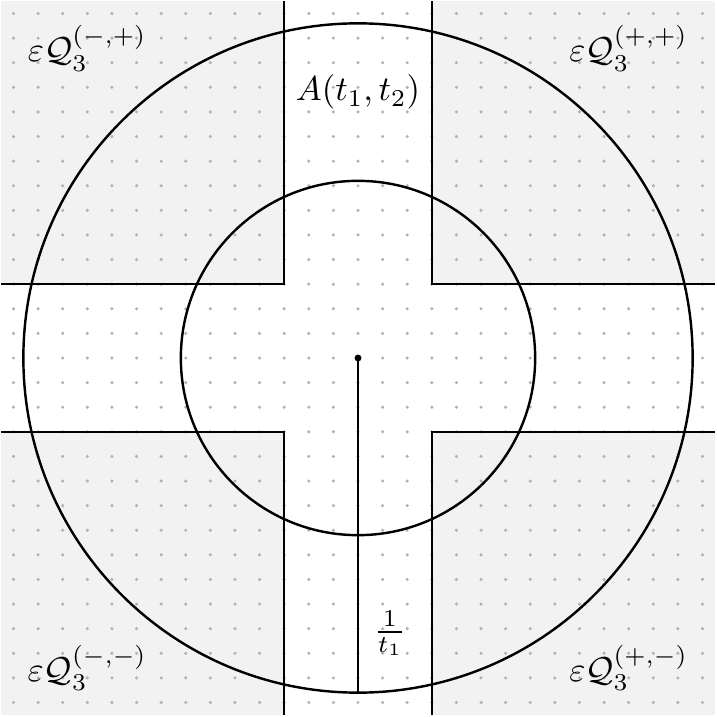}
	}
	\caption{Schematic illustration of the trimmed quadrants $\e \mathcal{Q}^s_3$. }
	\label{fig:trimmingannulus}
\end{figure}
 
Fix $n=3$. We then consider interactions $\langle i,j\rangle$ where both points belong to one trimmed quadrant and the remaining interactions. For the latter we use the estimate \eqref{eq:estimatenocases}, noting that $\max\{|i\cdot e_1|,|i\cdot e_2|\}\geq \tfrac{1}{\sqrt{2}}|i| \geq \tfrac{1}{\sqrt{2}}t_2$ and that for $t_2$ small enough, i.e., the inner circle of the annulus large enough, the interactions outside the trimmed quadrants can be counted along $20$ lines parallel to one of the coordinate axes in a way that the maximal component strictly increases along the line (cf. Figure \ref{fig:trimmingannulus}). Summing over all pairs of signs $s \in \{(+,+),(-,+),(+,-),(-,-)\}$ then yields
\begin{align*}
\frac{1}{2}\sum_{\substack{\langle i,j \rangle \\ i,j \in A(t_1,t_2)}} \big| \tfrac{i}{|i|} - \tfrac{j}{|j|} \big|^2 & \leq \frac{1}{2}\sum_s\sum_{\substack{\langle i,j \rangle \\ i,j \in \mathcal{Q}^s_3 \cap A(t_1,t_2)}}\big| \tfrac{i}{|i|} - \tfrac{j}{|j|} \big|^2 + C \sum_{k = \lfloor \frac{1}{\sqrt{2}t_2}\rfloor}^{\lceil \frac{1}{t_1} \rceil} \frac{k^2}{(k-1)^4}
\\
& \leq \frac{1}{2}\sum_s\sum_{\substack{\langle i,j \rangle \\ i,j \in \mathcal{Q}^s_3 \cap A(t_1,t_2)}}\big| \tfrac{i}{|i|} - \tfrac{j}{|j|} \big|^2 + \rho_{t_2}\,,
\end{align*}
where we used that the series $\sum_{k=2}^{\infty}\tfrac{k^2}{(k-1)^4}$ converges. The contributions on the trimmed cubes have to be treated more carefully since we need the precise pre-factor $2\pi$
 in \eqref{eq:monotonicity of I}. The idea is two switch from the discrete lattice $\ZZ^2$ to a continuum environment that leads to an integral. We have
 \begin{align*}
 I_s:=\frac{1}{2}\sum_{\substack{\langle i,j \rangle \\ i,j \in \mathcal{Q}^s_3 \cap A(t_1,t_2)}}\big| \tfrac{i}{|i|} - \tfrac{j}{|j|} \big|^2=\sum_{\substack{i\in\ZZ^2 \\ i\in\mathcal{Q}^s_3 \cap A(t_1,t_2)}	}\big| \tfrac{i+s_1e_1}{|i+s_1e_1|} - \tfrac{i}{|i|} \big|^2+\big| \tfrac{i+s_2e_2}{|i+s_2e_2|} - \tfrac{i}{|i|} \big|^2.
 \end{align*} 
For each term on the right hand side we apply \eqref{eq:estimatecases} noting that $|t(i+s_re_r)+(1-t)i|=|i+ts_re_r|\geq |i|$ for $i\in \mathcal{Q}^s_3$, $s=(s_1,s_2)$, and $r\in\{1,2\}$, so that by Jensen's inequality
\begin{equation*}
\big| \tfrac{i+s_1e_1}{|i+s_1e_1|}  - \tfrac{i}{|i|} \big|^2+\big|  \tfrac{i+s_2e_2}{|i+s_2e_2|}  - \tfrac{i}{|i|} \big|^2\leq \frac{|i\cdot e_1|^2+|i\cdot e_2|^2}{|i|^4}=\frac{1}{|i|^2}	
\end{equation*}
Note that for $i\in\ZZ^2\cap\mathcal{Q}^s_3$ it holds that $\frac{1}{|x|^2}\geq \frac{1}{|i|^2}$ for all $x\in i-s_1e_1-s_2e_2+[-\tfrac{1}{2},\tfrac{1}{2}]^2$. Since $i-s_1e_1-s_2e_2+[-\tfrac{1}{2},\tfrac{1}{2}]^2\in\mathcal{Q}_2^s$, we can control $I_s$ by
\begin{equation}\label{eq:trickytrimmer}
I_s\leq \sum_{\substack{i\in\ZZ^2 \\ i\in\mathcal{Q}^s_3 \cap A(t_1,t_2)}	}\frac{1}{|i|^2}\leq \integral{\mathcal{Q}_2^s\cap B_\frac{1}{t_1} \sm B_{\frac{1}{t_2}- 4} }{\frac{1}{|x|^2}}{\mathrm{d}x}.	
\end{equation}
Summing this estimate over all couples of signs $s$, we infer that
\begin{equation*}
\sum_s I_s\leq \integral{B_\frac{1}{t_1} \sm \ol B_{\frac{1}{t_2}- 4}}{\frac{1}{|x|^2}}{\d x} + \varrho_{t_2} \leq 2 \pi \log\tfrac{t_2}{t_1} + \varrho_{t_2}
\end{equation*}
where we also used that  by the mean value theorem $|\log\tfrac{1}{t_2}-\log(\tfrac{1}{t_2}-4)|\leq 4 |\tfrac{t_2}{1-4 t_2}|$. 	This proves~\eqref{eq:monotonicity of I}. As a consequence, the limit $\lim_{t \to 0} \big[ I(t) - 2\pi|\log t|\big] =: \gamma$ exists. Since $\gamma(\e,r_\e) = I\big(\tfrac{\e}{r_\e}\big)$, it only remains to show that $\gamma\neq -\infty$. To this end, we show that the boundary conditions in the definition of $\gamma(\e,1)$ force concentration of the Jacobians, so that we can use localized lower bounds. Let $v_{\e} \colon \e\ZZ^2\cap B_1\to\SS^1$ be an admissible minimizer for the problem defining $\gamma(\e,1)$ and extend it to $\e\ZZ^2\setminus B_1$ via $v_{\e}(\e i)=\tfrac{\e i}{|\e i|}$. Then, using the boundary conditions imposed on $v_{\e}$ and \eqref{eq:xmoxcontinuous}, we deduce that
\begin{align}\label{eq:splittingtrick}
\frac{1}{\e^2}E_{\e}(v_{\e};B_3)&\leq \frac{1}{\e^2}E_{\e}(v_{\e};B_1)+\frac{1}{2\e^2}\sum_{\substack{\nn\\ \e i,\e j\in B_3\setminus B_{1/2}}}\e^2\left|\frac{\e i}{|\e i|}-\frac{\e j}{|\e j|}\right|^2\nonumber
\\
&\leq \gamma(\e, 1)+\frac{C}{\e^2} \hspace{-1.5em}\sum_{\substack{\nn\\ \e i,\e j\in B_3\setminus B_{1/2}}} \hspace{-1.5em} \e^4\leq \gamma(\e,1)+C\,.
\end{align}
Since we already proved that $\gamma(\e,1)-2\pi|\log \e |$ remains bounded from above  when $\e\to 0$, Proposition~\ref{prop:XY classical} implies that (up to a subsequence) $\mu_{v_{\e}}\mres B_3\flat \mu$ for some $\mu=d_1\delta_{x_1}$ with $d_1\in\ZZ$ and $x_1\in B_3$. We claim that $\mu\neq 0$. Indeed, let us denote by $\widehat{v}_{\e}$ the piecewise affine interpolation of $v_{\e}$ and let $\eta \colon [0,3]\to\RR$ be the piecewise affine function such that $\eta=1$ on $[0,1]$, $\eta=0$ on $[2,3]$ and $\eta$ is affine on $[1,2]$. Then define the Lipschitz function $\varphi\in C_c^{0,1}(B_3)$ via $\varphi(x)=\eta(|x|)$. Using the flat convergence of $\mu_{v_\e}$, which transfers to the scaled Jacobian $\pi^{-1}\mathrm{J}\widehat{v}_{\e}$ due to \cite[Proposition 5.2]{Ali-Cic-Pon}, we infer that 
\begin{equation*}
\langle\mu,\varphi\rangle=\lim_{\e\to 0}\frac{1}{\pi}\int_{B_3}\mathrm{J}\widehat{v}_{\e}\,\varphi\,\mathrm{d}x=-\lim_{\e \to 0}\frac{1}{2\pi}\int_{B_2\setminus B_1}\begin{pmatrix}(v_{\e})_1\partial_2 (v_{\e})_2-(v_{\e})_2\partial_2 (v_{\e})_1 \\-(v_{\e})_1\partial_1 (v_{\e})_2+(v_{\e})_2\partial_1 (v_{\e})_1\end{pmatrix} \cdot \nabla\varphi\,\mathrm{d}x\,,
\end{equation*}
where in the last equality we integrated by parts due to the fact that in dimension two the Jacobian can be written as a divergence (in two ways). Note that on $B_2\setminus B_1$ the function $v_{\e}$ agrees with the discrete version of $x/|x|$. Hence one can pass to the limit in $\e$ as $\widehat{v}_{\e}$ converges to $x/|x|$ weakly in $H^1(B_2\setminus B_1)$. Moreover, it holds that $\nabla\varphi(x)=-x/|x|$ a.e. in $B_2\setminus B_1$. An explicit computation shows that
\begin{equation}\label{eq:explicitdegree}
\langle \mu,\varphi\rangle = \frac{1}{2\pi}\int_{B_2\setminus B_1}|x|^{-1}\,\mathrm{d}x=1\,.
\end{equation}
Consequently $\mu\mres B_3= \delta_{x_1}$. Let $\sigma<\frac{1}{2}\dist(x_1,\partial B_3)$. Then \cite[Theorem 3.1(ii)]{Ali-DL-Gar-Pon} yields
\begin{align*}
\liminf_{\e \to 0}\left(\frac{1}{\e^2}E_{\e}(v_{\e};B_3)-2\pi |\log \e |\right)&\geq\liminf_{\e \to 0}\left(\frac{1}{\e^2}E_{\e}(v_{\e};B_{\sigma}(x_1))-2\pi\log\tfrac{\sigma}{\e}\right)+2\pi \log \sigma 
\\
&\geq -C+2\pi\log \sigma 
\end{align*}
for some constant $C$. Combining this lower bound with \eqref{eq:splittingtrick} yields that
\begin{equation*}
-C+2\pi\log \sigma \leq \lim_{\e \to 0}\left(\gamma(\e,1)-2\pi|\log \e |\right)+C\,.
\end{equation*}
This shows that $\gamma>-\infty$ and concludes the proof. 
\end{proof}
Below we will use a shifted version of $\gamma(\e,r_{\e})$. More precisely, given $x_0\in\Omega$, set
\begin{equation*}
\gamma_{x_0}(\e,r) := \min \Big\{ \frac{1}{\e^2}E_\e(v;B_r(x_0)): \, v \colon \e \ZZ^2 \cap B_r(x_0) \to \SS^1 , \, v(x) = \tfrac{x-x_0}{|x-x_0|} \text{ on } \de_\e B_r(x_0) \Big\} \, .
\end{equation*}
As shown in the lemma below the asymptotic behavior does not depend on $x_0$.
\begin{lemma}\label{lemma:nonzerodcore}
Let $\gamma\in\RR$ be given by Lemma \ref{lemma:core energy eps dependent} and let $\e\ll r_{\e}\leq C$. Then it holds that
\begin{equation*}
\lim_{\e \to 0}\left(\gamma_{x_0}(\e,r_{\e})-2\pi \big|\log\tfrac{\e}{r_{\e}} \big| \right)=\gamma\,.
\end{equation*}
\end{lemma}
\begin{proof}
Consider a point $x_{\e}\in\e\ZZ^2$ such that $|x_0-x_{\e}|\leq 2\e$. Then, given a minimizer $v \colon \e\ZZ^2\cap B_{r_{\e}-4\e}\to\SS^1$  for the problem defining $\gamma(\e,r_{\e}-4\e)$ (extended via the boundary conditions on $\e\ZZ^2\setminus B_{r_{\e}-4\e})$ we define $\tilde{v}(\e i)=v(\e i-x_{\e})$. This function is admissible in the definition of $\gamma_{x_0}(\e, r_{\e})$. Hence
\begin{align*}
\gamma_{x_0}(\e,r_{\e})-2\pi\big|\log\tfrac{\e}{r_{\e}}\big|\leq \gamma(\e,r_{\e} - 4\e)-2\pi\big|\log\tfrac{\e}{r_{\e}}\big|+\frac{1}{2\e^2}\sum_{\substack{\nn\\\e i,\e j\in B_{r_{\e}+2\e}\\ \e i,\e j\notin B_{r_{\e}-5\e}}}\e^2 \left|\frac{\e i}{|\e i|}-\frac{\e j}{|\e j|}\right|^2.
\end{align*} 
The last sum can be bounded applying \eqref{eq:xmoxcontinuous} which leads to
\begin{equation*}
\frac{1}{2\e^2}\sum_{\substack{\nn\\\e i,\e j\in B_{r_{\e}+2\e}\\ \e i,\e j\notin B_{r_{\e}-5\e}}}\e^2 \left|\frac{\e i}{|\e i|}-\frac{\e j}{|\e j|}\right|^2\leq \frac{C}{r_{\e}^2} \sum_{\substack{\e i\e\ZZ^2\\\e i\in B_{r_{\e}+2\e}\\ \e i\notin B_{r_{\e}-5\e}}}\e^2\leq \frac{C}{r_{\e}^2}|B_{r_{\e}+4\e}\setminus B_{r_{\e}-7\e}|\leq C\frac{r_{\e}\e+\e^2}{r_{\e}^2}\,,
\end{equation*}
which vanishes due to the assumption that $\e\ll r_{\e}$. Thus we proved that
\begin{equation*}
\limsup_{\e \to 0}\left(\gamma_{x_0}(\e,r_{\e})-2\pi\big|\log\tfrac{\e}{r_{\e}}\big| \right) \leq\gamma\,.
\end{equation*}
The reverse inequality for the $\liminf$ can be proven by a similar argument.
\end{proof}

\subsection{Compactness and $\Gamma$-convergence}
We recall the compactness result and the $\Gamma$-liminf inequality obtained in~\cite[Theorem~4.2]{Ali-DL-Gar-Pon}.
We emphasize that these results also hold in our setting, regarding $u_\e \in \PC_\e(\S_\e)$ as a spin field $u_\e \colon \e \ZZ^2 \to \SS^1$, that means, neglecting the $\S_\e$ constraint. 

\begin{theorem}[Theorem~4.2 in \cite{Ali-DL-Gar-Pon}]\label{thm:alidlgarpon} The following results hold:
	\begin{enumerate}
		\item[i)] (Compactness) Let $M \in \NN$ and let $u_\e \colon \e \ZZ^2 \cap \Omega \to \SS^1$ be such that $\frac{1}{\e^2} E_\e(u_\e) - 2 \pi M |\log \e| \leq C$. 
		Then there exists a subsequence (which we do not relabel) such that $\mu_{u_\e} \flat \mu$ for some $\mu = \sum_{h=1}^{M'} d_h \delta_{x_h}$ with $|\mu|(\Omega) \leq M$. Moreover, if $|\mu|(\Omega) = M$, then $|d_h|=1$.
		\item[ii)] ($\Gamma$-liminf inequality) Let $u_\e \colon \e \ZZ^2 \cap \Omega \to \SS^1$ be such that $\mu_{u_\e} \flat \mu$ with $\mu = \sum_{h=1}^M d_h \delta_{x_h}$, $|d_h|=1$. Then
	\end{enumerate}
	\begin{equation*}
		\liminf_{\e \to 0} \Big[ \frac{1}{\e^2} E_\e(u_\e) - 2 \pi M |\log \e| \Big] \geq \WW(\mu) + M \gamma \, .
	\end{equation*}
		
\end{theorem}

For the construction of the recovery sequence our arguments slightly differ from the proof of~\cite[Theorem~4.2]{Ali-DL-Gar-Pon}. For the reader's convenience we give here the detailed proof, which together with Theorem~\ref{thm:alidlgarpon} establishes Theorem~\ref{thm:theta smaller eps}.

\begin{proposition}[$\Gamma$-limsup inequality]
	Let $\mu = \sum_{h=1}^M d_h \delta_{x_h}$ with $|d_h| = 1$. Then there exists a sequence $u_\e \colon \e \ZZ^2 \cap \Omega \to \S_\e$ with $\mu_{u_\e} \flat \mu$ such that 
	\begin{equation} \label{eq:recovery in theta less eps}
		\limsup_{\e \to 0} \Big[ \frac{1}{\e^2}E_\e(u_\e) - 2 \pi M |\log \e|\Big] \leq \WW(\mu) + M \gamma \, .
	\end{equation}
\end{proposition}

\begin{proof}
To avoid confusion among infinitesimal sequences, in this proof we denote by $\varrho_\e$ a sequence, which may change from line to line, such that $\varrho_\e \to 0$ when $\e \to 0$.

\medskip
\ul{Step 1} (Construction of the recovery sequence)

\noindent Let us fix $0 < \eta' < \eta < 1$ with $\eta$ small enough such that the balls $\ol B_\eta(x_h)$ are pairwise disjoint and their union is contained in $\Omega$. We denote by~$w^\eta$ a solution to the minimum problem~\eqref{eq:min for renormalized energy} in~$\Omega_\mu^\eta$. Then for $h=1,\dots,M$ there exists $\alpha^\eta_h \in \CC$ with $|\alpha^\eta_h| = 1$ such that $w^\eta(x) = \alpha^\eta_h \odot \big(\tfrac{x - x_h}{|x - x_h|}\big)^{d_h}$ for $x \in \de B_\eta(x_h)$. Extend $w^{\eta}$ to $\Omega_\mu^{\eta'/2}$ by $w^\eta(x) := \alpha^\eta_h \odot \big(\tfrac{x - x_h}{|x - x_h|}\big)^{d_h}$ for $x \in B_{\eta}(x_h)\setminus \ol{B}_{\eta'/2}(x_h)$. To reduce notation, we set $A^\eta_{\eta'}(x_h) := B_\eta(x_h) \sm \ol B_{\eta'}(x_h)$. 
The extension $w^{\eta}$ then belongs to $W^{1,2}(\Omega_{\mu}^{\eta'};\SS^1)$ and its Dirichlet energy is given by
\begin{equation} \label{eq:2311181024}
	\integral{\Omega_\mu^{\eta'}}{|\nabla w^{\eta}|^2}{\d x} = \integral{\Omega_\mu^{\eta}}{|\nabla w^{\eta}|^2}{\d x} + \sum_{h=1}^M \integral{A^\eta_{\eta'}(x_h)}{\frac{1}{|x-x_h|^2}}{\d x}
	= \tilde m(\eta,\mu) + 2 \pi M \log \tfrac{\eta}{\eta'}\, .
\end{equation} 
Set $r_\e = |\log \e|^{-\frac{1}{2}} \gg \e$ and let $\tilde u_\e \colon \e \ZZ^2 \cap B_{r_\e}(x_h) \to \SS^1$ be a function that agrees with $x\mapsto\alpha^\eta_h \odot \big(\tfrac{x - x_h}{|x - x_h|}\big)^{d_h}$ on $\de_\e B_{r_\e}(x_h)$ and such that, cf. Lemmata~\ref{lemma:core energy eps dependent} and~\ref{lemma:nonzerodcore},
\begin{equation} \label{eq:energy is gamma}
	\frac{1}{\e^2} E_\e(\tilde u_\e; B_{r_\e}(x_h)) = \gamma_{x_h}(\e,r_\e) = 2 \pi \log \tfrac{r_\e}{\e} + \gamma + \varrho_\e  \leq 2 \pi |\log \e | + \gamma + \varrho_\e  \, .
\end{equation}
We now extend $\tilde u_\e$ to $\e \ZZ^2 \cap \Omega$ distinguishing two cases: we set 
\begin{equation} \label{eq:def of ol u in annulus}
	\tilde u_\e (\e i) := \alpha^\eta_h \odot \big(\tfrac{\e i - x_h}{|\e i - x_h|}\big)^{d_h}  \quad \text{if } \e i \in \ol B_{\eta'}(x_h) \sm B_{r_\e}(x_h) \, ;
\end{equation}
on $\e\ZZ^2\cap \Omega_{\mu}^{\eta'}$ the definition is more involved since $w^{\eta}$ has only Sobolev regularity up the (Lipschitz)-boundary and we are not aware of any density results preserving the traces on part of the boundary and the $\SS^1$-constraint. First we need to extend $w^{\eta}$ to $\tilde{\Omega}_{\mu}^{\eta}$ for some open set $\tilde{\Omega}\supset\Omega$ with Lipschitz boundary. This can be achieved via a local reflection as in \eqref{eq:extensionbyreflection}, so that we may assume from now on that $w^{\eta}\in W^{1,2}(\tilde\Omega_{\mu}^{\eta'};\SS^1)$. We further extend it to $\RR^2$ with compact support (neglecting the $\SS^1$ constraint outside $\tilde \Omega_{\mu}^{\eta'/2}$). Now let us define the discrete approximation of this extended $w^{\eta}$. Consider the shifted lattice $\ZZ_{\e}^{x}=x+\e\ZZ^2$ with $x\in B_{\e}$ and denote by $\widehat{w}_{\e,x}^{\eta}\in W^{1,2}(\RR^2;\RR^2)$ the piecewise affine interpolation of (the quasicontinuous representative of) $w^{\eta}$ on a standard triangulation associated to $\ZZ_{\e}^{x}$. As shown in the proof of \cite[Theorem 1]{vSc} there exists $x_{\e}\in B_{\e}$ such that $\widehat{w}_{\e,x_{\e}}^{\eta}\to w^{\eta}$ strongly in $W^{1,2}(\RR^2;\RR^2)$ (the proof is given in the scalar-case, but the argument also works component-wise; see also \cite[Section 3.1]{vSc}). Thus it is natural to define $\tilde{u}_{\e} \colon \e\ZZ^2\cap \Omega_{\mu}^{\eta'}\to\SS^1$ by 
\begin{equation}\label{eq:shiftaffine}
\tilde{u}_{\e}(\e i)=w^{\eta}_{\e,x_{\e}}(\e i+x_{\e})\,.
\end{equation}
Observe that since $w^{\eta}$ is defined on $\tilde{\Omega}_{\mu}^{\eta'/2}$ with values in $\SS^1$, for $\e$ small enough $\tilde{u}_{\e}$ is indeed $\SS^1$-valued. Moreover, the strong convergence of the affine interpolations ensures that
\begin{equation}\label{eq:finitedifferenceconv}
\frac{1}{2}\sum_{\substack{\nn\\ \e i,\e j\in\Omega_{\mu}^{\eta'}}}\e^2\left|\frac{\tilde{u}_{\e}(\e i)-\tilde{u}_{\e}(\e j)}{\e}\right|^2\leq \sum_{\substack{T \text{ triangle}\\ \e T\cap\Omega_{\mu}^{\eta'}\neq\emptyset }}\integral{\e T+x_{\e}}{|\nabla \widehat{w}^{\eta}_{\e,x_{\e}}|^2}{\mathrm{d}x}\leq\integral{\Omega_{\mu}^{\eta'}}{|\nabla w^{\eta}|}{\mathrm{d}x}+\varrho_{\e}\,.
\end{equation}
Finally, we define the global sequence $u_\e := \mathfrak{P}_\e(\tilde u_\e) \colon \e \ZZ^2 \cap \Omega \to \S_\e$ with the function $\mathfrak{P}_{\e}$ given by \eqref{eq:defproj}. Note that the piecewise constant maps $u_\e$ and $\tilde u_\e$ actually depend on $\eta'$ and $\eta$. For the following computations however we drop the dependence on these parameters to simplify notation.
	
We start estimating the error in energy due to the projection $\mathfrak{P}_\e$. To this end, we use the elementary inequality
\begin{equation}\label{eq:elementary}
|a|^2-|b|^2\leq |a-b|(|a|+|b|)\leq 2|b||a-b|+|a-b|^2,
\end{equation}
which yields that 
\begin{equation} \label{eq:error of projection}
\frac{1}{\e^2}E_\e(u_\e)  = \frac{1}{2}  \sum_{\substack{\langle i, j \rangle}} \big|u_\e(\e i) - u_\e(\e j) \big|^2 
\leq \frac{1}{\e^2}E_\e(\tilde u_\e ) + C |\Omega| \frac{\theta_\e^2}{\e^2} + 2  \sum_{\substack{\langle i, j \rangle }} \theta_\e \big|\tilde u_\e(\e i) - \tilde u_\e(\e j) \big| \, .
\end{equation}
We shall prove that $\frac{1}{\e^2}E_\e(\tilde u_\e )$ carries the whole energy, that means,
\begin{equation} \label{eq:discretization gives renormalized}
\limsup_{\e \to 0} \Big[ \frac{1}{\e^2}E_\e(\tilde u_\e ) - 2 \pi M |\log \e| \Big] \leq \WW(\mu) + M \gamma  \, ,
\end{equation}
whereas the remainder satisfies
\begin{equation} \label{eq:rest is vanishing}
\lim_{\e \to 0} \sum_{\substack{\langle i, j \rangle }} \theta_\e \big|\tilde u_\e(\e i) - \tilde u_\e(\e j) \big| = 0 \, .
\end{equation}
Inequalities~\eqref{eq:error of projection}, \eqref{eq:discretization gives renormalized}, and~\eqref{eq:rest is vanishing} then yield~\eqref{eq:recovery in theta less eps} thanks to the assumption~\eqref{eq:theta less eps}.

In order to prove both~\eqref{eq:discretization gives renormalized} and~\eqref{eq:rest is vanishing}, we estimate separately the contribution of the energy and that of the remainder in the regions $B_{ r_\e}(x_h)$, $\Omega^{\eta'}_\mu$, and $(B_{\eta'+\e}(x_h) \sm B_{r_\e-\e}(x_h) )$. We remark that this decomposition of $\e \ZZ^2 \cap \Omega$ takes into account all the nearest-neighbors interactions.
 
\medskip
\ul{Step 2} (Estimates close to the singularities)

\noindent Let us start with the estimates inside $B_{ r_\e}(x_h)$. 
Notice that~\eqref{eq:energy is gamma} already gives explicitly the value of $\frac{1}{\e^2}E_\e(\tilde u_\e; B_{r_\e})$, so we only have to estimate the remainder in $B_{ r_\e}(x_h)$.
Combining the Cauchy-Schwarz inequality with~\eqref{eq:energy is gamma}, ~\eqref{eq:theta less eps}, and taking into account that $r_\e = |\log \e|^{-\frac{1}{2}}$, we obtain that 
\begin{equation} \label{eq:2211182222}
\begin{split}
	\sum_{\substack{\langle i, j \rangle \\ \e i, \e j \in B_{r_\e}(x_h)}} \hspace{-1.3em} \theta_\e \big|\tilde u_\e(\e i) -  \tilde u_\e(\e j) \big| & \leq \bigg( \sum_{\substack{\langle i, j \rangle \\ \e i, \e j \in B_{r_\e}(x_h)}}  \hspace{-1.2em} \theta_\e^2  \bigg)^{\!\!\frac{1}{2}} \bigg( \sum_{\substack{\langle i, j \rangle \\ \e i, \e j \in B_{r_\e}(x_h)}} \hspace{-1.2em} \big| \tilde u_\e(\e i) - \tilde u_\e(\e j) \big|^2 \bigg)^{\!\!\frac{1}{2}} \\
	& \leq C r_\e \frac{\theta_\e}{\e} \Big( \frac{1}{\e^2}E_\e(\tilde u_\e; B_{r_\e}(x_h)) \Big)^{\! \frac{1}{2}}  \\
	& \leq  C r_\e \frac{\theta_\e}{\e} \big( 2\pi |\log \e |   + \gamma + \varrho_\e \big)^{\! \frac{1}{2}} \to 0  \quad \text{as } \e \to 0\, .
\end{split}
\end{equation}

\ul{Step 3} (Estimates in the perforated domain)
	
\noindent We go on with the estimates inside $\Omega^{\eta'}_\mu$. In this set the function $\tilde{u}_{\e}$ is given by \eqref{eq:shiftaffine}. 
In particular, by~\eqref{eq:2311181024} and~\eqref{eq:finitedifferenceconv},
\begin{equation} \label{eq:2511181755}
\frac{1}{\e^2}E_\e(\tilde u_\e;\Omega^{\eta'}_\mu) \leq \integral{\Omega^{\eta'}_\mu}{|\nabla w^{\eta}|^2}{\d x} + \varrho_\e = \tilde m(\eta,\mu) + 2 \pi M \log \tfrac{\eta}{\eta'}+ \varrho_\e \, .
\end{equation} 
Concerning the remainder, the Cauchy-Schwarz inequality, \eqref{eq:2511181755}, and~\eqref{eq:theta less eps} imply that 
\begin{equation} \label{eq:2711181519}
	\begin{split}
	\sum_{\substack{\langle i, j \rangle \\ \e i, \e j \in \Omega^{\eta'}_\mu}} \hspace{-0.5em} \theta_\e \big|\tilde u_\e(\e i) -  \tilde u_\e(\e j) \big| & \leq \bigg( \sum_{\substack{\langle i, j \rangle \\ \e i, \e j \in \Omega^{\eta'}_\mu}}  \hspace{-0.5em} \theta_\e^2  \bigg)^{\!\!\frac{1}{2}} \bigg( \sum_{\substack{\langle i, j \rangle \\ \e i, \e j \in \Omega^{\eta'}_\mu}} \hspace{-0.5em} \big| \tilde u_\e(\e i) - \tilde u_\e(\e j) \big|^2 \bigg)^{\!\!\frac{1}{2}} \\
	& \leq C |\Omega|^\frac{1}{2} \frac{\theta_\e}{\e} \Big( \frac{1}{\e^2}E_\e(\tilde u_\e; \Omega^{\eta'}_\mu) \Big)^{\! \frac{1}{2}}  \\
	& \leq  C |\Omega|^\frac{1}{2} \frac{\theta_\e}{\e} \big( \tilde m(\eta,\mu) + 2\pi M \log\tfrac{\eta}{\eta'} + \varrho_\e \big)^{\! \frac{1}{2}} \to 0 \quad \text{as } \e \to 0\,.
\end{split}
\end{equation}

\ul{Step 4} (Estimates in the annulus)

\noindent Finally, we need to estimate the energy $\frac{1}{\e^2}E_\e$ in the set $( B_{\eta'+\e}(x_h) \sm  B_{r_\e-\e}(x_h) )$, where $\tilde u_\e(x) = \alpha^\eta_h \odot \big(\tfrac{x-x_h}{|x-x_h|}\big)^{d_h}$ (or slightly shifted in $B_{\eta'+\e}(x_h)\setminus \ol B_{\eta'}(x_h)$ due to \eqref{eq:shiftaffine}, for which we recall that $w^{\eta}=\alpha^\eta_h \odot \big(\tfrac{x-x_h}{|x-x_h|}\big)^{d_h}$ in $B_{\eta}(x_h)\setminus B_{\eta'/2}(x_h)$). To simplify notation, for $R>r>0$ we denote in this step $A^{R}_{r}(x):=B_{R}(x)\setminus \ol{B}_{r}(x)$. 

We first estimate the contribution involving the shifted function. For any $\e i,\e j$ with $|i-j|=1$ and $\e i\in A^{\eta'+\e}_{\eta'}(x_h)$ the condition $|x_{\e}|\leq\e$ and \eqref{eq:xmoxcontinuous} imply that
\begin{equation*}
|\tilde{u}_{\e}(\e i)-\tilde{u}_{\e}(\e j)|\leq \frac{4 \e}{\eta' -\e}\,.
\end{equation*}
Summing these estimate we deduce that
\begin{equation*}
\frac{1}{\e^2}\hspace{-1em}\sum_{\substack{\nn\\ \e i\in A^{\eta'+\e}_{\eta'}(x_h)}}\hspace{-1em}\e^2|\tilde{u}_{\e}(\e i)-\tilde{u}_{\e}(\e j)|^2\leq C\frac{(\eta'+3\e)^2-(\eta'-2\e)^2}{(\eta' -\e)^2}\leq C\frac{\e\eta'+\e^2}{(\eta' -\e)^2}\to 0\quad\text{ as }\e\to 0\,.
\end{equation*}
Hence we can write 
\begin{align}\label{eq:neglectshift}
\frac{1}{\e^2}E_{\e}(\tilde{u}_{\e}; A^{\eta'+\e}_{r_{\e}-\e}(x_h))\leq& \frac{1}{\e^2}E_{\e}(\tilde{u}_{\e};\ol{A}^{\eta'}_{r_{\e}-\e}(x_h))+\varrho_{\e}\,.
\end{align} 
Note that on the set $\e\ZZ^2\cap\ol{A}^{\eta'}_{r_{\e}-\e}(x_h)$ the function $\tilde{u}_{\e}$ coincides with $x\mapsto \alpha^\eta_h \odot \big(\tfrac{x-x_h}{|x-x_h|}\big)^{d_h}$, so that the invariance of the discrete energy under orthogonal transformations implies that
\begin{equation*}
\frac{1}{\e^2}E_{\e}(\tilde{u}_{\e};\ol{A}^{\eta'}_{r_{\e}-\e}(x_h))=\frac{1}{\e^2}\hspace{-1em}\sum_{\substack{\nn\\ \e i,\e j\in \ol{A}^{\eta'}_{r_{\e}-\e}(x_h)}}\hspace{-1em}\e^2\left|\frac{\e i-x_h}{|\e i-x_h|}-\frac{\e j-x_h}{|\e j-x_h|}\right|^2.
\end{equation*}
Using a shifted version of the trimmed quadrants defined in~\eqref{eq:trimmed quadrants} and summing over all possible pairs of signs $s \in \{(+,+),(-,+),(+,-),(-,-)\}$, we can split the energy as 
\begin{equation}\label{eq:2711181115}
	\begin{split}
	\frac{1}{\e^2} E_\e(\tilde u_\e; \ol{A}^{\eta'}_{r_\e-\e}(x_h)) \leq \sum_{s}\frac{1}{\e^2} E_\e(\tilde u_\e; (\e \mathcal{Q}^s_3+x_h) \cap \ol{A}^{\eta'}_{r_\e-\e}(x_h) ) + C \sum_{k = \big\lfloor\tfrac{r_{\e}}{\e}\big\rfloor -2 }^{\infty} \frac{1}{k^2} \, ,
	\end{split}
\end{equation}
where we used the bound \eqref{eq:xmoxcontinuous} to estimate the contributions not fully contained in one of the trimmed quadrants by the last sum. Since the sum $\sum_{k = 1}^{+\infty} k^{-2}$ is finite and $\frac{r_\e}{\e} \to +\infty$, the second term in the right-hand side is infinitesimal as $\e \to 0$. On each trimmed quadrant we use a shifted version of~\eqref{eq:estimatecases} and a monotonicity argument as in \eqref{eq:trickytrimmer} to deduce that 
\begin{equation*}
	\frac{1}{\e^2} E_\e(\tilde u_\e; (\e \mathcal{Q}^s_3+x_h) \cap \ol{A}^{\eta'}_{r_\e-\e}(x_h) ) \leq \sum_{\substack{ \e i \in \ol{A}^{\eta'}_{r_\e-\e}(x_h)  \\ \e i \in \e\ZZ^2 \cap (\mathcal{Q}^s_3+x_h) }}   \e^2 \frac{1}{|\e i-x_h|^2} \leq \integral{\e \mathcal{Q}^s_2  \cap A^{\eta'}_{r_\e-3\e}}{\frac{1}{|x|^2}}{\d x} \, ,
\end{equation*}
with the annulus $A^{\eta'}_{r_\e-3\e} = B_{\eta'} \sm \ol B_{r_\e - 3\e}$ centered at 0. Summing over all four quadrants, since $\e \ll r_\e$, we get that
\begin{equation*} 
	\sum_s \frac{1}{\e^2} E_\e(\tilde u_\e; (\e \mathcal{Q}^s_3+x_h) \cap \ol{A}^{\eta'}_{r_\e-\e}(x_h) ) \leq \integral{A^{\eta'}_{r_\e - 3\e}}{\frac{1}{|x|^2}}{\d x} = 2\pi \log \tfrac{\eta'}{r_\e -3 \e} = 2\pi \log \tfrac{\eta'}{r_\e} + \varrho_\e \, .
\end{equation*} 
In combination with~\eqref{eq:neglectshift} and~\eqref{eq:2711181115} we conclude that
\begin{equation}\label{eq:2711181451}
\frac{1}{\e^2}E_{\e}(\tilde{u}_{\e}; A^{\eta'+\e}_{r_{\e}-\e}(x_h))\leq 2\pi \log \tfrac{\eta'}{r_\e} + \varrho_\e\,. 
\end{equation}
We now estimate the remainder term in $A^{\eta'+\e}_{r_{\e}-\e}(x_h)$, for which applying the Cauchy-Schwarz inequality as in~\eqref{eq:2211182222} is too rough.  However, note that for any $i\in\ZZ^2$ and $x\in \e i+[0,\e)^2$ with $|\e i-x_h|\gg \e$, we have for $\e$ small enough that
\begin{equation*}
|\e i-x_h|\geq |x-x_h|-\sqrt{2}\e\geq \frac{1}{2}|x-x_h|\,.
\end{equation*}
Hence, using \eqref{eq:xmoxcontinuous} and a change of variables we obtain that 
\begin{align} \label{eq:2711181209}
\sum_{\substack{\langle i, j \rangle \\ \e i, \e j \in A^{\eta'+\e}_{r_\e-\e}(x_h) }} \hspace{-1em}\theta_\e \big|\tilde u_\e(\e i) - \tilde u_\e(\e j) \big|& \leq C\hspace{-1em}\sum_{\e i\in\e\ZZ^2\cap A^{\eta'+\e}_{r_{\e-\e}}(x_h)}\hspace{-1em}\theta_{\e}\frac{\e}{|\e i-x_h|}\leq C\frac{\theta_{\e}}{\e}\integral{A_{r_{\e}-3\e}^{\eta'+3\e}}{\frac{1}{|x|}}{\mathrm{d}x}  \, .
\end{align}
Since the last integral is proportional to $\eta'$, inserting assumption~\eqref{eq:theta less eps} shows that the right-hand side of~\eqref{eq:2711181209} is infinitesimal as $\e \to 0$.

\medskip  

\ul{Step 5} (Proof of~\eqref{eq:discretization gives renormalized} and~\eqref{eq:rest is vanishing} and conclusion)
To prove~\eqref{eq:discretization gives renormalized}, we employ~\eqref{eq:energy is gamma}, \eqref{eq:2511181755},  and~\eqref{eq:2711181451} to split the energy as follows
\begin{equation*}
	\begin{split}
	\frac{1}{\e^2} E_\e(\tilde u_\e) - 2 \pi M |\log \e|  & \leq \frac{1}{\e^2} E_\e(\tilde u_\e; \Omega^{\eta'}_\mu) - 2 \pi M \log \tfrac{\eta}{\eta'} +2\pi M\log \eta  
	\\
	& \quad + \sum_{h=1}^M \Big[ \frac{1}{\e^2} E_\e(\tilde u_\e; B_{r_\e}(x_h)) - 2 \pi  \log \tfrac{r_\e}{\e} \Big]\\
	& \quad + \sum_{h=1}^M \Big[ \frac{1}{\e^2} E_\e(\tilde u_\e; A^{\eta'+\e}_{r_\e-\e}(x_h)) - 2 \pi  \log \tfrac{\eta'}{r_\e} \Big] 
	\\
	& \leq \tilde m(\eta,\mu) + 2 \pi M \log \eta  + M \gamma + \varrho_\e \, . 
	\end{split}
\end{equation*}
Now we stress the dependence of $\tilde u_\e$ on $\eta'$ and $\eta$, denoting the sequence by~$\tilde u_{\e,\eta}$ (set for instance $\eta'=\eta/2$). Letting $\e \to 0$, for $\eta<1$ we deduce that
\begin{equation*}
	\limsup_{\e \to 0} \Big[ \frac{1}{\e^2} E_\e(\tilde u_{\e, \eta}) - 2 \pi M |\log \e| \Big] \leq \tilde m(\eta,\mu) - 2 \pi M |\log \eta |  + M \gamma \, .
\end{equation*}
Moreover, \eqref{eq:rest is vanishing} follows from~\eqref{eq:2211182222}, \eqref{eq:2711181519}, and \eqref{eq:2711181209} splitting the remainder in the same way. Hence, for each $\eta<1$ we found a sequence $u_{\e,\eta}\in\mathcal{PC}_{\e}(\S_{\e})$ such that
\begin{equation}\label{eq:almostthere}
\limsup_{\e \to 0} \Big[ \frac{1}{\e^2} E_\e(u_{\e, \eta}) - 2 \pi M |\log \e| \Big] \leq \tilde m(\eta,\mu) - 2 \pi M |\log \eta |  + M \gamma \, .
\end{equation}
Before we conclude via a diagonal argument, we have to identify the flat limit of the vorticity measure $\mu_{u_{\e,\eta}}$. From the above energy estimate and Proposition \ref{prop:XY classical} we deduce that, passing to a subsequence, $\mu_{u_{\e,\eta}}\flat\mu_{\eta}$ for some $\mu_{\eta}=\sum_{k=1}^{M} d^{\eta}_k\delta_{x^{\eta}_k}$ with $|\mu_{\eta}|(\Omega)\leq M$ (we allow for $d_k^{\eta}=0$ to sum up to $M$). Fix a Lipschitz set $A=A_{\eta}\subset\subset \Omega$ such that ${\rm supp}(\mu_{\eta})\subset A$. Due the logarithmic energy bound we can apply \cite[Proposition 5.2]{Ali-Cic-Pon} and deduce that on the set $A$ it holds that $\pi^{-1}\mathrm{J}\widehat{u}_{\e,\eta}-\mu_{u_{\e,\eta}}\flat 0$, where $\widehat{u}_{\e,\eta}$ denotes the piecewise affine interpolation associated to $u_{\e,\eta}$ (which is at least defined on $A$ for $\e$ small enough). Let now $\widehat{v}_{\e,\eta}$ be the function defined via piecewise affine interpolation of the values~$\tilde u_{\e,\eta}(\e i)$, $\e i \in \e\ZZ^2\cap\Omega$. We argue that on $A$ the Jacobian of $\widehat{u}_{\e,\eta}$ is close with respect to the flat convergence to the Jacobian of $\widehat{v}_{\e,\eta}$, i.e., 
\begin{equation} \label{eq:Jacobianafterprojection}
	\mathrm{J}(\widehat{u}_{\e,\eta}) - \mathrm{J}(\widehat{v}_{\e,\eta}) \flat 0\quad\text{ on }A \, .
\end{equation}
To this end, we apply \cite[Lemma 3.1]{Ali-Cic-Pon} which states that the Jacobians of two functions $u$ and $w$ are close if $\|u-w\|_{L^2}(\|\nabla u\|_{L^2}+\|\nabla w\|_{L^2})$ is small. 
Since by definition $u_{\e,\eta}=\mathfrak{P}_{\e}(\widetilde{u}_{\e,\eta})$, we know that
\begin{equation*}
|u_{\e,\eta}(\e i)-\widetilde{u}_{\e,\eta}(\e i)|\leq\theta_{\e}.
\end{equation*}
Inserting this estimate in the definition of the piecewise affine interpolation one can show that
\begin{equation}
|\widehat{u}_{\e,\eta}(x)-\widehat{v}_{\e,\eta}(x)|\leq C\theta_{\e}\quad\text{for all }x\in A\,.
\end{equation}
Taking into account the energy bounds~\eqref{eq:almostthere} and~\eqref{eq:discretization gives renormalized}, we conclude that
\begin{align}\label{eq:diffpiecewiseaffine}
	&\|\widehat{u}_{\e,\eta}-\widehat{v}_{\e,\eta}\|_{L^2(A)}\left(\|\nabla\widehat{u}_{\e,\eta}\|_{L^2(A)}+\|\nabla\widehat{v}_{\e,\eta}\|_{L^2(A)}\right)\nonumber
	\\
	&\quad  \leq \, C\sqrt{|A|}\theta_{\e}\left(\frac{1}{\e^2}E_{\e}(u_{\e,\eta};\Omega)+\frac{1}{\e^2}E_{\e}(\widetilde{u}_{\e,\eta};\Omega)\right)^{\! \frac{1}{2}}\leq C \sqrt{|A|}\theta_{\e}|\log \e|^{\frac{1}{2}}\,.
\end{align}
The above right hand side vanishes when $\e\to 0$, so that \cite[Lemma 3.1]{Ali-Cic-Pon} implies \eqref{eq:Jacobianafterprojection}. Hence it suffices to study the limit of the Jacobians of $\widehat{v}_{\e,\eta}$. We show that the limit carries mass in each ball $B_{\rho}(x_h)$ for all $\rho>0$ small enough such that $B_{\rho}(x_h)\subset\subset A$. To this end, we test the flat convergence against the Lipschitz cut-off function 
\begin{equation*}
\varphi^h_{\rho}(x) := \min\{\max\{\tfrac{2}{\rho} (\rho - |x-x_h|),0\},1\} \, . 
\end{equation*}
Using the distributional divergence form of the Jacobian and the fact that $\widehat{v}_{\e,\eta}$ agrees with the piecewise affine interpolation of the map $x\mapsto \alpha^\eta_h \odot \big(\tfrac{x-x_h}{|x-x_h|}\big)^{d_h}$ on the support of the gradient of $\varphi_{\rho}^h$ provided $\rho<\eta/4$ and $r_{\e}\ll \rho/2$, we infer that
\begin{align*}
\langle\mu_{\eta},\varphi_{\rho}^h\rangle&=\lim_{\e \to 0}\langle \pi^{-1}\mathrm{J}\widehat{v}_{\e,\eta},\varphi^h_{\rho}\rangle 
\\
&=-\lim_{\e \to 0}\frac{1}{2\pi}\integral{A^{\rho}_{\rho/2}(x_h)}{\begin{pmatrix}(\widehat{v}_{\e,\eta})_1\partial_2 (\widehat{v}_{\e,\eta})_2-(\widehat{v}_{\e,\eta})_2\partial_2 (\widehat{v}_{\e,\eta})_1 \\-(\widehat{v}_{\e,\eta})_1\partial_1 (\widehat{v}_{\e,\eta})_2+(\widehat{v}_{\e,\eta})_2\partial_1 (\widehat{v}_{\e,\eta})_1\end{pmatrix}\nabla\varphi_{\rho}^h}{\,\mathrm{d}x}=d_h\,,
\end{align*}
where the limit can be calculated similarly to \eqref{eq:explicitdegree}. From this equality and the arbitrariness of $\rho>0$, we deduce that $\{x_1,\dots,x_M\}\subset {\rm supp}(\mu_{\eta})$ and $\mu\mres\{x_h\}=d_h\delta_{x_h}$. Since $|\mu_{\eta}|(\Omega)\leq M$, it follows that $\mu_{\eta}=\mu$ independently of $\eta$ and the subsequence of $\e$. Since the flat convergence is given by a metric, we can thus use a diagonal argument with $\eta=\eta_{\e}$ to find a sequence $u_{\e}:=u_{\e,\eta_{\e}}$ satisfying $\mu_{u_\e}\flat\mu$ and, due to \eqref{eq:almostthere}, also the claimed inequality \eqref{eq:recovery in theta less eps}.
\end{proof}

\noindent {\bf Acknowledgments.}  The work of M.\ Cicalese was supported by the DFG Collaborative Research Center TRR 109, “Discretization in Geometry and Dynamics”. G.\ Orlando has been supported by the Alexander von Humboldt Foundation and the European Union's Horizon 2020 research and innovation programme under the Marie Sk\l odowska-Curie grant agreement No 792583. M.\ Ruf acknowledges financial support from the European Research Council under
the European Community's Seventh Framework Program (FP7/2014-2019 Grant Agreement QUANTHOM 335410).


\begin{thebibliography}{99} \frenchspacing


\bibitem{Alb-Bal-Orl} {\sc G. Alberti, S. Baldo, G. Orlandi.} Variational convergence for functionals of Ginzburg-Landau type. {\em Indiana Univ. Math. J.} {\bf 54} (2005), 1411--1472.
%
%	
\bibitem{Ali-Bra-Cic} {\sc R. Alicandro, A. Braides, M. Cicalese}. Phase and anti-phase boundaries in binary discrete systems: a variational viewpoint. {\em Netw. Heterog. Media} {\bf 1} (2006), 85--107.	
\bibitem{Ali-Bra-Cic-DL-Pia} {\sc R. Alicandro, A. Braides, M. Cicalese, L. De Luca, A. Piatnitski} Topological singularities in periodic media: Ginzburg-Landau and core-radius approaches. {\em In preparation.}

%	
\bibitem{Ali-Cic} {\sc R. Alicandro, M. Cicalese}. Variational analysis of the asymptotics of the XY model. {\em Arch. Ration. Mech. Anal.} {\bf 192} (2009), 501--536.

\bibitem{Ali-Cic-Ruf} {\sc R. Alicandro, M. Cicalese, M. Ruf}. Domain formation in magnetic polymer composites: an approach via stochastic homogenization. {\em Arch. Ration. Mech. Anal.} {\bf 218} (2015), 945--984.

\bibitem{Ali-Cic-Pon} {\sc R. Alicandro, M. Cicalese, M. Ponsiglione}. Variational equivalence between Ginzburg-Landau, XY spin systems and screw dislocations energies. {\em Indiana Univ. Math. J.} {\bf 60} (2011), 171--208.
%

\bibitem{Ali-DL-Gar-Pon} {\sc R. Alicandro, L. De Luca, A. Garroni, M. Ponsiglione.} Metastability and dynamics of discrete topological singularities in two dimensions: a $\Gamma$-convergence approach. {\em Arch. Ration. Mech. Anal.} {\bf 214} (2014), 269--330.
%
\bibitem{Ali-Gel} {\sc R. Alicandro, M. S. Gelli.} Local and nonlocal continuum limits of Ising-type energies for spin systems. {\em SIAM J. Math. Anal.} {\bf 48} (2016), 895--931.
%
\bibitem{Ali-Pon} {\sc R. Alicandro, M. Ponsiglione.} Ginzburg-Landau functionals and renormalized energy: a revised $\Gamma$-convergence approach. {\em J. Funct. Anal.} {\bf 266} (2014), 4890--4907.
%
\bibitem{Amb} {\sc L. Ambrosio}. Metric space valued functions of bounded variation. {\em Ann. Sc. Norm. Super. Pisa Cl. Sci. (4)} {\bf 17} (1990), 439--478. 
%
\bibitem{Amb-Fus-Pal} {\sc L. Ambrosio, N. Fusco, D. Pallara}. {\em Functions of Bounded Variation and Free Discontinuity Problems}, Clarendon Press Oxford, 2000.
%
\bibitem{Bac-Cic-Kre-Orl-surf} {\sc A. Bach, M. Cicalese, L. Kreutz, G. Orlando} The antiferromagnetic XY model on the triangular lattice: chirality transitions at the surface scaling. {\em Preprint.} arXiv:2004.01416.

\bibitem{BacCicKreOrl} {\sc A. Bach, M. Cicalese, L. Kreutz, G. Orlando} The antiferromagnetic XY model on the triangular lattice: topological singularities. {\em Preprint.} arXiv:2011.10445.
%
\bibitem{Bad-Cic-DL-Pon} {\sc R. Badal, M. Cicalese, L. De Luca, M. Ponsiglione.} $\Gamma$-convergence analysis of a generalized $XY$ model: fractional vortices and string defects. {\em Comm. Math. Phys.} {\bf 358} (2018),  705--739.

%
\bibitem{Ber} {\sc V.L. Berezinskii.} Destruction of long range order in one-dimensional and two dimensional systems having a continuous symmetry group. I. Classical systems. {\em Sov. Phys. JETP} {\bf 32} (1971), 493--500.
%
\bibitem{Bet-Bre-Hel} {\sc F. Bethuel, H. Brezis, F. H\'elein.} {\em Ginzburg-Landau vortices.} Progress in Nonlinear Differential Equations and their Applications, 13. Birkh\"auser Boston MA, 1994.
%

\bibitem{Bra-Cic} {\sc A. Braides and M. Cicalese}. Interfaces, modulated phases and textures in lattice systems. {\em Arch. Ration. Mech. Anal.} {\bf 223} (2017), 977--1017.

\bibitem{Bra-Cic-Ruf} {\sc A. Braides, M. Cicalese, M. Ruf}. Continuum limit and stochastic homogenization of discrete ferromagnetic thin films. {\em Anal. PDE} {\bf 11} (2018), 499-553. 

\bibitem{Bra-Cic-Sol} {\sc A. Braides, M. Cicalese, F. Solombrino}. $Q$-tensor continuum energies as limits of head-to-tail symmetric spin systems. {\em SIAM J. Math. Anal.} {\bf 47} (2015), 2832--2867.

\bibitem{Bra-Kre}
{\sc A. Braides, L. Kreutz}. Design of lattice surface energies. {\em Calc. Var. Partial Differential Equations} {\bf 57}:97 (2018). 

\bibitem{Bra-Pia} {\sc A. Braides and A. Piatnitski}. Homogenization of surface and length energies for spin systems. {\em J.Funct. Anal.} {\bf 264} (2013), 1296--1328.

\bibitem{BreCorLie} {\sc H. Brezis, J.-M. Coron, E. Lieb.} Harmonic maps with defects. {\em Comm. Math. Phys.} {\bf 107} (1986), 649--705. 
%

\bibitem{Caf-DLL} {\sc L. A. Caffarelli, R. de la Llave.} Interfaces of ground states in Ising models with periodic coefficients. {\em J. Stat. Phys.} {\bf 118} (2005), 687--719.

\bibitem{Cic-For-Orl} {\sc M. Cicalese, M. Forster, G. Orlando.} Variational analysis of a two-dimensional frustrated spin system: emergence and rigidity of chirality transitions. {\em SIAM J. Math. Anal.} {\bf 51} (2019), 4848--4893.

\bibitem{CicOrlRuf} {\sc M. Cicalese, G. Orlando, M. Ruf.} Emergence of concentration effects in the variational analysis of the $N$-clock model. {\em Preprint.} arXiv:2005.13365.
%
\bibitem{CicOrlRuf3} {\sc M. Cicalese, G. Orlando, M. Ruf.} Coarse graining and large-$N$ behavior of the $d$-dimensional $N$-clock model. {\em Preprint.} arXiv:2004.02217.
%
%
\bibitem{Cic-Sol} {\sc M. Cicalese, F. Solombrino.} Frustrated ferromagnetic spin chains: a variational approach to chirality transitions. {\em J. Nonlinear Sci.} {\bf 25} (2015), 291--313.

\bibitem{Fed} {\sc H. Federer}. {\em Geometric measure theory}. (Grundlehren Math. Wiss. 153. Bd) Berlin Heidelberg New York: Springer 1969.
%
\bibitem{Fro-Spe} {\sc J. Fr\"ohlich, T. Spencer.} The Kosterlitz-Thouless transition in two-dimensional abelian spin systems and the Coulomb gas. {\em Comm. Math. Phys.} {\bf 81} (1981), 527--602.
%
\bibitem{Gia-Mod-Sou-I} {\sc M. Giaquinta, G. Modica, J. Sou\v{c}ek}. {\em Cartesian currents in the calculus of variations, I}. Ergebnisse Math. Grenzgebiete (III Ser), 37, Springer, Berlin (1998).
%
\bibitem{Gia-Mod-Sou-II} {\sc M. Giaquinta, G. Modica, J. Sou\v{c}ek}. {\em Cartesian currents in the calculus of variations, II}. Ergebnisse Math. Grenzgebiete (III Ser), 38, Springer, Berlin (1998).
%
\bibitem{Gia-Mod-Sou-S1} {\sc M. Giaquinta, G. Modica, J. Sou\v{c}ek}. Variational problems for maps of bounded variation with values in $\SS^1$. {\em Calc. Var.} {\bf 1} (1993), 87--121.

\bibitem{Guo-Ran-Jos} {\sc H. Guo, A. Rangarajan, S. Joshi}. Diffeomorphic Point Matching. In: {\sc N. Paragios, Y. Chen, O. Faugeras} (Eds.) {\em Handbook of Mathematical Models in Computer Vision} (2006), 205--219. Springer, Boston, MA.
%

\bibitem{Kos} {\sc J.M. Kosterlitz.} The critical properties of the two-dimensional xy model. {\em J. Phys. C} {\bf 6} (1973), 1046--1060.
%
\bibitem{Kos-Tho} {\sc J.M. Kosterlitz, D.J. Thouless.} Ordering, metastability and phase transitions in two-dimensional systems. {\em J. Phys. C} {\bf 6} (1973), 1181--1203.

\bibitem{Lic} {\sc M. W. Licht}. Smoothed projections over weakly Lipschitz domains. {\em Math. Comp.} {\bf 88} (2019), 179--210.
%
\bibitem{Luu-Vae} {\sc J. Luukkainen and J. V\"ais\"al\"a}. Elements of Lipschitz topology. {\em Ann. Acad. Sci. Fenn. Ser. A I Math.} {\bf 3} (1977), 85--122.

\bibitem{Pon} {\sc  M. Ponsiglione.} Elastic energy stored in a crystal induced by screw dislocations: from discrete to continuous. {\em SIAM J. Math. Anal.} {\bf 39} (2007), 449--469.

\bibitem{San-Ser-book} {\sc E. Sandier, S. Serfaty.} {\em Vortices in the Magnetic Ginzburg-Landau Model.} Progress in Nonlinear Differential Equations and their Applications, 70. Birkh\"auser Boston, Inc., Boston, MA, 2007.

\bibitem{vSc} {\sc J. van Schaftingen}. Approximation in Sobolev spaces by piecewise affine interpolation. {\em J. Math. Anal. Appl.} {\bf 420} (2014), 40--47.
\end{thebibliography}
\end{document}